\documentclass[11pt]{article}
\usepackage[margin=1in]{geometry}
\usepackage[utf8]{inputenc}
\usepackage{amsmath,amssymb,amsthm}
\usepackage{enumitem}
\usepackage{graphicx}
\usepackage{booktabs}
\usepackage{array}
\usepackage[svgnames]{xcolor}
\usepackage{colortbl}
\definecolor{abC}{HTML}{2A78D6}
\definecolor{naC}{HTML}{EB6834}
\usepackage{authblk}
\usepackage{hyperref}
\newcommand{\dburl}{\cite{Lee_2026_bicycle_race_database}}%
\usepackage[section]{placeins}%

\usepackage{tikz}
\usetikzlibrary{arrows.meta,positioning,calc}
\usepackage{subcaption}
\usepackage{float}%
\definecolor{leftq}{HTML}{2A78D6}
\definecolor{rightq}{HTML}{EB6834}
\definecolor{checkq}{HTML}{1BAF7A}
\graphicspath{{figures/}}

\newcommand{\Aut}{\operatorname{Aut}}

\newcommand{\ann}{\operatorname{ann}}
\newcommand{\rank}{\operatorname{rank}}
\theoremstyle{plain}
\newtheorem{theorem}{Theorem}[section]
\newtheorem{proposition}[theorem]{Proposition}

\newtheorem{corollary}[theorem]{Corollary}

\theoremstyle{definition}

\theoremstyle{remark}

\title{Bicycle Race: small two-block codes on the Pareto frontier}
\author{Gideon Lee\thanks{\href{mailto:gideonlee@uchicago.edu}{\texttt{gideonlee@uchicago.edu}}}}
\affil{Pritzker School of Molecular Engineering, University of Chicago}
\date{\today}

\begin{document}
\maketitle

\nocite{Lin_PRA_2024_2bga}
\begin{abstract}
Two-block group-algebra codes are among the simplest quantum LDPC codes, and many of them are
translation invariant on a twisted torus, which suits neutral-atom arrays and other platforms with movable qubits.  We ask how far they go at
small sizes when searched systematically across all groups. We use algebraic constraints and screens to design a search pipeline capable of enumerating broad and complete swathes of two-block codes on a laptop. With it we census about $10^8$ two-block codes with up to $320$ qubits at check weights five to nine. Several of these censuses are exhaustive, and support complete statements; for instance, the Gross code is, up to equivalence, the only abelian weight-six two-block code with $n\le200$ to reach $kd^2/n=12$, no weight-six two-block code with $n\le300$ beats $[[288,16,16]]$, and the known $[[64,18,8]]$ is the best two-block code by $kd^2/n$ across weights five to nine with $n\le94$. The censuses are also large enough to give insights into where good codes live -- they concentrate in a handful of groups; abelian groups give the best codes at small $n$ and non-abelian groups at large $n$, with the crossover appearing to move out as the check weight grows; and high quality codes become $100$ to $1000$ times rarer per draw from weight six to weights eight and nine.  Using the group algebra and the census data to steer a statistical search finds four to fifteen times as many good codes as uniform draws in equal time.  Among the codes found are $[[102,20,10]]$ and $[[120,24,10]]$ at check weight nine, which we simulate in noisy syndrome-extraction circuits, and $[[186,20,18]]$, a weight-nine code with every check in a $5\times5$ window.  At check weights seven to nine, our best codes exceed the best published $kd^2/n$ with exact distance at the same check weight in most length bands up to $320$ qubits.  All codes are available in a public database.
\end{abstract}

\section{Introduction}
\label{sec:intro}

Two-block codes are one of the simplest possible classes of LDPC codes~\cite{MacKay_IEEEIT_2004_sparse_graph,Kovalev_PRA_2013_kronecker,Panteleev_Quantum_2021_bposd,Lin_PRA_2024_2bga}. For a decade after their introduction, they were mainly studied by Pryadko and collaborators, who set the foundation for much of the field~\cite{Kovalev_PRA_2013_kronecker,Wang_arXiv_2022_gb_distance_bounds,Lin_PRA_2024_2bga}. At small system sizes, they also constitute some of the best codes, by almost any metric: one need only look as far as the Gross code~\cite{Bravyi_Nature_2024_bb_codes} to see this sentiment backed by titans of industry. The advent of the Gross code ushered in an explosion of interest in two-block codes, with a massive amount of work done in the last two years alone (Table~\ref{tab:literature-sources} lists those that contribute codes to our comparison).

Besides simplicity, the other big boon of two-block codes is that many of them are embeddable in 2D and translationally invariant. This is always possible for abelian groups generated by at most two elements, which includes every cyclic group, via a twisted torus~\cite{Bravyi_Nature_2024_bb_codes,Liang_PRXQ_2025_twisted_tori}; abelian groups that need three generators and non-abelian groups may require larger unit cells. Translation invariance also permits matching decoders: bivariate bicycle codes, and more generally two-dimensional translation-invariant codes, can be decoded by minimum-weight matching~\cite{Sahay_arXiv_2026_matching_decoder,Tan_arXiv_2026_generalized_matching}. While translationally invariant, their checks generically involve a large amount of crossing. However, the recent rise of platforms with movable qubits~\cite{Bluvstein_Nature_2022_coherenttransport,Bluvstein_Nature_2024_logicalprocessor,Pino_Nature_2021_qccd} means that the crossing of a connection is immaterial -- far more important is efficiency of movement. The work of Ref.~\cite{Viszlai_ACMTQC_2026_qsieve} shows that bivariate bicycle codes on an ordinary torus are almost preternaturally co-designed for atom movement -- any syndrome extraction schedule reduces to rigid movements with only two transfers between static and mobile traps; the same construction applies on twisted tori, with at most four rigid moves per stencil vector as on an ordinary torus (App.~\ref{sec:layout}). 
Given that atom transfer is at the moment the noisiest part of an acousto-optic deflector (AOD) movement schedule~\cite{Manetsch_Nature_2025_6100qubits,Bluvstein_Nature_2024_logicalprocessor,Wang_ChinesePhysicsLetters_2026_fast_transfer}, as well as the fact that non-rigidity in motion can be responsible for further heating effects~\cite{Bluvstein_Nature_2022_coherenttransport,Manetsch_Nature_2025_6100qubits}, we argue that the rise of neutral atom platforms demands we at least take this gift horse out of the stable for a ride.

Accordingly, logic in two-block codes is an extremely active line of work, with schemes ranging from the flagship architecture of IBM built on the Gross code~\cite{Yoder_arXiv_2025_tour_de_gross}, to automorphism and fold-transversal gates for bivariate bicycle codes~\cite{Eberhardt_arXiv_2024_bb_logical_operators}, to entangling gates from generalized Dehn twists for some classes of bivariate bicycle codes~\cite{Tiew_IEEEIT_2025_dehn_twists}, to lifted surgery over abelian~\cite{Berent_arXiv_2026_lifted_surgery} and non-abelian two-block codes~\cite{Pandey_arXiv_2026_nonabelian_lifted_surgery}. Two-block codes also number among the most frequent running examples in general surgery schemes for qLDPC codes~\cite{Williamson_arXiv_2024_gauging_logical,He_arXiv_2025_extractors,Cowtan_arXiv_2025_transversal_surgery,Swaroop_arXiv_2024_universal_adapters,Zheng_arXiv_2025_high_rate_surgery}. Translation automorphisms make such measurement gadgets cheaper to address: IBM's architecture shifts logicals into a fixed measurement unit by automorphisms~\cite{Yoder_arXiv_2025_tour_de_gross}. Every two-block code has such translations, which we expect will greatly simplify logic. With this in mind, we will not revisit logic here.

This brings us to this work. Recent attention has turned to other kinds of lifted product and qLDPC codes (see App.~\ref{app:literature} for a thorough but non-exhaustive survey). While two-block codes are great for movement, and great for logic, they seem to lose out to more complex lifted constructions in terms of raw parameters \cite{Kasai_arXiv_2026_orthogonality_barrier,Lee_arXiv_2026_ultra_high_rate,Bhardwaj_arXiv_2026_high_rate_processors,Hong_arXiv_2026_rate_fifth, Zhao_arXiv_2026_ultra_high_rate_atoms}. This fits the usual narrative -- the simpler and easier to operate a code is, the worse its parameters. In this work, we strive to show that this is not the case, at least, in the small code regime (here taken as $n \leq 320$) and at the check weights comparable to other high rate codes. Along the way we also unearth new good codes at check weights $5$ to $9$. Our method is nothing more than algebra-constrained code search for good codes.

This paper contains no new construction, only a more stubborn search. Two-block codes are as old as qLDPC codes themselves, and searching them is not new either \cite{Lin_PRA_2024_2bga,Wang_arXiv_2023_two_block,Bravyi_Nature_2024_bb_codes}. What we do differently is search them harder: we search more exhaustively, more carefully against the literature, and fast enough to do so on a laptop. This work does not exist in a vacuum, and is enabled by many recent advances in both the theory literature~\cite{Lin_PRA_2024_2bga,Eberhardt_arXiv_2024_bb_logical_operators,Symons_arXiv_2025_bb_covering_graphs,Hirasaki_arXiv_2026_lifting_lp,Guemard_IEEEIT_2025_lifts_css,Sabo_arXiv_2026_semisimple_bb} as well as distance estimation techniques~\cite{Dumer_IEEEIT_2017_distance_verification,Pryadko_JOSS_2022_qdistrnd,Cuellar_SwarmEvolComput_2021_genetic_distance,Webster_arXiv_2026_distance_finding,Lee_arXiv_2026_ultra_high_rate}. The initial rounds of brute force uniform sampling were large enough that we could make statistical observations about where good codes lived, with rather striking concentration by group structure. This allowed us to develop targeted sampling methods, further increasing our yield of good codes. It turns out that this alone is enough to put two-block codes back along almost the entire Pareto front at small sizes. Perhaps the most interesting outcome is not any single code but a large enough database revealing a map of where good codes live (Sec.~\ref{sec:where}). Good codes concentrate in a few groups, and abelian and non-abelian groups trade places as $n$ grows. These observations steered the later rounds of search. We offer some explanation for why some groups give logical qubits, and offer the open question of why some groups give large distance to the community.

The remainder of this paper is structured as follows: In the rest of the introduction, we motivate weight nine codes, and then take a brief detour to discuss AI methodology. In Sec.~\ref{sec:setup}, we introduce the codes and metrics used in our search, followed by a description of the search space and search pipeline. In Sec.~\ref{sec:results}, we discuss a selected slice of our results, including several standout codes, the Pareto frontier for check weights seven to nine, exhaustive censuses for weight-five and weight-six codes, and a description of the accompanying code database. In Sec.~\ref{sec:where}, we describe our insights into where good two-block codes live, and describe the targeted sampling routine. Throughout the main text, our discussion is informal. We defer all technical and mathematical material to an extended appendix, of which the last few sections are almost purely tables and figures.  

\paragraph{On weight nine. } Weight nine codes play a special role in this work. Our work unearths a large number of pretty good small weight-nine two-block codes. It is tempting to credit this to a superior search engine, but the truth is that weight-nine two-block codes are severely underexplored; the large number of codes shortlisted and presented in this work is perhaps as much a function of the small number of codes presented in the literature. The Lin--Pryadko census~\cite{Lin_PRA_2024_2bga}, the source of most lower-weight codes, stops at weight eight. Of the CSS codes with $n\le320$ in our literature set, $826$ have check weight eight but only $19$ have weight nine, and they were all from 2026, see Refs.~\cite{Hong_arXiv_2026_rate_fifth,Bhardwaj_arXiv_2026_high_rate_processors,Qian_arXiv_2026_multiagent_qldpc,Lee_arXiv_2026_ultra_high_rate, Zhang_arXiv_2026_routing_codes,Gu_arXiv_2026_planar_grid}. 

Why are weight-nine codes sparse in the literature? The first reason is practicality -- until recently weight nine looked quite impractical; weight nine with long-range connections even more so. Two recent developments have changed this: (i) the rise of movable platforms, whether atoms \cite{Bluvstein_Nature_2024_logicalprocessor,Bluvstein_Nature_2022_coherenttransport,Manetsch_Nature_2025_6100qubits} or ions \cite{Pino_Nature_2021_qccd} has made long-range checks cheap \cite{Viszlai_ACMTQC_2026_qsieve,Xu_NatPhys_2024_constant_overhead_atoms}. (ii) High-rate constructions at weight nine \cite{Bhardwaj_arXiv_2026_high_rate_processors,Hong_arXiv_2026_rate_fifth} and beyond \cite{Lee_arXiv_2026_ultra_high_rate,Yang_arXiv_2026_gala} show that such weights can be worth operating. We push this a little further by investigating weight-nine codes with 2D translational invariance (see App.~\ref{sec:layout}), which combine both the potency of (ii) and co-design with the rectangular AOD structures of (i).

The second reason is size: for any fixed $n$, the number of two-block codes grows drastically with check weight. As an example, per group of order $100$, there are $5.9 \times 10^{11}$ pairs of blocks $a, b$ (see Sec.~\ref{ssec:two-block} for the definition) containing the identity; we exhausted weight nine only to $n\le94$; extending to $n=100$, dominated by the $52$ groups of order $48$, would be the largest exhaustive census of this work, roughly a night or two on the laptop. Distances also grow with check weight, making rejection during search significantly slower as well. Furthermore, good codes are themselves rarer: against our own best code at each check weight, a uniformly drawn weight-nine code comes within $80\%$ of it a hundred to a thousand times less often than a weight-six one (Sec.~\ref{sec:where}). As such, to be able to make a thorough attack on weight nine at all does seem to require careful search design: the space grows, each code costs more to check, and near-best codes are rarer per draw. An unsteered search would pay all three costs at once. 

Finally, what about weight ten? We did not attempt this: everything in this work ran on a laptop.

\begin{figure*}[tbp]
\centering
\includegraphics[width=\linewidth]{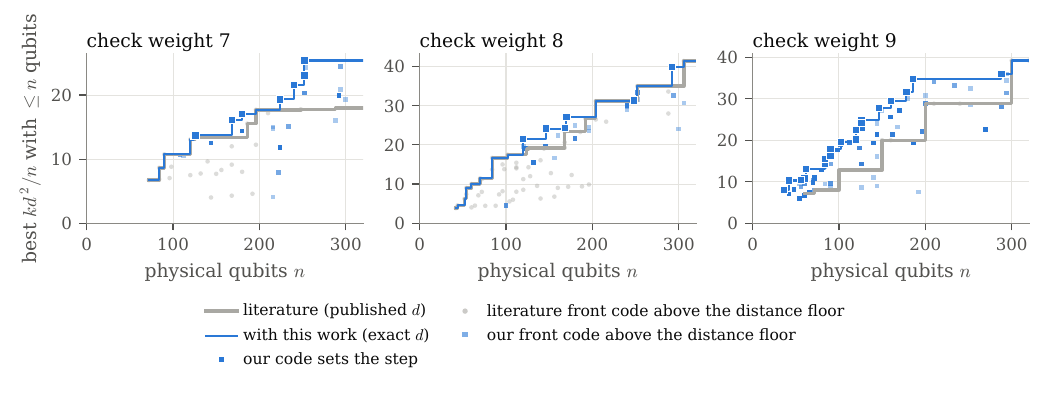}
\caption{The best $kd^2/n$ available with at most $n$ physical qubits, for CSS
codes of check weight seven, eight and nine and $n\le320$.  Gray: the codes from our literature set alone; blue: with the codes of this work.  Squares mark our codes that set a step; faint points are all codes on the $(n,k,d)$ Pareto front whose distance clears the floor of their respective bands. A step can carry over to larger $n$ through direct sums (two copies of a
code keep $kd^2/n$ and double $n$). $kd^2/n$ weights distance heavily and was the figure of merit our searches optimized; the same comparison for $kd/n$, for the overhead $n/k$ against $d$, and as Pareto fronts is in Apps.~\ref{app:fairness} and~\ref{app:pareto}.}
\label{fig:headline}
\end{figure*}

\paragraph{Hardware. } All computations in this work ran on two laptops: an
Apple MacBook Pro with an M5 chip (10 CPU cores, 16\,GB memory), which ran
the searches and the exact distance computations, and a 2020 13-inch MacBook
Pro with an Intel Core i5-1038NG7 (4 cores, 8 threads, 2.0\,GHz, 16\,GB
memory), which ran the independent verification (dist-m4ri, QDistRnd,
QDistEvol) and some of the censuses (weight five, and weight eight for $n\le100$).  Censuses used up to 8 worker processes; no run exceeded about
8 hours, and longer jobs were split into resumable chunks over several
nights.  No GPU or cluster time was used.

\paragraph{AI methodology. } It has become (hopefully!) standard to provide an AI disclosure where AI is used. A quick scan of the arXiv reveals that a growing share of work in this field is now somehow AI-powered.  Rather than the usual one-line AI disclosure, I instead offer a brief description of AI methodology, which I hope will become more standard in AI-assisted research. 

\textit{Literature} -- Claude was used for large-scale literature review. The finer details of this literature review process are in App.~\ref{app:literature}. I had two priorities here. First, given that I was trying to construct Pareto frontiers (see App.~\ref{app:pareto}), I felt I could not afford to miss out on good codes in the literature due to human error. Second, I did not want to fail to attribute theory results already in the literature. LLMs readily re-derive known results without flagging that they are known, so AI-assisted work is at particular risk of under-attribution. Every citation in this paper has been checked by hand.

\textit{Software} -- Prior to prompting, a brief plan was written by hand with two explicit kinds of gates. The first set of gates were human gates, which were focused on catching errors in implementation. This involved requiring outlines and notebooks performing certain tasks that were interpretable and testable. The second set of gates were soundness gates. These were focused on fixing critical interfaces, e.g.\ when a code is identified, it must provide two check matrices satisfying commutativity. This ensured modularity and boundedness in software-related tasks by fixing start and end points -- exactly the regime in which LLMs excel in this very transient moment. Soundness gates also allowed us to loosen up a little when optimizing search -- it is easier to experiment with search when you know that whatever happens it must return a valid two-block code.\footnote{I thank Thiago Bergamaschi for enlightening discussions on this idea.} Nevertheless, we still established strict verifiability bounds. In particular, I do not know C.\footnote{At least, it's been more than a decade since I used C, and I do not remember it.} For verifiability purposes, I asked Claude to write this project completely in Python, and fortunately, numba already proved fast enough for our purposes. The hand-written plan was refined in discussion with Claude, and then Claude was used to write all code. After establishing the bounds of the task, I observed far fewer errors than expected -- most of my corrections to the code were for legibility and reuse. It was only near the end of the plan that I started to notice significant drift, most of which was fairly benign and of the form of Claude over-doing tasks.

In terms of exporting results, I have tried to make sure that almost no substantive new result in this paper has to depend on trusting AI-written code: for upper bounds on code distance, we provide explicit witnesses for reported distances in the public database~\dburl, and the exact distances of our best codes with $d\lesssim15$ are confirmed by an independent open-source engine (dist-m4ri). Beyond that they rest on our engine, albeit cross-checked by independent upper bound searches for lighter logicals.

\textit{Theory} -- Fortunately (or unfortunately!) there is actually not much new theory in this work. All reductions and screens we used were already in the literature in some form -- most of the work was adapting these to our context, which mostly meant small extensions beyond abelian codes\footnote{ All new proofs use only standard techniques.}. Insofar as I did try to develop some new theory for the search, the resulting screens were simply not as effective, or marginal over what we already had (see App.~\ref{sec:screens-unused}). Claude played two important rules here. First, it enabled very quick iteration through testing each literature statement as a search screen. Second, it played an important role in pattern matching from the initial round of uniform sampling data, which allowed us to refine `where the good codes live' and shift to a targeted sampling. It also taught me the basics of Wedderburn theory, which eased the reading and verification against standard textbooks \cite{Knapp_2016_advanced_algebra,Webb_2016_finite_group_representations}.

\textit{Writing --} All prose in this document is written by me\footnote{The exception is App.~\ref{app:literature}, which concerns how the literature code set was constructed, which was AI-drafted and only human-edited for clarity.}, although most technical formatting (of text, tables, and figures) was AI-assisted. I expect most readers of this work to be LLMs. Nevertheless, if you are human and were to spend your hours reading this, the least I can do is to also spend my hours writing it. It is my hope that communication in our field retains a significant degree of human participation. 

\paragraph{Code availability.} The code library will only be made available with a forthcoming work, for which this work primarily serves as a benchmark. The main methods should be reproducible from their descriptions in the appendix. We emphasize that all results remain verifiable from this text and the accompanying database.

\paragraph{Acknowledgements.} I thank in particular Willers Yang for extensive discussion of the content of this work. I also thank Guo Zheng and Tomas Jochym-O'Connor for useful discussions. I thank Aashish Clerk, Liang Jiang, and Zhang Jiang for giving significant time, space, and support in pursuing independent interests.

\section{Codes, Metrics, and Search}
\label{sec:setup}

\subsection{Two-block group-algebra codes}
\label{ssec:two-block}

To describe the primary objects of this work, we now provide a brief informal description of \textit{two-block group-algebra codes} \cite{Lin_PRA_2024_2bga}. For interested but uninitiated readers, mathematical vocabulary is developed in App.~\ref{app:prelim}, followed by a more detailed description and review of two-block codes in App.~\ref{sec:codes}. 

Now let $G$ be a group and let $a, b \in \mathbb{F}_2[G]$ be a pair of group algebra elements. A two-block group-algebra code~\cite{Lin_PRA_2024_2bga}, written $C(a, b)$, has $n = 2 | G|$ qubits in left and right blocks, with each qubit within each block corresponding to a group element. The check matrices are
\begin{equation}
  \begin{aligned}
    &H_X = \bigl(L(a)\ \big|\ R(b)\bigr),\qquad
    &H_Z=\bigl(R(b)^{\mathsf T}\ \big|\ L(a)^{\mathsf T}\bigr),
  \end{aligned}
\end{equation}
with $L(R)$ being the left (right) regular representation. $a, b$ may be arbitrary group algebra elements; the check matrices commute due to the fact that the left and right regular representations commute.

Bivariate bicycle~\cite{Bravyi_Nature_2024_bb_codes} and generalized bicycle~\cite{Kovalev_PRA_2013_kronecker,Panteleev_Quantum_2021_bposd} codes are specializations of two-block codes. Both are instances of two-block codes over abelian groups -- the distinction between abelian and non-abelian two-block codes will be important for the discussion of our results.

\subsection{Metrics}

By convention, many papers use $kd^2/n$, which is inspired by the BPT bound~\cite{Bravyi_PhysRevLett_2010_tradeoffs}, either as a search objective \cite{Lu_arXiv_2026_divisor_driven,Lu_arXiv_2026_weight8_bb,Hirasaki_arXiv_2026_lifting_lp,Qian_arXiv_2026_multiagent_qldpc,CruzBenito_arXiv_2026_llm_bb_search}, or to rank and compare codes \cite{Symons_arXiv_2025_bb_covering_graphs,Hong_arXiv_2026_rate_fifth,Lee_arXiv_2026_ultra_high_rate,Wang_arXiv_2024_coprime_bb,Steffan_arXiv_2025_tile_codes,Aydin_arXiv_2026_coset_codes,Wang_arXiv_2026_bb_surface_codes,Sabo_arXiv_2026_semisimple_bb,Guemard_IEEEIT_2025_lifts_css}. This is not universal, e.g.\ Ref.~\cite{Bravyi_Nature_2024_bb_codes} primarily compares against the surface code, Ref.~\cite{Bhardwaj_arXiv_2026_high_rate_processors} evaluated their codes in terms of codesign for logical operation, and Refs.~\cite{Okada_arXiv_2026_pair_partition,Okada_arXiv_2026_high_girth_cpm_lifts} used rate and girth as metrics. Nevertheless, our primary metric in most of our initial searches was $kd^2/n$ -- when we informally say `best' code in some class, we mean highest $kd^2/n$. One might object that this over-privileges distance (the converse being that it under-weights high-rate codes) and does not make much sense when the study of finite block codes systematically excludes the thermodynamic limit. In App.~\ref{app:fairness}, we also provide an accounting for alternative metrics such as $kd/n$ which were not directly used in the search. 

Beyond $kd^2/n$, our other search metric was the Pareto front, which is inspired by the Unitary Foundation's qLDPC challenge leaderboard~\cite{qldpc_challenge_2026}. We may think of a code as a point $(n, k, d)$, and say that a code $A$ \textit{dominates} a code $B$ if $n_A\le n_B$, $k_A\ge k_B$ and $d_A\ge d_B$, with $(n_A,k_A,d_A)\ne(n_B,k_B,d_B)$.  The \emph{Pareto front} of a set of codes is the subset that no code of the set dominates.
 Codes with equal parameters do not dominate each other, so both stay on the front. Our second round of searching focused on pushing out the Pareto front. 

One caveat that stands throughout this work is that it seems like non-CSS codes far outperform the best CSS codes on these metrics~\cite{Lee_arXiv_2026_ultra_high_rate,Khesin_arXiv_2026_mirror_codes,CruzBenito_arXiv_2026_llm_bb_search}. We only stage our search against the best CSS codes, but will highlight where non-CSS codes sit on these scales. 

\subsection{The search space}

There are many, many, many groups over which we can build two-block codes, and each group yields a very large number of codes, which makes the thorough or exhaustive search hard. To get some sense of the size of the problem, for $n\le200$, that is for groups of order at most $100$, we have $183$ abelian groups, and $863$ non-abelian groups \cite{Besche_IJAC_2002_millennium}. To get some sense of the problem size per group, we consider the codes $C(a, b)$ with check weight $w$ over some group $G$. Each block $a, b$ can be thought of as a subset of members of $G$, and the check weight splits into the two blocks as $w = w_a + w_b$, the check weights coming from $a, b$ respectively, so after fixing a group, we still have ${|G| \choose w_a} \times {|G| \choose w_b}$ codes, which grows combinatorially with check weight. 

Fortunately, large is not the same as infinite. The remarkable project of enumerating every finite group up to order $2000$ has been performed by Besche, Eick and O'Brien in Ref.~\cite{Besche_IJAC_2002_millennium}, and distributed
with GAP. Furthermore, over and between these groups, large swathes of these codes are equivalent. Additionally, most of the space has $k = 0$, and the parts where $k > 0$ are fairly structured, in a way that can be charted out by theory. This theory is hinted at in the next couple sections and developed in the appendix. Finally, if one were only interested in finding good codes, then most codes can be screened out without a full distance calculation. Under such reductions and aggressive screening, the scale of the problem becomes such that on the laptop we could perform exhaustive search over every group with $w = 5$ up to $n = 300$, $w = 6$ up to $n = 300$, and $w = 8$ up to $n = 100$. Surveys up to $w = 7$ remain pretty thorough, such that an exhaustive sample of quite a large parameter range should be accessible via a cluster. Thorough sampling for $w = 9$ is a fairly hopeless task -- up to $n \leq 200$ there are some $10^{17}$ possible $a, b$ pairs, which is exacerbated by the fact that the weight-nine codes have larger distances, and take much longer to screen per code. This highlights the importance (and effectiveness!) of the targeted sampling step in the search pipeline below. 

\subsection{The search pipeline}
\label{ssec:search_pipeline}

Fast distance estimators from the literature mean that calculating distance no longer seriously bottlenecks the code search process, at least for small codes~\cite{Dumer_IEEEIT_2017_distance_verification,Pryadko_JOSS_2022_qdistrnd,Cuellar_SwarmEvolComput_2021_genetic_distance,Webster_arXiv_2026_distance_finding,Lee_arXiv_2026_ultra_high_rate}. In terms of two-block codes, the bottleneck becomes enumeration. As such, this work has focused on getting enumeration to a manageable level via theory-based screens and reductions. 

Our search pipeline is illustrated in Fig.~\ref{fig:pipeline-flow}. The upper level illustrates the outer loop which picks out candidates, whereas the lower level illustrates the per-candidate pipeline every code goes through. We first describe the pipeline at a high level, and then we describe the theory that goes into each stage. A full treatment is deferred to App.~\ref{sec:search}, but here we provide a high-level description, which is less particular to two-block codes. Table~\ref{tab:steps} lists the steps of the pipeline together with their literature sources where available, and how much impact it has on the search, with evaluation details in caption. We emphasize that every step is sound: a reduction can only identify codes that are equivalent, not decide that codes are inequivalent, and a screen can only reject a code that provably misses its anchor, not accept one that seems to be good. Note that the impact column of the table can vary greatly with the group and the check weight, and should be read as orders of magnitude.

\paragraph{Uniform Sampling Pipeline.} We begin by dividing our work into bands via $w, n$, with $w = 5, 6, 7, 8, 9$, and within each check weight splitting into $n \leq 130, 130 < n \leq 170, 170 < n \leq 220, 220 < n \leq 320$. For each check weight and band, we chose an anchor code, which is the best CSS code (measured by $kd^2/n$) we could find in the literature for that check weight in that size band. We note that some check weights are more competitive than others -- in particular check weight six was probably the most heavily mined set of two-block codes, and contained the undisputed middleweight champion Gross code; check weight eight contained the pair-partition codes of Ref.~\cite{Okada_arXiv_2026_pair_partition}, the coset codes of Ref.~\cite{Aydin_arXiv_2026_coset_codes} and the covers of Ref.~\cite{Symons_arXiv_2025_bb_covering_graphs}, and check weight nine contained mitten codes \cite{Bhardwaj_arXiv_2026_high_rate_processors,Hong_arXiv_2026_rate_fifth}. In comparison check weight seven appeared the most under-sampled in the literature. The cataloguing process is detailed in App.~\ref{app:literature}, the table of anchor codes is provided in Table~\ref{tab:anchors}. The following search was then run either in `BPT mode', which optimized for $kd^2/n$ or `frontier mode', which optimized for codes that would push the Pareto frontier while also having a distance above the floor $d - 2$ of their anchor codes.

In the first pass of the code search, we perform uniform sampling over all groups and block pairs. In particular, we set a time budget for sampling over each group given in the GAP SmallGroups library \cite{GAP4,Besche_IJAC_2002_millennium}. For each group we fix the identity (see Cor.~\ref{cor:normal-form}, or Ref.~\cite{Lin_PRA_2024_2bga}) to appear in $a, b$ and then sample over the remaining group elements to complete $a, b$ for some check weight. The candidate codes are reduced up to equivalence -- the basic moves are listed in Table~\ref{tab:steps} and detailed in App.~\ref{sec:equivalences}. They are then screened using cheap bounds on $k$ and $d$. These are also listed in Table~\ref{tab:steps} and detailed in App.~\ref{sec:screens}; in brief most codes wind up being trivial, or having $k$ or $d$ sufficiently bounded that given the value of $n$ we already know that the code cannot match its anchor. Table~\ref{tab:w6-funnel} (App.~\ref{sec:search}) gives a numerical account of these steps for our weight-six census. 

At this point we finally perform calculations. Distance estimation is performed as a sandwich: a randomized estimator works its way downwards, while an exact calculator works its way upwards. If at any point, either reveals that the code will not meet its anchor we discard it. Once it comes within some distance ($90 - 100\%$ of the anchor parameters for different bands), we shortlist it and move on. The distance estimation stack is detailed in App.~\ref{sec:distance} and only implements techniques already known to the literature~\cite{Dumer_IEEEIT_2017_distance_verification,Lee_arXiv_2026_ultra_high_rate,Webster_arXiv_2026_distance_finding,Cuellar_SwarmEvolComput_2021_genetic_distance}. A further constant factor improvement comes from known automorphisms of the underlying groups and Tanner graphs.

At this point, we will have obtained some shortlist of codes. Our symmetry reductions capture most equivalences but there remain equivalences such as across groups which are not captured. To deal with that, these codes are deduplicated via reduction to a canonical form determined on the Tanner graph by pynauty~\cite{pynauty,McKay_JSC_2014_nauty}. Finally, the exact distances are calculated for the survivors. Logical error witnesses are kept for each code, and particularly good codes are verified independently by open-source distance estimation packages.

\paragraph{Targeted Sampling.} After an initial uniform sampling effort, the results are inspected to see what the winners have in common. Overwhelmingly it appeared that winners tended to cluster in particular groups, which suggested group theoretic structure. This additional structure is formalized in Sec.~\ref{sec:where} and App.~\ref{app:targeting}, and divides along the lines of abelian and non-abelian groups. In particular, for abelian groups, the effort was concentrated on finding groups with non-zero $k$. On the other hand, non-abelian groups tended to yield non-zero $k$ already, so the effort was concentrated on obtaining explicit lower bounds for $k$. Subsequent rounds of code search then switched to targeted sampling, which gave more time to groups obeying these conditions. This increased hit rate for survivors by four to fifteen times across the board, although we expect this is pretty dependent on details such as the available anchors in the literature at the time. We still maintained uniform sampling runs, as groups without good guarantees could still (though rarely) give good codes.

\begin{figure}[p]%
\centering
\scalebox{0.75}{{\hyphenpenalty=10000\exhyphenpenalty=10000
\begin{tikzpicture}[
    box/.style={draw,rounded corners=2pt,align=center,font=\small,text width=3.7cm,
                minimum height=1.65cm,inner sep=3pt,fill=white},
    outer/.style={box,fill=black!4},
    flow/.style={->,black!60,line width=0.6pt},
    note/.style={font=\footnotesize,text=black!60,align=center},
    >={Stealth[length=5pt]}]
  \node[outer] (t) at (0,2.75) {\textbf{target}\\check weight $w$, $n$ band};
  \node[outer] (u) at (4.4,2.75) {\textbf{uniform sampling}\\over all groups\\and block pairs};
  \node[outer] (wt) at (8.8,2.75) {\textbf{targeted sampling}\\toward the structures\\theory picked out};
  \node[outer] (th) at (13.2,2.75) {\textbf{theory}\\why are the hits good?};
  \draw[flow] (t) -- (u);
  \draw[flow] (th) -- (wt);
  \node[box] (p) at (0,0) {\textbf{reduce}\\pairs up to\\equivalence};
  \node[box] (s) at (4.4,0) {\textbf{screens}\\cheap bounds\\on $k$ and $d$};
  \node[box] (d) at (8.8,0) {\textbf{distance}\\witnesses first,\\then exact};
  \node[box] (r) at (13.2,0) {\textbf{record}\\fronts against\\the literature};
  \foreach \i/\j in {p/s,s/d,d/r} \draw[flow] (\i) -- (\j);
  \coordinate (bus) at (0,1.37);
  \draw[black!60,line width=0.6pt] (wt.south) -- (wt.south |- bus) -- (bus);
  \draw[black!60,line width=0.6pt] (u.south) -- (u.south |- bus);
  \draw[flow] (bus) -- (p.north);
  \node[note,above right=0pt and 2pt] at (bus) {candidates $(G,a,b)$};
  \draw[flow] (r.north) -- node[note,right=2pt] {hits} (th.south);
  \draw[-{Stealth[length=7pt]},black!25,line width=3pt] ([yshift=-0.25cm]p.south west) -- ([yshift=-0.25cm]r.south east)
       node[midway,below=3pt,note] {fewer candidates at each stage, each more expensive to test};
\end{tikzpicture}}}%
\caption{The search pipeline. Top: the outer loop.  A target (check weight and $n$
band) is first sampled uniformly over groups and block pairs; we then return to theory to try to
explain what the good codes have in common. Later rounds then sample toward
those structures.  Bottom: every candidate $(G,a,b)$ is reduced up to
equivalence, screened with cheap bounds on $k$ and $d$, and given a distance,
witnesses first and exact for the survivors; the results are compared with the
literature and may feed the next round.}
\label{fig:pipeline-flow}

\vspace{0.1em}

\centering
\footnotesize
\begin{tabular}{@{}l>{\raggedright\arraybackslash}p{0.42\linewidth}>{\raggedright\arraybackslash}p{0.2\linewidth}lr@{}}
\toprule
stage & step & origin & effect & measured \\
\midrule
reduce & two-sided translations of each block (N1) & published~\cite{Lin_PRA_2024_2bga} & major & up to $4\times10^3\times$ \\
 & automorphisms of $G$ on both blocks (N2) & published~\cite{Lin_PRA_2024_2bga} & major & up to $3\times10^2\times$ \\
 & antipode with block swap; $X\leftrightarrow Z$ (N3, N4) & published~\cite{Lin_PRA_2024_2bga} & major & ${\sim}2\times$ \\
\midrule
sample & shared factors: both blocks singular in a common component of $\mathbb{F}_2[G]$ (abelian) & published~\cite{Panteleev_Quantum_2021_bposd,Lu_arXiv_2026_divisor_driven} & major & $7$--$15\times$ \\
 & rank targeting in simple components (non-abelian) & new & major & ${\sim}4\times$ \\
 & ideal targeting in non-simple components & new & minor & ${\sim}1.6\times$ \\
 & time per group weighted by the odd part of $|G|$ & new & minor & ${\sim}1\times$ \\
\midrule
screens: $k$ & connectivity & published~\cite{Lin_PRA_2024_2bga} & major & $3$--$90\%$ \\
 & $k$ from annihilators & published~\cite{Lin_PRA_2024_2bga,Eberhardt_arXiv_2024_bb_logical_operators} & major & $40$--$90\%$ \\
 & odd-weight blocks over $2$-groups ($k=0$) & new & major & whole groups \\
\midrule
screens: $d$ & hypergraph-product detection & published~\cite{Lin_PRA_2024_2bga} & minor & ${\le}3\%$ \\
\midrule
distance & witness-first ladder & published~\cite{Webster_arXiv_2026_distance_finding,Cuellar_SwarmEvolComput_2021_genetic_distance} & major & all rejections \\
 & connected-cluster search & published~\cite{Dumer_IEEEIT_2017_distance_verification} & major & -- \\
 & pruning bounds & published~\cite{Lee_arXiv_2026_ultra_high_rate} & major & $40$--$65\times$ \\
 & symmetry roots: one start per orbit of Tanner-graph automorphisms & published~\cite{Okada_arXiv_2026_pair_partition,McKay_JSC_2014_nauty} & major & $4$--$11\times$ \\
\midrule
record & canonical keys & published~\cite{McKay_JSC_2014_nauty} & minor & $0$--$36\%$ \\
\bottomrule
\end{tabular}
\captionof{table}{The steps of the pipeline (Fig.~\ref{fig:pipeline-flow}). The step names refer to the names of particular statements developed in the appendix, but omitted from the main text for readability. Origin:
published, with source, or new.  Effect: \emph{major}, the step removes a large share and the cost of a census
depends on it; \emph{minor}, rarely decisive.  Measured: for reductions, the
factor by which the step cuts the number of classes; for sampling, the yield of good codes
against uniform draws at equal time;
for screens, the share of the codes
reaching it that it rejects; for the distance stage, the factor by which the step cuts the nodes of the exact search tree, against the same search without it, measured at $d=10$ on a two-block and a lifted-product code (Table~\ref{tab:roots-102-mitten}; the pruning saving grows with $d$); the connected-cluster search is the exact engine itself, as in dist-m4ri, and App.~\ref{sec:distance} compares the whole engine with dist-m4ri. These measurements should be regarded as orders of magnitude, and are obtained from probes and census records detailed in App.~\ref{sec:search}.}
\label{tab:steps}
\end{figure}

\section{Results}
\label{sec:results}

\subsection{Several standout small codes}

We begin our results with a handful of small codes worth knowing about. This curated set of codes is displayed in Table~\ref{tab:standout}, and was chosen primarily for their high values of $kd^2/n$. The first six weight-nine codes in the table each sit above the best published $kd^2/n$ for weight-nine CSS codes of their size (see App.~\ref{app:pareto} and Fig.~\ref{fig:dvn-w9}). Interestingly, these six are all over abelian groups. For weight nine, our best non-abelian code ties the best abelian one at $n \simeq 300$: $[[294,24,21]]$ over a non-abelian group of order $147$ reaches $kd^2/n=36.0$, equal to the abelian $[[288,32,18]]$. The gap appears to close with $n$: in our four length bands the best non-abelian $kd^2/n$ is $0.44$, $0.67$, $0.73$ and $1.00$ of the best abelian one.  Non-abelian codes may therefore take over at sizes beyond those explored here, though the weight-nine non-abelian search was sampled, not exhaustive. At the smaller sizes, these also look somewhat like high-rate codes, with $[[102, 20, 10]], [[120, 24, 10]]$ coming in at rate $\sim 1/5$. Higher rates cost some distance, but the distance does not vanish: $[[270,72,6]]$, over $C_{45}\times C_3$ encodes $72$ logical qubits at rate $0.27$.  However, at such rates, pair-partition codes reach much larger distance at weight eight (e.g.\ $[[232,62,12]]$~\cite{Okada_arXiv_2026_pair_partition}). Sitting in the middle is $[[294,36,16]]$, which encodes $36$ logical qubits at distance $16$. 

In contrast, many of the good codes we found at lower check weights, say five, were non-abelian, e.g.\ the $[[288, 4, 20]]$ code. We note that weight-five codes were exhausted for $n\le300$ over all groups (Sec.~\ref{sec:exhausted}), so this is not a sampling artifact -- non-abelian weight-five codes do actually have better parameters. We note that the weight-eight anchors, given mostly by Refs.~\cite{Aydin_arXiv_2026_coset_codes,Okada_arXiv_2026_pair_partition} were particularly strong. The $[[292,22,23]]$ code we highlight reaches $kd^2/n=39.9$, the highest among our weight-eight codes.  Its closest rival is the coset code $[[306,22,24]]$ of Ref.~\cite{Aydin_arXiv_2026_coset_codes}, with the same $k$ on $14$ more qubits.  Its distance was published as an upper bound and confirmed via our engine, so it reaches $41.4$ and is the best weight-eight CSS code with $n\le320$ that we know of.  It is not the anchor of its band only because anchors were fixed on distances known to be exact at the time (App.~\ref{app:anchors}).\footnote{Our earliest comparisons, made before the anchors were fixed, stopped at $n=300$, which this code just exceeds; the length bands were later extended to $n\le320$.}  The band's actual anchor is the pair-partition $[[232,62,12]]$ of Ref.~\cite{Okada_arXiv_2026_pair_partition} ($38.5$), which sits in a different regime, at almost three times the rate. For scale, most of the best known stabilizer codes at the same $n, k$ have $d$ about $1.5$--$2\times$ larger, at much heavier checks, e.g.\ Ref.~\cite{Grassl_codetables} gives a code that is $[[102, 20, 16]]$ but has check weights $25$--$51$. 

Most of the chosen codes have a translation-invariant layout on a twisted torus, with the window column describing the window within which their checks would sit. The weight-five codes we chose to highlight have particularly small windows. All distances are exact, and each code comes with a witness; verification records are kept in the database~\dburl. Rather unfortunately, we did not discover new good codes with thickness two (such as the Gross code, see Ref.~\cite{Bravyi_Nature_2024_bb_codes}). As such, we primarily target platforms with movable qubits~\cite{Bluvstein_Nature_2022_coherenttransport,Bluvstein_Nature_2024_logicalprocessor,Pino_Nature_2021_qccd}, where crossings are free. We do not time movement schedules here as such numbers date quickly as hardware improves.  However, certain structural facts remain true. On a twisted torus every check is a translate of one stencil \cite{Liang_PRXQ_2025_twisted_tori} (App.~\ref{sec:layout}), so the construction of qSieve~\cite{Viszlai_ACMTQC_2026_qsieve} applies: each stencil vector takes at most four rigid translations of the ancilla array, as on an ordinary torus, and syndrome extraction needs only two transfers between static and mobile traps; qSieve times such schedules on ordinary tori at $3$--$6$\,ms.\footnote{We note that such timings depend greatly on assumed hardware, and every paper (this one included!) tends to quote numbers favoring their studied class of code.  Comparisons between codes are meaningful only when their schedules are timed on equal footing, which we do not attempt here.}

For two of these codes, we also provide LER estimates from simulated syndrome-extraction circuits, with noise at gates, preparation and measurement but none on idle qubits. We obtained syndrome extraction circuits by brute force search, which we did not attempt to optimize to any significant degree. Circuits are provided in the companion database~\dburl; the actual syndrome extraction circuits used have circuit distance $8$, and we leave the design of distance-preserving circuits to future work. For a decoder, we used the open source Relay-BP \cite{Muller_arXiv_2025_relay_bp}. In Fig.~\ref{fig:ler-102} we plot the results for the $[[102,20,10]]$ and $[[120,24,10]]$ codes. These are run each with only one ancilla per check. Both circuits are sequential ($X$ half then $Z$ half) and equivariant (one rigid move per CNOT layer). These simulations were run with gate and SPAM error only, as has become customary for neutral atom adjacent work \cite{Zhao_arXiv_2026_ultra_high_rate_atoms,Yang_arXiv_2026_gala,Hong_arXiv_2026_rate_fifth}. We stress that such simulations do not take into account the rather serious impact of transfer noise expected in many qLDPC codes. For these codes we expect transfer noise to be minimal due to their structure \cite{Viszlai_ACMTQC_2026_qsieve}. We performed $10$ rounds of syndrome extraction, for which we did memory experiments in $X$ and $Z$ and took the mean. We emphasize that these circuits have not been optimized, and simply provide a proof of principle of equivariant circuits suitable for atom motion with good distance and low overhead, and we expect a systematic study to vastly improve performance. Nevertheless, we discuss these results: Per logical qubit per round, $[[102,20,10]]$ reaches $p_L\approx2\times10^{-6}$ at $p=10^{-3}$, and $[[120,24,10]]$ about $4\times10^{-7}$. The larger code stays below the smaller one across the whole range, by about $6\times$ at $p=10^{-3}$ and $2$--$4\times$ for $p=1.5$--$3\times10^{-3}$, shrinking to $1.1$--$1.4\times$ near break-even.  Break-even, $p_L=p$, is at $p\approx0.57\%$ and $0.62\%$ respectively.

\begin{figure}[tbp]
\centering
\begin{subfigure}[b]{0.49\linewidth}
\centering
\includegraphics[width=\linewidth]{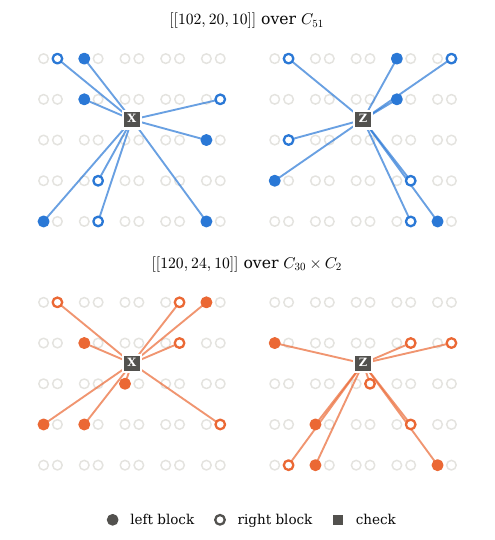}
\caption{$X$- and $Z$-check stencils.}
\label{fig:stencils-102-120}
\end{subfigure}\hfill
\begin{subfigure}[b]{0.49\linewidth}
\centering
\includegraphics[width=\linewidth]{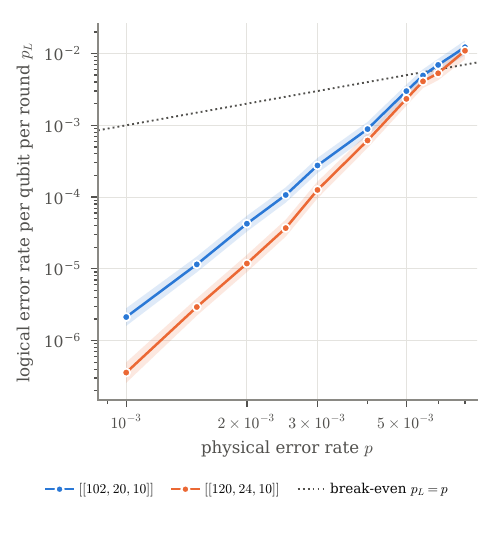}
\caption{Logical error rate.}
\label{fig:ler-panel}
\end{subfigure}
\caption{(a) The $X$ and $Z$ checks of each code on its most local twisted torus, which sit in a $5\times5$ window (App.~\ref{sec:layout}). Filled circles are left-block qubits, open circles right-block qubits, and the square is the check; every other check of the same type is a translation of the one shown. (b) Logical error rates of $[[102,20,10]]$ (blue) and $[[120,24,10]]$ (orange), per logical qubit per round, $p_L=1-(1-P_{\mathrm{fail}})^{1/(10k)}$, against the physical error rate $p$, averaged over $X$- and $Z$-memory experiments of $10$ rounds.  Circuits: one ancilla per check, circuit distance $8$, on $204$ and $240$ qubits.  Noise: depolarizing noise at gates, preparation, and measurement only, as in Refs.~\cite{Zhao_arXiv_2026_ultra_high_rate_atoms,Yang_arXiv_2026_gala}.  The decoder used is Relay-BP~\cite{Muller_arXiv_2025_relay_bp}.  Dotted: break-even, $p_L=p$.  Shaded bands show $95\%$ confidence intervals.}
\label{fig:ler-102}
\end{figure}

\begin{table*}[tbp]
\centering
\footnotesize
\resizebox{\linewidth}{!}{\begin{tabular}{@{}lrrllll@{}}
\toprule
code & $w$ & $kd^2/n$ & group & $a$ & $b$ & window \\
\midrule
$[[102,20,10]]$ & 9 & 19.6 & $C_{51}$ & $1 + x^{19} + x^{27} + x^{35} + x^{48}$ & $1 + x^{7} + x^{9} + x^{19}$ & 5 \\
$[[120,24,10]]$ & 9 & 20.0 & $C_{30}\times C_{2}$ & $1 + x^{6} + x^{16}y + x^{23} + x^{24}y$ & $1 + x^{2}y + x^{21}y + x^{29}$ & 5 \\
$[[126,16,14]]$ & 9 & 24.9 & $C_{63}$ & $1 + x^{17} + x^{21} + x^{31} + x^{60}$ & $1 + x^{6} + x^{19} + x^{28}$ & 5 \\
$[[150,24,13]]$ & 9 & 27.0 & $C_{15}\times C_{5}$ & $1 + y + x^{6}y^{2} + x^{7}y^{4} + x^{10}y^{3}$ & $1 + x^{5}y + x^{6} + x^{7}y$ & 5 \\
$[[186,20,18]]$ & 9 & 34.8 & $C_{93}$ & $1 + x^{7} + x^{52} + x^{57} + x^{65}$ & $1 + x^{20} + x^{31} + x^{32}$ & 5 \\
$[[210,28,16]]$ & 9 & 34.1 & $C_{105}$ & $1 + x^{56} + x^{66} + x^{72} + x^{94}$ & $1 + x^{15} + x^{19} + x^{61}$ & 6 \\
$[[270,72,6]]$ & 9 & 9.6 & $C_{45}\times C_{3}$ & $1 + x^{6}y^{2} + x^{10} + x^{25}y^{2} + x^{27} + x^{40}y$ & $1 + x^{3}y + x^{36}$ & 7 \\
$[[294,36,16]]$ & 9 & 31.3 & $(C_7\times C_7)\rtimes C_3$ & $1 + c + uv + u^{3}v^{3} + u^{4}v^{6}c + u^{5}v^{2}c$ & $1 + uv^{6} + u^{3}v^{2}$ & -- \\
$[[294,24,21]]$ & 9 & 36.0 & $(C_7\times C_7)\rtimes C_3$ & $1 + c + u + uv^{4} + u^{4}c^{2} + u^{4}v^{6}$ & $1 + u^{2}v^{5} + u^{3}v$ & -- \\
$[[292,22,23]]$ & 8 & 39.9 & $C_{146}$ & $1 + x^{7} + x^{47} + x^{110}$ & $1 + x^{9} + x^{44} + x^{61}$ & 6 \\
$[[252,16,20]]$ & 7 & 25.4 & $C_7\times D_{18}$ & $1 + r + xs + x^{3}r^{7}s$ & $1 + r^{4} + xr^{2}$ & -- \\
$[[294,18,20]]$ & 7 & 24.5 & $(C_7\times C_7)\rtimes C_3$ & $1 + c + v^{2}c^{2} + u$ & $1 + uv + u^{2}v^{4}$ & -- \\
$[[288,4,20]]$ & 5 & 5.6 & $C_{24}\times S_3$ & $1 + x^{8}s + x^{10}r$ & $1 + x$ & -- \\
$[[238,6,14]]$ & 5 & 4.9 & $C_{119}$ & $1 + x + x^{89}$ & $1 + x^{21}$ & 4 \\
$[[280,6,15]]$ & 5 & 4.8 & $C_{140}$ & $1 + x^{50} + x^{101}$ & $1 + x^{133}$ & 3 \\
\bottomrule
\end{tabular}
}
\caption{Standout codes of this work.  $C(a,b)$ with
$H_X=(L(a)\,|\,R(b))$ and $H_Z=(R(b)^{\mathsf T}\,|\,L(a)^{\mathsf T})$; $w$ is the check
weight.  For abelian groups $a,b$ are polynomials in generators $x$ (and $y$)
of the stated orders.  For the non-abelian groups they are words in a
presentation: $C_{24}\times S_3=\langle x\rangle\times\langle r,s\mid r^3=s^2=1,\ srs^{-1}=r^{-1}\rangle$
with $x$ of order $24$, $C_7\times D_{18}=\langle x\rangle\times\langle r,s\mid r^9=s^2=1,\ srs^{-1}=r^{-1}\rangle$ with $x$ of order $7$, and $(C_7\times C_7)\rtimes C_3=\langle
u,v,c\mid u^7=v^7=c^3=1,\ uv=vu,\ cuc^{-1}=u^4,\ cvc^{-1}=v^{\lambda}\rangle$, with $\lambda=4$ (a scalar action) for $[[294,18,20]]$ and $\lambda=2$ for $[[294,36,16]]$ and $[[294,24,21]]$; these are two non-isomorphic groups of order $147$.
Window: the smallest twisted-torus window, minimized over all our codes with
these parameters (--: no 2D layout).  Every distance is exact, with a
witness (Appendix~\ref{sec:verification}). The verification record of each code is kept in the
database. }
\label{tab:standout}
\end{table*}

\subsection{Moving the frontier: check weights seven to nine}
\label{sec:frontier}

Beyond $kd^2/n$, we also considered the Pareto frontier, as alluded to in the introduction. These Pareto frontier plots are contained in Appendix~\ref{app:pareto}. Our codes significantly push out the frontier for check weights seven to nine, and across the board, we find that two-block codes (whether new in this work or previously in the literature) dominate the frontier for small codes up to $n \sim 320$. The exception is that the high-rate end of the frontier tends to be captured by other codes, such as mitten or hyperbolic or pair-partition codes. 

As one broad snapshot of the frontier for weight-nine codes, in Fig.~\ref{fig:dvn-w9}, we plot the CSS codes with the best distance with at most $n$ qubits and at least $K$ logical qubits. For $K \leq 24$, our codes set this line, but from $K = 40$ on, other codes take over. In Fig.~\ref{fig:dvn-w7}, we show a similar plot for weight seven. Here, this work introduces fewer new points on the curve, but we note that every single point along these curves is a two-block code.

Finally, we note that these curves are primarily illustrative, as there is much more to a code than $n, k, d$, and a frontier is only as good as the database that created it. For reproducibility, the database, which is described in detail in Sec.~\ref{sec:database} and Appendix~\ref{app:literature}, is provided in Ref.~\dburl, with a literature cut-off of 29 September 2026. We invite corrections for codes from before 29 September that we have missed and should be on these plots.

\begin{figure}[tbp]
\centering
\begin{subfigure}[b]{0.54\linewidth}
\centering
\includegraphics[width=\linewidth]{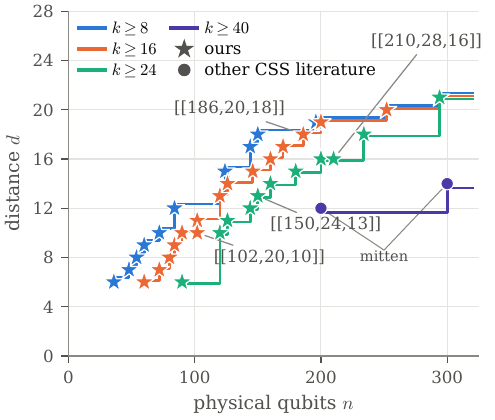}
\caption{Best distance against $n$.}
\label{fig:dvn-w9-panel}
\end{subfigure}\hfill
\begin{subfigure}[b]{0.43\linewidth}
\centering
\includegraphics[width=\linewidth]{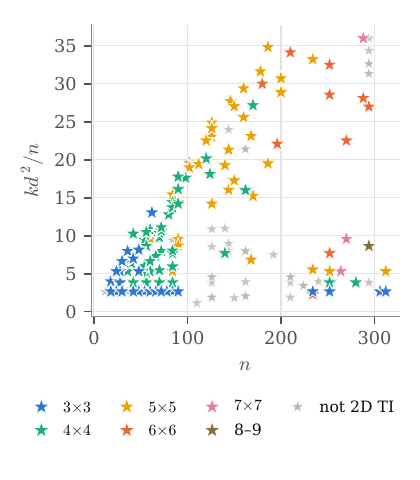}
\caption{Locality.}
\label{fig:locality-w9}
\end{subfigure}
\caption{Check weight nine.  (a) The best distance with at most $n$ physical qubits and at
least $k$ logical qubits, for several $k$, at check weight nine (including all CSS codes with check weight nine and $n\le320$ in the literature set).  Markers give
the code that sets each step, in the color of the largest $k$ it serves: stars, new with this work; squares, two-block codes from the literature; circles, other literature codes, tagged by family; faint, literature codes whose distance is only an upper bound.  This is the $(n,k,d)$ Pareto front sliced at fixed $k$, with $n$ kept on the
horizontal axis.  Check weight seven is in Fig.~\ref{fig:dvn-w7}, and weights five, six and eight in Fig.~\ref{fig:dvn-app}.  (b) Our weight-nine codes on the $(n,kd^2/n)$ plane, colored by the side of the smallest square window that holds every check on the code's most local twisted torus (App.~\ref{sec:layout}); gray: no translation-invariant 2D layout with one site per group element (abelian groups needing three generators, and non-abelian groups).}
\label{fig:dvn}\label{fig:dvn-w9}
\end{figure}

\subsection{What we can rule out: exhausted censuses}
\label{sec:exhausted}

Our search is fast enough to allow us to run exhaustive censuses of certain classes of codes on the laptop. Here exhaustive means that we check every single possible code in that class to see if they meet some particular standard of $kd^2/n$ or a frontier. Every code that falls short is rejected either by a logical lighter than required (a valid, often loose, upper bound on $d$) or because it has too few logical qubits.

Hence, in addition to finding particular new codes with good parameters, this work also allows us to make qualitative statements about the landscape of two-block codes. The scale of these censuses also played a key role in informing the content of Sec.~\ref{sec:where}.

The list of exhaustive censuses in this work include: 
\begin{itemize}
  \item Every two-block code at weight five up to $n = 300$, both abelian and non-abelian. 
  \item Every abelian two-block code at weight six for $n\le300$, and every non-abelian two-block code at weight six up to $n\le300$. 
  \item Every two-block code at weight eight up to $n\le100$, every group.
  \item Every two-block code at weight seven up to $n\le100$, and at weight nine up to $n\le94$, every group.
\end{itemize}

At small $n$ every check weight from five to nine is now exhausted (Table~\ref{tab:small-best}), and one code stands out: the $[[64,18,8]]$ of the Lin--Pryadko archive~\cite{Lin_PRA_2024_2bga}, over $C_4\times C_4\times C_2$ with $4{+}4$ blocks, whose distance we make exact here.  At $kd^2/n=18.0$ it is the best two-block code of any check weight from five to nine with $n\le94$: the censuses prove that no other code in this range reaches it.  Its group is a $2$-group, so at odd check weight it gives no code with $k>0$ at all (Prop.~\ref{prop:unit-blocks}). As such, weight nine cannot reach it \textit{in principle}. There are several other codes of similar size that are worth knowing about: at weight eight, $[[84,14,10]]$ ($16.7$; two classes, of which the one over $C_7\times S_3$ is a Lin--Pryadko archive code~\cite{Lin_PRA_2024_2bga}); at weight nine, $[[90,16,10]]$ ($17.8$), and $[[98,12,12]]$ ($17.6$). Longer lists, by weight and size, are in the database~\dburl.  Of course, code parameters alone do not tell the whole story. The small codes also differ in layout. Unfortunately, $[[64,18,8]]$ has no single-layer two-dimensional layout, since $C_4\times C_4\times C_2$ needs three generators. $[[90,16,10]]$ over $C_{15}\times C_3$ and $[[98,12,12]]$ over $C_7\times C_7$ both fit a $4\times4$ window on a single twisted torus, and so are the best small codes found by this work that live on a strictly two-dimensional layout.

\begin{table}[tbp]
\centering
\small
\resizebox{\linewidth}{!}{\begin{tabular}{@{}cllrlll@{}}
\toprule
$w$ & exhausted & best code & $kd^2/n$ & group & $a$ & $b$ \\
\midrule
5 & $n\le300$ & $[[60,4,8]]$~\cite{Lin_PRA_2024_2bga} & $4.27$ & $C_{30}$ & $1 + x^{10} + x^{26}$ & $1 + x^{21}$ \\
6 & $n\le300$ & $[[90,8,10]]$~\cite{Bravyi_Nature_2024_bb_codes} & $8.89$ & $C_{15}\times C_{3}$ & $y + y^{2} + x^{9}$ & $1 + x^{2} + x^{7}$ \\
7 & $n\le100$ & $[[90,12,9]]$~\cite{Lin_PRA_2024_2bga} & $10.80$ & $C_{15}\times C_{3}$ & $1 + x^{3}y + x^{10}y + x^{11}$ & $1 + x^{6}y^{2} + x^{13}y$ \\
8 & $n\le100$ & $[[64,18,8]]$~\cite{Lin_PRA_2024_2bga} & $18.00$ & $C_{4}\times C_{4}\times C_{2}$ & $1 + y + x + xyz$ & $1 + xy + x^{2}y^{3}z + x^{3}y^{2}$ \\
9 & $n\le94$ & $[[90,16,10]]$ & $17.78$ & $C_{15}\times C_{3}$ & $1 + y + x^{7}y + x^{10} + x^{11}$ & $1 + x^{3}y^{2} + x^{6} + x^{13}$ \\
\bottomrule
\end{tabular}
}
\caption{The best two-block codes with $n\le100$ by check weight $w$ ($n\le94$ at weight nine), all with exact distances; ``exhausted'' gives the range over which the census of that weight is complete, so no two-block code of that weight in the range has larger $kd^2/n$.  Codes with a citation are published; the others are ours.  Blocks as in Table~\ref{tab:standout}; $[[64,18,8]]$ needs three generators $x,y,z$ of orders $4,4,2$.  For comparison, sampling at weight nine found $[[98,12,12]]$ ($17.63$) just above this range.}
\label{tab:small-best}
\end{table}

This allows us to make definitive statements about certain classes of codes. One particularly striking statement is that the Gross code at $[[144,12,12]]$ is the unique best abelian weight-six two-block code with $n\le200$. It gives the only class of code (up to equivalence) at $kd^2/n=12$; the next best is $[[180,8,16]]$ at $11.4$. Over non-abelian groups it is also the best up to $n\le191$; at $n=192$ the non-abelian $[[192,12,14]]$ ($12.25$) beats it and $[[192,16,12]]$ ties it, with both codes known from the Lin--Pryadko archive~\cite{Lin_PRA_2024_2bga}. Table~\ref{tab:w6-top10} details the ten best weight-six codes with $n\le200$, of which nine were already published.

Over non-abelian groups, $[[288,4,20]]$, $[[264,4,19]]$ and $[[240,4,18]]$ beat every published weight-five code of distance above six, whose best is $kd^2/n=5.33$, at $5.56$, $5.47$ and $5.40$. Only two hyperbolic surface codes~\cite{Breuckmann_thesis_2018_homological_codes} reach higher, $[[80,18,5]]$ at $5.62$ and $[[150,32,6]]$ at $7.68$, through rate rather than distance; both sit below the distance floors of their bands. At weight six the non-abelian census obtained $[[288,16,16]]$ at $kd^2/n=14.22$, over two non-abelian groups of order $144$, exceeding the best published weight-six code in this range, the $[[254,14,16]]$ of Liang et al.~\cite{Liang_PRXQ_2025_twisted_tori} ($14.11$). One of its two classes had already been submitted to the qLDPC Challenge leaderboard~\cite{qldpc_challenge_2026}, with an upper bound on its distance; the census makes $d=16$ exact. As a consistency check, every published two-block code that we could rebuild and that lies within the scope of one of our exhaustive censuses ($919$ codes, at check weights five to eight) appears among that census's classes, identified by canonical key. We stress that since the censuses were exhaustive, the above-mentioned codes indeed have the highest possible $kd^2/n$ for two-block codes with check weights five and six and $n \leq 300$.

\subsection{The Code Database}
\label{sec:database}

The numbers and codes in this paper can be checked against a public database~\dburl, which holds three things:
\begin{enumerate}
  \item Our codes, which number about $12{,}400$, each specified by its group (a SmallGroup identifier from GAP), its blocks, its parameters $[[n, k, d]]$ with $d_X, d_Z$ separately, the method by which its distance was obtained, with its verification record, and a witness logical of weight exactly $d$: every code in the database has an exact distance, so its witness is tight (for three codes, one of $d_X, d_Z$ is only bounded below).  (Codes rejected by our censuses are not stored as codes; the coverage records keep their bounds on $d$, which may be loose.)
  \item The literature set of codes (App.~\ref{app:literature}), comprising $2{,}611$ published codes from $47$ papers, selected by the rules stated there, each with its source and location and with our status for it: confirmed (our exact distance equals the published value), lowered (with a witness), or taken as published. We note that we computed the distances for every code in the Lin--Pryadko archive, and found their upper bounds to be largely accurate; only $121$ out of the $1{,}660$ codes whose published distance was an upper bound had upper bounds which reflected only one of the $X$ or $Z$ sector distances without the other. This speaks to the accuracy of randomized upper bounds.
  \item Coverage records of all the censuses we carried out for this work.
\end{enumerate}

The literature set has a cut-off of 29 September 2026 and is versioned. We welcome corrections, both to our codes and to the literature set, in particular codes published before the cut-off that we have missed. If you use the database, please cite this paper, and for literature records also the papers they come from, which each record names.

\section{Where the good codes live}
\label{sec:where}

Our censuses are large enough that we can say where good codes are, not just which ones exist -- a key empirical observation is that they cluster in certain groups; this is what enables the targeted sampling half of our search pipeline. In this section, we informally discuss these observations. The data is provided in the public database~\dburl{} for anyone to take a gander. We provide a more formal treatment of one particular observation in App.~\ref{app:targeting}.

Our first observation concerns a fairly consistent split between abelian and non-abelian two-block codes. Table~\ref{tab:abelian-vs-nonabelian} quantifies this by our best codes in each category. Across check weights, abelian codes outperform non-abelian codes at smaller $n$, with a crossover at larger $n$. The larger the check weight is, the further out this crossover seems to be. Weights five and six are exhausted here, and weight eight for $n\le100$; weight seven is sampled, though a full census at these sizes looks within reach of a cluster. At weights eight and nine even a cluster would not exhaust the space, so those rows are statistical samples. For weight nine, the codes tie at the highest $n$ band, which may just be sampling luck on the non-abelian side -- our shortlists contained far more abelian than non-abelian codes.

Our second observation is a null statement concerning covers. Previous works have managed to systematically build bigger and better two-block codes by covering smaller good codes \cite{Symons_arXiv_2025_bb_covering_graphs}. At the same time, the distance of a cover is bounded by that of the base code -- in particular, the distance of the cover is at most $h$ times the distance of the base code, where $h$ is the number of times it covers the base code, whenever the base code has $k>0$~\cite{Hirasaki_arXiv_2026_lifting_lp,Symons_arXiv_2025_bb_covering_graphs,Guemard_IEEEIT_2025_lifts_css}. One question we hoped to answer with our large censuses is which of these effects win out? In the end, it seemed like we had just as many frontier codes that were covers as were primitive, and we were unable to obtain any consistent rule for when one outperformed the other.

Our third observation concerns the specific case of weight-five codes: our exhaustive census of the weight-five codes also reveals a best distance scaling like $\sqrt{n}$, consistent with the numerical fits in Ref.~\cite{Lin_PRA_2024_2bga}. The best analytical bound derived from an effective dimension, via Refs.~\cite{Hopkin_arXiv_2026_compactified_fracton,Arnault_arXiv_2025_bravyi_terhal_variant}, is $O(n^{2/3})$. We might then conjecture a bound that goes like $d \lesssim c \sqrt{n}$ for some constant $c$. For the codes up to $n = 300$ we know that abelian codes stay below $\sim 1.11\sqrt{n}$. Non-abelian codes with $k=2$ reach about $1.6\sqrt{n}$, but for $k \geq 4$, non-abelian codes stay below $\sim 1.18 \sqrt{n}$. This gives a stopping rule for when we can say a sampled distance is good enough. Figure~\ref{fig:w5-sqrt} shows the best weight-five distance against $n$ for abelian and non-abelian codes, with the line $d=1.18\sqrt{n}$; non-abelian codes pull ahead from $n\approx155$, and the abelian ones level off at $d=13$--$15$.

Our fourth observation is that good codes are not uniformly spread across groups. The most extreme case of this is the odd-weight $2$-group screen: at odd check weight, $2$-groups give no code with $k>0$ at all. However, even across the rest of the sample space, good codes concentrate in certain groups. In Table~\ref{tab:concentration}, we tabulate results for a selection of our censuses, in which we find that often $90 \%$ of short-listed codes concentrate in as small as $3$--$18\%$ of groups. The top groups for each of these censuses are also listed. Most strikingly, the top group for the weight-six abelian census for $n \leq 200$ contained $41 \%$ of the short-listed codes with $kd^2/n > 9$. 

Besides these cumulative numbers, we also list several stray observations on where good codes concentrate. The order of the group seems to matter, though less cleanly than the concentration above.  At weight nine, odd-order groups give good codes more often: in the non-abelian uniform run, $31$ of the $39$ codes above the bar came from odd-order groups, which received under a quarter of the draws. Much of this comes from a few groups: $24$ of the $31$ non-abelian codes come from just two, $(C_9\times C_3)\rtimes C_3$ and $C_3\times(C_9\rtimes C_3)$, so the parity of the order is partly the concentration of the previous paragraph seen through a coarse lens.  At weight seven the same ordering appears, but on too few codes to be sure (five above the bar from odd-order groups); at weight six it differs, with groups of order twice an odd number leading among abelian groups.  These numbers sound a little like numerology, and we were not able to turn them into a mechanism: the forecast of $k>0$ below, for instance, gives about the same chance to odd orders and to orders divisible by four.  Operationally we used the observation only by spending more sampling time on the most productive groups.

Our fifth observation concerns $k$. The idea is that many codes have $k=0$, but this often occurs for systematic reasons that allow us to rule out many groups before even building the codes (e.g.\ the odd-weight $2$-group rule). Perhaps more surprisingly, one can construct a statistical model for predicting $k$. 

To understand this, note that group algebras can be decomposed into components, analogous to familiar irrep spaces, via Wedderburn theory \cite{Knapp_2016_advanced_algebra,Webb_2016_finite_group_representations}. These are roughly speaking the finest block decompositions of the algebra (for abelian groups, these are just the characters). From the general formula for the $k$ of a two-block code (see Ref.~\cite{Lin_PRA_2024_2bga}, or Prop.~\ref{prop:k}), one requires $a, b$ to be singular in a common component (for abelian groups, where every component is a field, singular means vanishing) for a two-block code to have non-zero $k$. Investigating these components, which can be done on the level of the group (in contrast to the code) allows us to forecast how often $k>0$. In general, groups whose algebra has many small components tend to have $k > 0$ more often -- the converse observation is that large components are highly unlikely to be singular by chance. A forecast of how often $k>0$ can be built from the group alone, by modeling the images of the blocks in the components as random matrices, and empirically it works quite well (Appendix~\ref{app:targeted-sampling}), which means that uniform sampling really does look quite a lot like a random matrix process. Note that calculating $k$ once a code is obtained is fairly cheap; however, this lets you anticipate values of $k$ on the level of the group rather than the code. This is described in more technical detail in App.~\ref{app:components}.

Our final observation is that good codes become rarer as the check weight grows. To quantify this, we compare our short-listed codes obtained during uniform sampling runs and exhausted censuses against our best two-block codes at each check weight and length band. This measures how hard each weight's best is to reach by uniform draws. About $10^{-4}$ of all weight-six block pairs (split $3{+}3$, $n=170$--$300$) give a code within $80\%$ of the best, while uniform draws at weight nine reach that level only about $10^{-7}$--$10^{-6}$ of the time. For weight eight, a test run of about $2\times10^5$ uniform draws did not turn up any such near-best codes, and the exhaustive census at $n\le100$ gives about $2\times10^{-7}$. Every near-best weight-nine code we have came from targeted sampling. This rarity is one possible reason why check weight does not order our results: our best weight-eight code, $[[292,22,23]]$ ($39.9$), beats our best weight-nine code at $n\le300$ ($36.0$), since the weight-nine space is the largest and its near-best codes the hardest to reach.  Parity may be another: even weights keep the $2$-groups, which odd weights exclude (Prop.~\ref{prop:unit-blocks}), and give every pair $k\ge2$ for free, since two blocks of even weight both annihilate the sum of all group elements.

\paragraph{From observations to samplers. } We now discuss how we tried to turn some of these observations into targeted samplers. In each case, we ran separate experiments where we ran uniform sampling and targeted sampling for a particular strategy and compared the results. We only ran longer censuses if we observed a significant difference in these experiments. Not all our targeted sampling strategies worked, but the ones that did increased code yield by about four to fifteen times in fixed time tests compared to uniform sampling. 

As mentioned, we attempted to use the second observation, and sampled for covers and non-covers, but found no consistent rule for when covers win. We did, however, run several censuses sampling covers on purpose (these runs are labeled ``structured sampling'').

The fourth observation is the easiest to apply, but also the least satisfying -- since a small number of groups yield most good codes, we simply aim for these groups. One heuristic that appeared to work (weakly) here was to aim for either odd-order groups, or even-order groups with no nontrivial normal $2$-subgroup. However, we were unable to derive any systematic rule that improved on just re-sampling good groups. 

The fifth observation actually did lead to a successful targeted sampler. The idea here is to perform targeted sampling based on lower bounds for $k$. This splits into abelian and non-abelian groups. Uniform draws for abelian groups mostly give $k = 0$, so it sufficed to bound $k > 0$. This led to a \textit{shared-factor sampling} strategy. For non-abelian groups, most draws already give $k > 0$, so one requires a larger bound for $k$, which then becomes a hyperparameter for the search strategy. In this case, this led to a \textit{rank-targeting} strategy. These strategies can be derived from the decomposition of the group algebras into components; in brief, investigating components allows us with cheap linear algebra computations to aim for block pairs with non-trivial $k$. This decomposition as well as shared-factor and rank-targeting strategies are described in App.~\ref{app:components} and App.~\ref{app:targeted-sampling}.

We would have liked observations about $d$ as sharp as those about $k$.  The one structural bound we know, the quotient cap, gave no consistent rule for when covers win. We could not explain the parity pattern either, which seems to concern $d$ rather than $k$; although, as mentioned the parity pattern may simply be a coarse-grained view of the most concentrated group. We did try several screens on $d$, but they were all either too costly or too weak to pay off.  Targeted sampling therefore guarantees $k$ and leaves $d$ to the search.

We close with an important caveat. Targeted sampling greatly increases the yield of decent codes. As we see in this work, this includes some pretty good codes. However, it remains possible that the best codes hide in groups that are largely barren. As an example, three of the twelve best non-abelian weight-nine codes come from groups that have order twice an odd number, which were almost completely barren otherwise. At the moment, we have no systematic theory to explain this, but it suggests caution before throwing one's lot in with targeted sampling -- at the very least, we might suggest a mix of sampling styles for the best yield.

\begin{table}[tbp]
\centering
\small
\begin{tabular}{ccccc}
\toprule
$w$ & $n\le130$ & $130<n\le170$ & $170<n\le220$ & $220<n\le300$ \\
\midrule
5 & \cellcolor{abC!30}A \,4.3 & \cellcolor{abC!30}A \,4.7 & \cellcolor{naC!35}N \,5.3 & \cellcolor{naC!35}N \,5.6 \\
6 & \cellcolor{black!8}tie \,9.6 & \cellcolor{abC!30}A \,12.0 & \cellcolor{naC!35}N \,12.2 & \cellcolor{naC!35}N \,14.2 \\
7 & \cellcolor{abC!30}A \,13.7 & \cellcolor{abC!30}A \,16.1 & \cellcolor{abC!30}A \,17.7 & \cellcolor{naC!35}N \,25.4 \\
8 & \cellcolor{abC!30}A \,21.6 & \cellcolor{abC!30}A \,27.1 & \cellcolor{abC!30}A \,26.5 & \cellcolor{abC!30}A \,39.9 \\
9 & \cellcolor{abC!30}A \,24.9 & \cellcolor{abC!30}A \,29.4 & \cellcolor{abC!30}A \,34.8 & \cellcolor{black!8}tie \,36.0 \\
\bottomrule
\end{tabular}

\caption{Which side holds the best two-block code, by check weight $w$ and length band: A, an abelian group; N, a non-abelian one (a code realized over both counts as abelian); the number is the best $kd^2/n$ in the cell.  At weight six and $n\le130$ the best abelian and non-abelian codes have the same parameters, $[[120,8,12]]$; at weight nine and $220<n\le300$ the abelian $[[288,32,18]]$ and the non-abelian $[[294,24,21]]$ tie at $36.0$.}
\label{tab:abelian-vs-nonabelian}
\end{table}

\begin{table}[tbp]
\centering
\small
\resizebox{\linewidth}{!}{\begin{tabular}{@{}clllcrrrrl@{}}
\toprule
$w$ & groups & $n$ & source & bar & groups & with hits & $50\%$ in & $90\%$ in & top group (share) \\
\midrule
6 & abelian & $n\le200$ & exhausted & $9.0$ & 183 & 17 & 2 & 8 & $C_{14}\times C_{7}$ (41\%) \\
6 & non-abelian & $200<n\le300$ & exhausted & $10.6$ & 419 & 155 & 13 & 74 & $C_{3}\times (C_{2}^{4}\rtimes C_{3})$ (24\%) \\
7 & non-abelian & $240\le n\le300$ & uniform & front & 309 & 56 & 16 & 48 & $C_{7}\times D_{18}$ (9\%) \\
9 & abelian & $90\le n\le164$ & uniform & $16.6$ & 67 & 18 & 4 & 12 & $C_{63}$ (22\%) \\
9 & non-abelian & $112\le n\le164$ & uniform & $16.6$ & 146 & 8 & 2 & 5 & $(C_{9}\times C_{3})\rtimes C_{3}$ (36\%) \\
\bottomrule
\end{tabular}
}
\caption{Good codes concentrate in a few groups.  For each source, the number of groups searched, the number with any code above the bar, and how many groups hold $50\%$ and $90\%$ of the codes above it, with the group holding the most.  This table records only unbiased sources: including exhausted censuses and uniform sampling runs.  Bar: the census bar on $kd^2/n$, or ``front'' for a run in frontier mode. An extended table is available in the supplementary data.}
\label{tab:concentration}
\end{table}

\section{Discussion}
\label{sec:discussion}

We end with some discussion. The primary message of this paper is clear, and hopefully by now, believable: small two-block codes are competitive once searched hard, and the previous gap from more elaborate constructions at $n \le 300$ is not intrinsic. 

To contextualize our work, we will compare our work with prior code searches. The main precedent is Lin--Pryadko's exhaustive search and archive \cite{Lin_PRA_2024_2bga} (see also Refs.~\cite{Arnault_arXiv_2025_22gb, Khesin_arXiv_2026_mirror_codes}). In spirit, the main difference from Ref.~\cite{Lin_PRA_2024_2bga} is only that recent advances in numerical techniques have allowed us to scale drastically. In contrast to these aforementioned searches, most other searches are highly directed and far from exhaustive, usually along some fixed ansatz.

This work was also inspired in large part by the many recently published search pipelines, especially in the sense of establishing a search pipeline on top of strong algebraic constraints. The general idea of structural filters, including co-design-related filters, followed by distance calculations is preceded by Refs.~\cite{Bhardwaj_arXiv_2026_high_rate_processors,Hong_arXiv_2026_rate_fifth,Kasai_arXiv_2026_orthogonality_barrier,Lee_arXiv_2026_ultra_high_rate}. The main difference here is that we return to an earlier, simpler code family. Among previous searches on two-block codes, our use of algebra to constrain the search was informed by Refs.~\cite{Lu_arXiv_2026_divisor_driven, Postema_arXiv_2025_bb_existence, Yang_arXiv_2026_gala, Lee_arXiv_2026_logical_spectroscopy}. As mentioned in the introduction, much of the theory has already been developed in the literature, albeit scattered across many different sources. A key contribution of this work is hence also a thorough literature review for two-block codes and isolation of the most useful filters for code search. 

We should also compare our work to previous learning or LLM proposals for code search. As discussed in the introduction, LLMs were used in building the software tools. However, as compared to Refs.~\cite{CruzBenito_arXiv_2026_llm_bb_search,CruzBenito_arXiv_2026_weight5_bb}, we do not allow inference itself in the search loop. Other LLM-driven searches pursue different targets: OmniQEC~\cite{Yan_arXiv_2026_omniqec} searches for codes by their circuit-level logical error rate at a fixed qubit budget, and structured concept evolution~\cite{Liu_arXiv_2026_structured_concept_evolution} evolves lifted-product code families through an algebraic mutation grammar, with codes of $n\ge1360$; in both, and in contrast to our work, the model proposed the candidates. Any modifications of search (such as transitioning from uniform to targeted sampling) require an additional theory period in which algebraic statements have to be proven and statistical statements tested before we allow any modification of the loop. We believe that in addition to tightening up the search loop, this also greatly boosts interpretability of search methods. Our controls are inspired by the equal-budget comparison with uniform search of Ref.~\cite{Chengyu_arXiv_2026_bayesian_codes}. Sampling aside, at weights five and six the space is small enough to census completely, so at these sizes the census closes the search for the best small ($n \leq 300$) two-block codes.

While the goal of this work is not to develop new theory, we note that our observations of the landscape of two-block codes open up many new questions for theory. We have tried to give partial answers on which algebraic structures lower-bound $k$ and how to exploit them. The concentration of good codes via group structure is extremely evocative, and we invite new theory work based on the large number of numerical results in our database. We know that good codes are sometimes covers and sometimes not, with no consistent rule found in this work. Some insights into covers were developed alongside this work, and will appear in forthcoming work. To go beyond $n = 300$, it seems likely that one has to put more time into the non-abelian search. A complete set of equivalence moves remains open, and there is clear space for improvement, which will further ease search. The scope of this work is $n\le320$ and check weight at most nine.  Going beyond weight nine will take more than compute alone -- near-best codes become rarer per draw as the check weight grows (Sec.~\ref{sec:where}), and the search space grows by orders of magnitude with each unit of weight. As such,  either a far larger budget or much sharper targeting, ideally backed by theory, will be needed.

The clear next steps are to develop logic for some of the good codes revealed by this work. Due to the fact that two-block codes are among the most symmetric qLDPC codes, as well as the large amount of prior work on logic for two-block codes, we expect this to be largely an engineering task. What remains open is the co-design and performance of logical gadgets such as extractors on actual neutral atom devices, and whether the highly compatible nature of syndrome extraction of two-block codes with AOD constraints carries over to logic. While we did some small-scale LER simulation in this work, circuits and decoders were not optimized, which also remains an engineering challenge. 

Finally, we hope that the database becomes a valuable reference for the community. It is meant as a checkable benchmark, and corrections are welcome.

\FloatBarrier
\appendix

\section{Mathematical Preliminaries}
\label{app:prelim}

The material in this section can all be found in the following textbooks, Refs.~\cite{DummitFoote_2004_abstract_algebra,Knapp_2016_basic_algebra,Knapp_2016_advanced_algebra,Webb_2016_finite_group_representations}. For two-block codes, everything happens inside one algebra, with most arguments not more than a line or two once the vocabulary is settled. This section fixes vocabulary and collects mathematical facts that will be used without proof in later sections. This section is elementary but heavy. Familiar readers may skip it, whereas unfamiliar readers are recommended to use this as a reference section, returning when necessary. 

\paragraph{Groups.} Let $G$ be a finite group. A homomorphism $\varphi: G \rightarrow G'$ is a group multiplication-preserving map which must satisfy $\varphi(g_1) \varphi(g_2) = \varphi(g_1 g_2)$, and an automorphism is a bijective homomorphism from $G$ to itself. A subgroup $N \leq G$ is \textit{normal}, written $N \trianglelefteq G$, if $g N g^{-1} = N$ for all $g \in G$. The cosets are $g N$, which form a group $G/N$. By commutativity, every subgroup of an abelian group is normal. For a prime $p$, a \textit{$p$-group} is a group of order $p^{\alpha}$ with $\alpha \geq 0$. In this work, $2$-groups will play a particularly outsized role. Every $G$ has a unique largest normal $2$-subgroup $O_2(G)$, and we will write $O_2(G) = 1$ when $G$ has no nontrivial normal $2$-subgroup.

\paragraph{Algebras.} Let $G$ be a finite group. We will work over the regular representation, which supplies us with orthonormal basis vectors $| g \rangle$ for each $g \in G$, and a left and right action $L, R$ satisfying 
\begin{equation*}
  \begin{aligned}
    &L(g) | h \rangle = | g h \rangle, & R(g) | h \rangle = | h g \rangle,
  \end{aligned}
\end{equation*}
for every $g, h \in G$. Note the contrast to the usual convention for right multiplication with an inverse. Here $L$ is a homomorphism $L(g) L(h) = L(gh)$ whereas $R$ is an anti-homomorphism $R(g) R(h) = R(hg)$. More generally, we can write down the left (right) action of an element $c = \sum_g c_g g \in \mathbb{F}_2[G]$, the group algebra. This is a formal sum over elements of $G$ with coefficients in $\mathbb{F}_2$. The action extends linearly, with $L(c) = \sum_{g} c_g L(g), R(c) = \sum_g c_g R(g)$. The left and right action thus give linear maps, specifying $|G| \times |G| $ matrices. For single elements, we sometimes refer to these maps $L(g), R(g)$ as the \textit{left} and \textit{right translations} by $g \in G$. For an abelian two-block code identified with a torus, these literally give physical translation. 

The ambient space is a vector space of dimension $|G|$ with basis indexed by $G$, where we can identify an element with its coefficient vector. We exclusively work over $\mathbb{F}_2$, where $-1 = 1$, so we use addition throughout, e.g.\ $g + e, x^m + 1$, and so on. The \textit{support} of $c$ is ${\rm supp}\, c = \{g: c_g = 1\}$, and its weight is $|c| = | {\rm supp}\, c |$. Over $\mathbb{F}_2$, an element is determined by its support, so we identify $c$ with its support, such that we can write expressions like $g \in c$ if we have $c_g = 1$. 

The \textit{antipode} is $c^* = \sum_g c_g g^{-1}$, and acts like an adjoint, $(cc')^* = c'^* c^*$. In particular, $L(c)^{\mathsf T}=L(c^*)$ and $R(c)^{\mathsf T}=R(c^*)$, which is why $H_Z=(R(b)^{\mathsf T}\,|\,L(a)^{\mathsf T})$ is the natural partner of $H_X=(L(a)\,|\,R(b))$. The \textit{augmentation} $\varepsilon(c) = \sum_g c_g \in \mathbb{F}_2$ is the weight modulo two. The \textit{augmentation ideal} is the kernel of $\varepsilon$, and is spanned by elements $g+e$. The center $Z( \mathbb{F}_2[G]) $ consists of elements that commute with everything, and is spanned by sums $\sum_{g \in C} g$ over the conjugacy classes $C$.

\paragraph{Ideals and quotients.} A \textit{left ideal} is a subspace $J$ that absorbs the algebra from the left, in the sense that $\mathbb{F}_2[G] J \subseteq J$. A \textit{right ideal} analogously is a subspace $J$ satisfying $J \mathbb{F}_2[G] \subseteq J$. A \textit{two-sided ideal} is both a left and right ideal. A two-sided ideal is analogous to a normal subgroup in that it lets one obtain the quotient algebra $\mathbb{F}_2[G]/J$. In fact, the map to a quotient group $G \rightarrow G/N$ linearly extends to a surjection $\mathbb{F}_2[G] \rightarrow \mathbb{F}_2[G/N]$.  The \textit{principal} left ideal generated by one element $c$ is 
\begin{equation*}
  \mathbb{F}_2[G] c = \{y c: y \in \mathbb{F}_2[G]\},
\end{equation*}
the left multiples of $c$. It is also the smallest left ideal containing $c$. The principal right ideal is defined accordingly. As vectors, these are also the column spaces of $R(c), L(c)$ respectively.  

\paragraph{Units, zero divisors, idempotents.} A \textit{unit} is an element of an algebra with a two-sided inverse. A \textit{zero divisor} is a nonzero $c$ with $cy = 0$ or $yc = 0$ for some nonzero $y$. An \textit{idempotent} is an element with $c^2 = c$. We say it is \textit{central} if it lies in $Z(\mathbb{F}_2[G])$, and we say a central idempotent is \textit{primitive} if it is not the sum of two nonzero orthogonal ($\eta\eta'=0$) central idempotents. An element $c$ is nilpotent if $c^m = 0$ for some $m$, and an ideal $J$ is nilpotent if $J^m=0$, i.e.\ every product of $m$ elements of $J$ vanishes. In a finite-dimensional algebra every nonzero element is either a unit or a zero divisor. Hence, $c$ is a unit exactly when $L(c)$ is invertible.

\paragraph{Annihilators.}  For $c \in \mathbb{F}_2[G]$, we write the right and left \textit{annihilator}, 
\begin{equation*}
  \begin{aligned}
    &{\rm ann}_r(c) = \{y \in \mathbb{F}_2[G]: cy = 0\}, & {\rm ann}_l(c) = \{y \in \mathbb{F}_2[G]: yc = 0\},
  \end{aligned}
\end{equation*}
which are a right and left ideal respectively. The \textit{nullity} of $c$ is $z(c) = \dim \ann_r(c) = |G| - {\rm rank} L(c)$, which is also the dimension of $\ann_l(c)$. To see this, note that with $P|g\rangle=|g^{-1}\rangle$ one has $R(c^*)=P\,L(c)\,P$, and $R(c^*)=R(c)^{\mathsf T}$, so $R(c)$ and $L(c)$ have the same rank.

\section{Two-block group-algebra codes}
\label{sec:codes}

Here we define the codes in the language of group algebras; for abelian groups this is the familiar polynomial representation. App.~\ref{app:prelim} collects the algebra we use.

\subsection{Definition and dimension}
\label{sec:definition}

Let $G$ be a finite group and let $a, b \in \mathbb{F}_2[G]$ be a pair of group algebra elements. A two-block group-algebra code~\cite{Lin_PRA_2024_2bga}, written $C(a, b)$, has $n = 2 | G|$ qubits in left and right blocks, with each qubit within each block corresponding to a group element. The check matrices are
\begin{equation}
  \begin{aligned}
    &H_X = \bigl(L(a)\ \big|\ R(b)\bigr),\qquad
    &H_Z=\bigl(R(b)^{\mathsf T}\ \big|\ L(a)^{\mathsf T}\bigr),
  \end{aligned}
\end{equation}
with $L(R)$ being the left (right) regular representation. $a, b$ may be arbitrary algebra elements; the check matrices commute due to the fact that the left and right regular representations commute. Recall that here we used the convention that the right regular representation acts as right multiplication, so it actually gives an anti-homomorphism satisfying $R(gh) = R(h)R(g)$. The check weight of this code is $|a| + |b|$, so it is LDPC when we choose $a, b$ to be fairly light elements. Symmetrically, each qubit is acted on by $|a| + |b|$ checks ($|a|$ of one type and $|b|$ of the other).

When $G$ is abelian, we have $L(c)=R(c)$. Every abelian group is the product of cyclic groups, which can be associated with polynomials in their generators. As special cases, cyclic groups yield generalized bicycle codes of Ref.~\cite{Kovalev_PRA_2013_kronecker} and groups of the form $C_m \times C_l$ are bivariate bicycle codes of Ref.~\cite{Bravyi_Nature_2024_bb_codes}.

The formula for the dimension of a two-block code is given as either Eq.~(25) of Ref.~\cite{Lin_PRA_2024_2bga} or Thm.~2.3 of \cite{Eberhardt_arXiv_2024_bb_logical_operators}. We state it below, and reprove it primarily to get used to the notion of annihilators, which will come into play later.

\begin{proposition}[Dimension \cite{Lin_PRA_2024_2bga,Eberhardt_arXiv_2024_bb_logical_operators}]
\label{prop:k}
Let
$t_1=\dim\bigl(\ann_r(a)\cap\ann_l(b)\bigr)$ and
$t_2=\dim\bigl(\ann_l(a)\cap\ann_r(b)\bigr)$.  Then
$\rank H_Z=|G|-t_1$, $\rank H_X=|G|-t_2$, and
\[
  k=t_1+t_2 .
\]
\end{proposition}

\begin{proof}
The checks of each type are labeled by $G$, so a vector of checks of
either type is again an element $y\in\mathbb{F}_2[G]$; this is the only
identification used.  A $Z$-check vector $y$ lies in the left kernel of $H_Z$ when $R(b)y=0$ and $L(a)y=0$, that is, $yb=0$ and $ay=0$; so this kernel is $\ann_r(a)\cap\ann_l(b)$, of dimension $t_1$, and $\rank H_Z=|G|-t_1$. Likewise an $X$-check vector $y$ lies in the left kernel of $H_X$ when $a^*y=0$ and $yb^*=0$; the antipode, which reverses products, maps this space onto $\ann_l(a)\cap\ann_r(b)$, so $\rank H_X=|G|-t_2$.  Finally $k=n-\rank H_X-\rank H_Z=t_1+t_2$.
\end{proof}

\subsection{Which pairs give the same code}
\label{sec:equivalences}

Many pairs $(a, b)$ give the same code with the qubits relabeled. There are only a few basic moves here, which divide the code landscape into classes; these have already been provided in Thm.~6 of Ref.~\cite{Lin_PRA_2024_2bga}. The following restates this over $\mathbb{F}_2$.

\begin{proposition}[Equivalences \cite{Lin_PRA_2024_2bga}]
\label{prop:equivalences}
Each of the following gives a code equivalent to $C(a,b)$.
\begin{enumerate}[label=(N\arabic*),leftmargin=30pt,itemsep=2pt]
  \item Two-sided actions of each block independently:
    $C(gah,\,kbl)$ for $g,h,k,l\in G$.
  \item Automorphisms on both blocks at once: $C(\varphi(a),\varphi(b))$ for $\varphi\in\Aut(G)$.
  \item Antipode with the blocks swapped: $C(b^*,a^*)$.
\end{enumerate}
In addition, (N4) exchanging $X$ and $Z$ in $C(a,b)$ gives $C(a^*,b^*)$,
which by (N3) is equivalent to $C(b,a)$.
\end{proposition}

Thoroughly accounting for the above is probably the single most important step in the code search, see Table~\ref{tab:w6-funnel}. How we enumerate and count these classes is explained in App.~\ref{sec:enumeration}. For the enumeration step, the following normal form is helpful.

\begin{corollary}[Classes in two stages]
\label{cor:normal-form}
\begin{enumerate}[label=(\roman*),leftmargin=24pt,itemsep=2pt]
  \item Every $C(a,b)$ is equivalent to a code $C(a',b')$ in which both blocks contain $e$.
  \item For a block $a$ containing $e$, let $T(a)$ be the set of its two-sided translations
    that contain $e$; these are exactly $g\,(a s^{-1})\,g^{-1}$ for $s\in a$ and $g\in G$.  Two blocks
    $a,b$ containing $e$ are related by (N1) if and only if $T(a)=T(b)$.
  \item The moves (N2)--(N4) send two-sided translations of a block to two-sided translations
    of its image, so they act on (N1)-classes: it suffices to apply them to one representative of each
    class.  The classes generated by (N1)--(N4) are the orbits of this action on pairs of (N1)-classes.
\end{enumerate}
\end{corollary}

\begin{proof}
(i) Translate each block by the inverse of one of its elements. This is equivalent via (N1). (ii) $uav$ contains $e$ if and only if $usv=e$ for some $s\in a$, i.e.\ $v=s^{-1}u^{-1}$, and then $uav=u(as^{-1})u^{-1}$.  If $b$ is a two-sided translation of $a$, the two have the same two-sided translations, so $T(a)=T(b)$; conversely $b\in T(b)=T(a)$ is a two-sided translation of $a$. (iii) Each move sends two-sided translations to two-sided translations:
$\varphi(gah)=\varphi(g)\varphi(a)\varphi(h)$ and $(gah)^*=h^{-1}a^*g^{-1}$, and the block swap (N3) only
exchanges which block is translated; (N4) applies the antipode to each block, which the formula for $(gah)^*$ covers.  So if $b$ is a two-sided translation of $a$, the image of $b$ under
a move is a two-sided translation of the image of $a$, and the moves act on (N1)-classes.
\end{proof}

\subsection{Layout on a lattice}
\label{sec:layout}

Many abelian two-block codes can be translationally invariantly laid out in 2D; in particular abelian codes with at most two generators can always be embedded on a twisted torus~\cite{Liang_PRXQ_2025_twisted_tori}. This describes many of our good codes. Codes with three or more generators or non-abelian codes may require larger unit cells, and have a lower degree of translational invariance. We treat the abelian two-generator case first, then the rest.

First, we describe how the checks are directly stenciled onto the bulk of a 2D square lattice without twisting. We start by considering the case with two generators; this is a bivariate bicycle code, laid out as in Ref.~\cite{Bravyi_Nature_2024_bb_codes}. We can write the generators $x, y$, with blocks being polynomials in $x, y$, and $x^m = y^l = 1$. With every site of the lattice, we associate a left (L) and right (R) qubit, as well as one $X$ check and one $Z$ check. We then associate the group element $x^j y^k$ with a translation in the horizontal direction by $j$ sites and the vertical direction by $k$ sites, so that the sites form an $m\times l$ torus.  Write $a=\sum_{\alpha\in A}x^{\alpha_1}y^{\alpha_2}$, where $A$ is the set of exponent pairs of the terms of $a$ (so $|A|=|a|$, and each $\alpha\in A$ is a single lattice point), and likewise $b$ with $B$.  Labelling sites by lattice points, the $X$ check at site $h$ acts on the left qubits at the sites $h-\alpha$, $\alpha\in A$, and on the right qubits at the sites $h-\beta$, $\beta\in B$; the $Z$ check at $h$ acts on the left qubits at $h+\beta$ and the right qubits at $h+\alpha$.  So every check of a type is a translate of one \emph{stencil}, and the $Z$ stencil is the $X$ stencil rotated by $180^\circ$ with the blocks swapped (Fig.~\ref{fig:stencils-102-120}).

While the above toroidal layout is most natural, it may not be the most compact. An alternative is to place the code on a twisted torus. Now set $x$ and $y$ to be the two unit translations of the lattice, one step right and one step up, and $u,v\in G$ the group elements they act as, with the site $(i,j)$ associated with $u^iv^j$. In the previous paragraph, one can take $u$ and $v$ to be the generators of the two factors, and then $x$ and $y$ are the familiar polynomial variables. However, any $u,v$ that together generate $G$ give a layout of the same code, and the ordinary torus is the special case in which $\langle u\rangle\cap\langle v\rangle=\{e\}$, so that no power of $v$ other than $e$ is a power of $u$.  In general some power of $v$ lies in $\langle u\rangle$, and walking up the lattice then brings one back to the bottom row shifted sideways, hence the twist of the torus; see Fig.~3 of Ref.~\cite{Liang_PRXQ_2025_twisted_tori} for an illustration of the twist and wrap-around.  This is unavoidable for a cyclic group of prime order, which has no decomposition into two factors. One advantage of considering twists is that a stencil on a twist can be more local than the natural torus layout: $[[120,24,10]]$ over $C_{30}\times C_2$ spans $21$ columns of its $30\times2$ torus, but fits in $5\times5$ with $u=xy$ and $v=x^{11}$ (Fig.~\ref{fig:layouts}(a,b)).

This picture also works for a single cyclic group $G = \langle g \rangle$ of order $m$. The natural picture is a ring of sites, but this is space inefficient. Instead, one can do something similar to before by simply choosing elements $u=g^p$ and $v=g^q$ that together generate $G$ (which requires that $\gcd(p,q,m)=1$). As an example, $[[102,20,10]]$ over $C_{51}=C_3\times C_{17}$ needs a $10\times10$ window on its most local ordinary torus but fits in $5\times5$ with $u=g^2$ and $v=g^7$ (Fig.~\ref{fig:layouts}(c)). All locality numbers (stencil windows) in the text are obtained by optimizing over such generators and layouts. Over the fronts of App.~\ref{app:pareto} at check weights seven to nine, $184$ of the $262$ parameter sets $(n,k,d)$ we hold have a two-dimensional layout, and $151$ of those fit a $5\times5$ window; at weight nine alone, $114$ of $139$ do, $100$ within $5\times5$.  The others are codes over non-abelian groups or over abelian groups that need three generators.  

For an abelian group that needs three generators, or a non-abelian group, there is no layout with one site per group element and every check a translate of one stencil. However periodicity still survives, albeit with a larger unit cell. To see this, relabel every qubit and every check in both blocks at once, by $g\mapsto gk$ for a fixed $k\in G$.  This is the permutation $R(k)$ on each block and on the checks. Since $R(k)$ commutes with $L(a)$ and conjugates $R(b)$ to $R(k^{-1}bk)$, it maps the code to itself whenever $k$ commutes with $b$ (likewise $g\mapsto kg$ when $k$ commutes with $a$).  A subgroup $K$ that acts this way and is abelian and generated by at most two elements can be laid out as a torus as above.  For an abelian group every element commutes with the blocks, and $K$ is needed only because $G$ itself needs three generators; for a non-abelian group $K$ must also commute with $b$.  The cosets $gK$ are then layers, each a copy of the torus of $K$, and a unit cell, one site in each of the $|G|/|K|$ layers, holds $2|G|/|K|$ qubits.  For example, $[[64,18,8]]$ over $C_4\times C_4\times C_2$ takes $K=C_4\times C_4$ and lives on two $4\times4$ layers, and $[[294,18,20]]$ over $(C_7\times C_7)\rtimes C_3$ takes $K=C_7\times C_7$ and lives on three $7\times7$ layers, with each check reaching into all three.

On an ordinary torus, qSieve~\cite{Viszlai_ACMTQC_2026_qsieve} compiles each stencil vector into at most four rigid translations of the ancilla array, one for each way the vector wraps around the torus, and syndrome extraction needs only two transfers between static and mobile traps.  The same construction applies on a twisted torus, if the atoms are placed on a parallelogram within the trap grid. 

\begin{figure}[tbp]
\centering
\includegraphics[width=\linewidth]{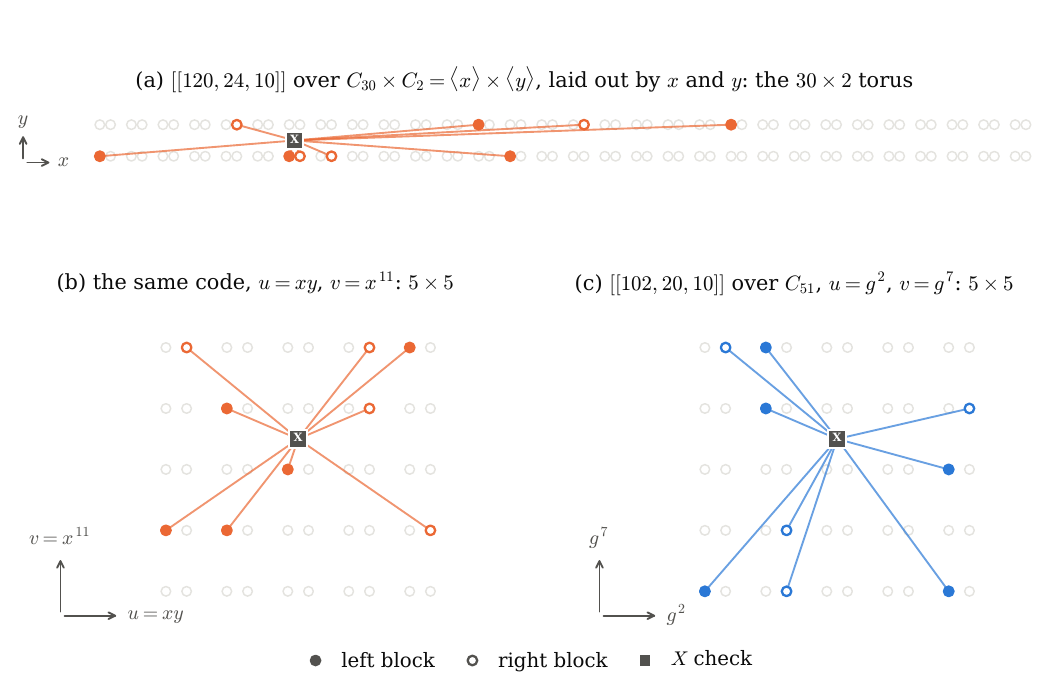}
\caption{Layouts of two-block codes, drawn as in Fig.~\ref{fig:stencils-102-120}: each site carries a left (filled) and a right (open) qubit, the square is the $X$ check, and the arrows give the group elements that one step right and one step up act as.  (a) $[[120,24,10]]$ over $C_{30}\times C_2=\langle x\rangle\times\langle y\rangle$ laid out by its generators: the ordinary $30\times2$ torus, on which one check spans $21$ columns.  (b) The same code laid out by $u=xy$ and $v=x^{11}$, a twisted torus: every check fits in $5\times5$ sites.  (c) $[[102,20,10]]$ over the cyclic group $C_{51}=\langle g\rangle$, which has a single generator, laid out by $g^2$ and $g^7$, again in $5\times5$.}
\label{fig:layouts}\label{fig:layouts-bb}\label{fig:layouts-gb}\label{fig:layouts-2bga}
\end{figure}

\section{Enumeration, screens and exact distance}
\label{sec:search}

\subsection{Enumeration up to equivalence}
\label{sec:enumeration}

Prop.~\ref{prop:equivalences} gives the basic steps for obtaining equivalent codes. Here we operationalize that into an enumeration procedure for searching two-block codes.

First, we itemize every group of each order within the band of interest. This is provided directly by the SmallGroups library~\cite{Besche_IJAC_2002_millennium} in GAP~\cite{GAP4}, with the groups identified by the SmallGroup id. Abelian and non-abelian groups are run as separate censuses.

Next, we obtain equivalence classes of codes, which are the orbits of the moves (N1) -- (N4) of Prop.~\ref{prop:equivalences}. Cor.~\ref{cor:normal-form} has two consequences here: it characterizes the (N1)-class of a block by the set $T(a)$ of its two-sided translations containing $e$, and it also tells us we can obtain orbits by first obtaining (N1)-classes, and then acting on one representative of each class to find which classes collapse under (N2) -- (N4). 

This proceeds in several systematic steps, performed for each group: 
\begin{enumerate}
  \item For each block weight, list the blocks containing $e$ and reduce them to (N1)-classes.
  \item Form pairs of block classes.
  \item Merge pairs under $\Aut(G)$ on both blocks at once (N2), the block swap (N3) and the $X/Z$ exchange (N4), applying each move to one representative of each class. The generators of $\Aut(G)$ are also obtained from GAP.
\end{enumerate}

This yields the classes of codes that we then iterate over. Before building any code, we also discard disconnected codes, since these are direct sums of smaller codes; with $e$ in both blocks, the connected components of the Tanner graph are read off the group multiplication table. We note that the equivalence moves are not complete -- to the best of our knowledge, a complete set of equivalence moves is an open question. On the remaining classes, we perform one more merge by canonical key. It suffices to note that a canonical key is a canonical labeling of the Tanner graph obtained through pynauty~\cite{pynauty,McKay_JSC_2014_nauty}, which identifies codes with isomorphic Tanner graphs. The canonical keys only compare check matrices, not stabilizer groups, so equal keys mean equivalent codes, but not conversely. We note that even after (N1) -- (N4), this step still merges a sizable fraction of classes in some groups ($36\%$ over $C_{24}\times S_3$) but none in others.

\begin{table}[tbp]
\centering
\scriptsize
\begin{tabular}{@{}llrrrr@{}}
\toprule
stage & move & \multicolumn{2}{c}{$3+3$} & \multicolumn{2}{c}{$4+2$} \\
\cmidrule(lr){3-4}\cmidrule(l){5-6}
 & & count & step & count & step \\
\midrule
all pairs of blocks &  & $6.49\times10^{12}$ &  & $4.82\times10^{12}$ &  \\
identity in both blocks & (N1), first step & $7.40\times10^{9}$ & $877\times$ & $4.88\times10^{9}$ & $988\times$ \\
one-sided translations & (N1), one side only & $8.24\times10^{8}$ & $9.0\times$ & $7.22\times10^{8}$ & $6.8\times$ \\
two-sided translations & (N1) & $3.49\times10^{7}$ & $23.6\times$ & $4.17\times10^{7}$ & $17.3\times$ \\
antipode and exchange & (N3), (N4) & $1.02\times10^{7}$ & $3.4\times$ & $2.62\times10^{7}$ & $1.6\times$ \\
automorphisms: census classes & (N2) & $1{,}525{,}317$ & $6.7\times$ & $3{,}821{,}760$ & $6.8\times$ \\
\midrule
connected & Prop. connectivity & $408{,}250$ & $3.7\times$ & $1{,}016{,}017$ & $3.8\times$ \\
$k>0$ & Prop. dimension & $52{,}286$ & $7.8\times$ & $1{,}016{,}017$ & $1\times$ \\
distinct canonical keys & Sec. canonical keys & $44{,}333$ & $1.2\times$ & $815{,}470$ & $1.2\times$ \\
\midrule
above the bar $kd^2/n>9$ & screens, witness ladder & $98$ & $452\times$ & $1$ & $815{,}470\times$ \\
exact $d$, distinct codes & exact engine & $93$ & $1.1\times$ & $1$ & $1\times$ \\
on the weight-six front & Pareto front & $6$ & $15.5\times$ & $0$ & --- \\
\bottomrule
\end{tabular}

\caption{A worked example of the whole search: weight-six codes over all
$863$ non-abelian groups of order at most $100$ ($n\le200$), $2$-groups
included, for the two splits of the check weight.  Each row applies one more
step to what the row above leaves, and ``step'' is the factor it removes.
The equivalences (top block) take $6.5\times10^{12}$ pairs to $1.5\times10^{6}$
classes at $3+3$; connectivity and $k>0$ (middle block) leave $44{,}333$
distinct codes; the screens and the exact engine (bottom block) leave $93$
codes with exact distance, of which $6$ are on the weight-six front (two of
them Lin--Pryadko archive codes found again).  At $3+3$ every $2$-group class
has $k=0$, since a block of odd weight is a unit over a $2$-group; at $4+2$
every connected class has $k>0$, because both blocks lie in the augmentation
ideal.  The $4+2$ codes over $2$-groups all fall below the bar ($kd^2/n\le9$ by witnessed
upper bounds).  Counts are taken group by group and summed, so a code realized over several groups is counted once for each: the $93$ codes with exact distance are $58$ distinct codes.}
\label{tab:w6-funnel}
\end{table}

\subsection{Screens}
\label{sec:screens}

After enumeration, most search time is spent on computing distances. In order to avoid this as much as possible, we implement \textit{screens}, which let us reject codes before computing any distances. The two requirements for a screen are:
\begin{enumerate}
  \item It must be cheap.
  \item It must be sound.
\end{enumerate} 
In particular, our search pipeline only permits screens that are backed by algebra. Furthermore, our screens are only allowed to reject codes, not accept them.

\begin{proposition}[Unit blocks]
\label{prop:unit-blocks}
\begin{enumerate}[label=(\roman*),leftmargin=24pt,itemsep=2pt]
  \item If $a$ or $b$ is a unit of $\mathbb{F}_2[G]$, then $k=0$.
  \item If $G$ is a $2$-group, an element of $\mathbb{F}_2[G]$ is a unit if and only if it has odd
    weight.  So every code over a $2$-group with a block of odd weight has $k=0$.
\end{enumerate}
\end{proposition}

\begin{proof}
(i) If $c$ is a unit and $cy=0$, then $y=c^{-1}cy=0$; likewise $yc=0$ gives $y=0$.  So $\ann_r(c)=\ann_l(c)=0$, and $t_1=t_2=0$ in Prop.~\ref{prop:k}.
(ii) Over a $2$-group the augmentation ideal $I=\{y:\varepsilon(y)=0\}$ is nilpotent
\cite[Prop.~6.3.3]{Webb_2016_finite_group_representations}: $I^m=0$ for some $m$.  An element of even
weight lies in $I$, so it is nilpotent and not a unit.  If $c$ has odd weight, then $y=c+e$ has even
weight, so $y\in I$ and $y^m=0$; then $c=e+y$ has inverse $e+y+\dots+y^{m-1}$, since
$(e+y)(e+y+\dots+y^{m-1})=e+y^m=e$. By (i) we then have $k = 0$.
\end{proof}

For abelian $G$, part~(i) is a special case of Cor.~2.6 of Ref.~\cite{Eberhardt_arXiv_2024_bb_logical_operators}. In terms of screens, part~(i) is tested once per block class, by checking the rank/invertibility of $L(c)$ as a matrix. Part~(ii) removes whole groups before any block is formed: it has the largest impact for weights five, seven and nine, where splitting the check-weight into two parts requires a block with odd weight, such that we can safely skip all $2$-groups ($36\%$ of the non-abelian
groups of order at most $100$, $73\%$ of those of order at most $150$).

The next proposition characterizes the connected components of the Tanner graph of a code. As mentioned earlier, we discard disconnected codes as they are direct sums of smaller codes. 

\begin{proposition}[Connectivity {\cite[Sec.~IV\,C]{Lin_PRA_2024_2bga}}]
\label{prop:connectivity}
Let both blocks contain $e$, and write $G_a$ and $G_b$ for the support groups, the subgroups
generated by their elements. For subsets $X,Y\subseteq G$ we write $XY=\{xy:x\in X,\,y\in Y\}$. The connected components of the Tanner graph of $C(a,b)$ are indexed by
the double cosets $G_axG_b$: the left and right qubits in a component are the
elements of one double coset.  So $C(a,b)$ is connected if and only if $G_aG_b=G$.
\end{proposition}

The basic intuition for the above is that multiplying by elements of the blocks gives moves across the Tanner graph; if one can reach everything via these moves, then we know the graph is connected.

Finally, an observation from Ref.~\cite{Lin_PRA_2024_2bga} about when two-block codes collapse to hypergraph product codes \cite{Tillich_IEEEIT_2014_hypergraph_product}, combined with the fact that the code is built from two classical codes, gives its parameters exactly, distance included.

\begin{proposition}[Hypergraph products {\cite[Statement~8 and the paragraph after it]{Lin_PRA_2024_2bga}}]
\label{prop:hgp}
Let both blocks contain $e$, and suppose $G_a\cap G_b=\{e\}$ and $G_aG_b=G$, so that every $g\in G$ is
uniquely $g=\alpha\beta$ with $\alpha\in G_a$, $\beta\in G_b$. Then $C(a,b)$ is the hypergraph product of
the classical codes with square check matrices $L_{G_a}(a)$, left multiplication by $a$ on
$\mathbb{F}_2[G_a]$, and $R_{G_b}(b)$, right multiplication by $b$ on $\mathbb{F}_2[G_b]$.  So
$k=2k_ak_b$, with $k_a,k_b$ the dimensions of the two classical codes, and $d$ is the least of their
distances and those of their transposes.
\end{proposition}

Such codes have distance bounded by classical codes of length $|G_a|$ and $|G_b|$, so they are rarely good. Our search simply recognizes them and records their exact parameters without a distance search. This screen has a pretty minor effect, but we kept it because minor is not zero, and it is extremely cheap ($2.4\%$ of the weight-five classes that reach it; $0$--$0.7\%$ on the probe samples).

\subsection{Screens considered but not run}
\label{sec:screens-unused}

We also considered several screens that the censuses do not run. These screens reject codes by exhibiting a light logical or proving that $k = 0$, similar in spirit to the ones in the previous section. We measured them by running every screen on the same $150$ codes of each kind, but found that none of them ever rejected a code that the witness-first ladder of App.~\ref{sec:distance} did not. All numbers below are quoted from this run. In order not to distract from the main discussion, here we only give a high-level overview, with further details and proofs, in particular of new statements, collected in a supplementary note in the
database repository~\dburl.

Two are strong but slow.  Logicals supported on a single block, the bound of Wang and Pryadko and of Lin and Pryadko~\cite{Wang_arXiv_2022_gb_distance_bounds,Lin_PRA_2024_2bga}, reject $53$--$98\%$ of codes on their own, but cost about three orders of magnitude more per code than the ladder.

Others are cheap but add nothing the ladder does not already do.  `Traps' read a light logical directly off the blocks: when $a=b$, when one block is a two-sided translation of the other, when the blocks share a
factor, or through the cosets of $\langle g\rangle$ when a block has a factor $e+g$.  These are published
for abelian groups~\cite{CruzBenito_arXiv_2026_llm_bb_search,Postema_arXiv_2025_bb_existence}, and their
two-sided forms hold for any group.  On their own they reject $4$--$71\%$ of codes, at a few times the ladder's cost.  The traps ran in the first non-abelian censuses. However, we dropped them as the witness-first ladder already rejected every code they rejected by the trap screens. The logicals supported on a single block were measured on probes only and were not adopted.

The rest act on $k$ or on the classes, and are weak or duplicate a cheaper step.  Multiplying both blocks by a common central unit leaves the code unchanged (for abelian groups, Statement~3(v) of Wang and Pryadko~\cite{Wang_arXiv_2022_gb_distance_bounds}), but it merges an almost negligible percentage of classes beyond (N1)--(N4). The unit test can be run in the smaller algebra $\mathbb{F}_2[G/O_2(G)]$, where $O_2(G)$ is the largest normal $2$-subgroup, and over a $2$-group a central involution $z$ gives $k\le2\,k(G/\langle z\rangle)$; both duplicate the per-block rank test and the $2$-group skip (Prop.~\ref{prop:unit-blocks}).  Finally, the nullity $z(c)=|G|-\rank L(c)$ of each block bounds $k\le2\min(z(a),z(b))$, which would prune blocks before pairing at a higher target $k$; every census requires only $k\ge1$.  

\subsection{Exact distance}
\label{sec:distance}

In this section, we discuss our distance engine, i.e. how we compute distances. Here, we take distance as $d = {\rm min}(d_X, d_Z)$. For abelian codes we always have $d_X = d_Z$~\cite[Lemma~3.7]{Lu_arXiv_2026_weight8_bb}. This can be asymmetric for non-abelian codes, but we did not take particular note of the asymmetry in this work. We note that our distance engine only uses algorithms already known in the literature. Besides accounting for the algorithms we use, this section will also try to offer a small explanation for why they go so far on two-block codes.

There are two parts to the distance computations. The first part of our search (see Sec.~\ref{ssec:search_pipeline}) is a coarse-grained short-listing, which only needs quick upper bounds to reject candidates. Upper bounds only require any witness of a particular weight. To obtain these, we used QDistEvol \cite{Webster_arXiv_2026_distance_finding,Cuellar_SwarmEvolComput_2021_genetic_distance}. Upper bounds were only allowed to reject, not accept a code. Each code was rejected immediately once it fell below a bar. In BPT mode, this was a $d$ that would cause its $kd^2/n$ to fall below the census bar, set at or below the anchor (App.~\ref{app:censuses}). In frontier mode, this was a $d$ that would push it off the frontier. All the codes we chose to publish come with such a witness. We call the combination the \emph{witness-first ladder}: for each candidate, QDistEvol first searches for at most $0.1$--$0.2$\,s for a logical lighter than the bar, and a logical it finds rejects the code as a witness. For our search, we found that when such a logical existed it usually appeared in about a millisecond, so most classes never reach the exact ladder. Only when no such logical is found does the exact ladder described below run, up to the bar. 

For our codebase, we ported QDistEvol to numba and benchmarked against the original package in the hope of speedup. However, we do not expect that this made a significant difference in the overall search efficiency. We capped the randomized upper bound per code by time (empirically chosen based on small trials), and before the port, we already found the original package to be extremely effective in obtaining randomized upper bounds. 
Next, for the shortlisted codes, we computed an exact distance. To do so, we had to obtain lower bounds via complete enumeration. We now discuss the exact distance engine in a little more detail, as this is where we find most of the speedup. We note that we expect this to be somewhat particular to two-block codes, but hope some of the lessons will generalize. The exact distance machine used three ingredients:
\begin{enumerate}
  \item The connected cluster search of Ref.~\cite{Dumer_IEEEIT_2017_distance_verification}.
  \item The pruning bounds from Ref.~\cite{Lee_arXiv_2026_ultra_high_rate}.
  \item Symmetry roots obtained from pynauty~\cite{pynauty,McKay_JSC_2014_nauty}.
\end{enumerate}
Each step was prototyped in python before being ported to numba. We briefly describe each step, deferring to the original references for details.

The connected cluster search is based on the observation that most candidate logicals can be rejected without even checking, for a simple reason. If the support of some operator $S$ satisfies every check, then every check that some $A \subseteq S$ leaves unsatisfied must also contain another qubit in $S - A$. To build a logical, one can imagine starting somewhere in this set. Then one only needs to grow the logical (or cluster) by checking which checks are unsatisfied, and then growing it along the support of these checks. We call the starting qubit the \textit{root}, as this growth induces a tree-like structure.

Pruning refers to two different methods used to prune branches off the induced tree. First, we note that finding the exact distance proceeds by going up a ladder, i.e. enumeration for $d = 1, 2, 3,...$ in order. The first pruning bound comes from the fact that in order to grow a cluster, one might require multiple qubits at once. At any given stage more than one check may be unsatisfied at once, and if these checks are disjoint, they cannot be satisfied by the same qubit. If the growth along this branch requires moving up the ladder, one can immediately prune the branch, i.e. it only needs to be checked later in the ladder. The other pruning bound is based on the observation that since each qubit sits in at most $D$ checks for an LDPC code, adding a qubit can change the number of violated checks by at most $D$. Hence at any given time, one can check how many checks are currently violated, and at least how many qubits must be added to the cluster in order to get to a null syndrome. If adding this many qubits moves us past the current weight of the ladder, we can prune the branch from the current step of the ladder.

Finally symmetry roots refers to the fact that if the Tanner graph has automorphisms, then the distance problem looks the same whether you start from some qubit, or its image under an automorphism. Enumerating which qubits are inequivalent under automorphisms reduces the problem accordingly. 

Empirically, we found two things made a big difference:
\begin{enumerate}
  \item For two-block codes, the trees obtained from the connected cluster search were extremely bushy near the leaves, which means that the pruning bounds made a massive difference. Hence, even though pruning bounds do not give a worst case improvement, they had a sizeable impact here. 
  \item Almost by construction, abelian two-block codes tend to have a large number of Tanner graph automorphisms. Non-abelian two-block codes have a smaller number of automorphisms, but they still made a significant difference.
\end{enumerate}

This adds up -- as an example for the two-Gross $[[288, 12, 18]]$, dist-m4ri took $10{,}459$\,s against about $2$\,s via our engine. Furthermore, we found computing distances to be roughly flat with $n$, only strongly scaling with $d$ itself. As a comparison, we also timed the computation on a non-two-block code (the mitten $[[150,30,10]]$ code took ${\sim}9$\,s via dist-m4ri, and ${\sim}0.02$\,s via our engine, against ${\sim}6$\,s and ${\sim}0.004$\,s for the two-block $[[102,20,10]]$), for which the speedup is still present but much less sizable. See Table~\ref{tab:roots-102-mitten} for a more systematic comparison.

Finally, we note that our early attempts used matching bipartition \cite{Webster_arXiv_2026_distance_finding,Dumer_ISIT_2016_distance_verification}, but found this to be too slow for our purposes. Similarly, we also tried SAT formulations, which were found to slow down too quickly with $n$, and set it aside. We did not benchmark for speed against the more standard integer programming methods. For checking the distances of our chosen codes, we further ran independent checks via QDistRnd \cite{Pryadko_JOSS_2022_qdistrnd} and dist-m4ri as published.

\begin{table}[tbp]
\centering
\footnotesize
\setlength{\tabcolsep}{4pt}
\begin{tabular}{@{}lccrrrrrrr@{}}
\toprule
 & & & & \multicolumn{3}{c}{nodes} & \multicolumn{3}{c}{seconds} \\
\cmidrule(lr){5-7}\cmidrule(l){8-10}
code & $|\Aut|$ & roots & sector & roots & no roots & no pruning & roots & no roots & dist-m4ri \\
\midrule
$[[102,20,10]]$ & $51$ & $2$ & $Z$ & $47{,}852$ & $538{,}209$ & $2{,}759{,}524$ & $0.004$ & $0.045$ & $5.99$ \\
 & & & $X$ & $69{,}325$ & $715{,}212$ & $4{,}387{,}262$ & $0.005$ & $0.057$ & $10.6$ \\
$[[150,30,10]]$ & $5$ & $30$ & $Z$ & $302{,}142$ & $1{,}108{,}554$ & $16{,}655{,}621$ & $0.024$ & $0.089$ & $9.29$ \\
 & & & $X$ & $333{,}652$ & $1{,}219{,}772$ & $13{,}922{,}707$ & $0.026$ & $0.093$ & $9.70$ \\
\bottomrule
\end{tabular}
\caption{Nodes visited by the exact search, summed over all rounds, for
the two-block $[[102,20,10]]$ over $C_{51}$ and the mitten $[[150,30,10]]$~\cite{Bhardwaj_arXiv_2026_high_rate_processors}, two codes of nearly the same rate ($0.196$ and $0.2$) and the same distance, $d=10$.  $|\Aut|$: automorphisms of
the Tanner graph (via pynauty); roots: one per orbit of
qubits; no pruning: roots kept, the pruning bounds of Ref.~\cite{Lee_arXiv_2026_ultra_high_rate} switched off.  Roots save $10$--$11\times$ on the two-block code and $3.7\times$ on the
lifted product, and the pruning bounds $58$--$63\times$ and $42$--$55\times$.  Node counts are exact and independent of machine or implementation. However, times are for one process on the laptop,
warm-up excluded, and compared to dist-m4ri with one thread (connected-cluster mode).}
\label{tab:roots-102-mitten}
\end{table}

\subsection{Verification}
\label{sec:verification}

There are two levels of verification: first, verification and benchmarking of our methods, and second verification of the provided codes. 

In terms of our methods, they were benchmarked against the literature as well as standard open-source libraries. In particular, we computed distances on the Lin--Pryadko archive in full, finding no contradictions with their upper bounds.\footnote{More precisely, every record whose published distance is a randomized upper bound of at most $15$ was settled; $121$ are lower than published, because the archive lists the distance of one sector.  Counts in Appendix~\ref{app:literature-set}.} We checked agreement against dist-m4ri wherever both reach. The engine was prototyped in pure-python before porting to numba, and both versions were extensively checked against each other.

An important part of the verification pipeline is the development of a canonical key. To obtain this, we first color the Tanner graph (qubits, $X$-checks, $Z$-checks), and then obtain a canonical label via  pynauty~\cite{pynauty,McKay_JSC_2014_nauty}. To save on memory, the key itself is a hash of the canonical label. Note this label is tied to a particular version of pynauty (2.8.8.1, wrapping nauty 2.8.8) but can be quickly recomputed on a new version if needed. 

In addition to deduplication during the enumeration step, this also helped us ensure that we could check our codes against known codes in the literature to confirm which codes were actually new. We note that this only checks that the Tanner graphs are not isomorphic -- the two codes could still have the same codespace. We do not consider this a huge issue for two reasons: first, we regard the presentation and checks of an LDPC code as an equally important part of specifying the code, and second, most of our standout codes are inequivalent to literature codes by parameters, which cannot hide such an equivalence.

Our database of codes is made public together with logical witnesses, so that their parameters can be checked by a reader completely independently of our distance engine and search pipeline. 

Finally, we have claimed that several of our censuses are exhaustive in some specified range. Every such census carries a proof of completeness: the orbit sizes, counted as the classes are built, must sum to $\binom{|G|}{w_a}$, $\binom{|G|}{w_b}$ and the raw pair count, and a separate script rechecks this from the records. We show this for the weight-six abelian census, amongst others, via a full-record mode in which every class also keeps a representative, and if rejected, its lightest logical. This is also available in the public database, and can be checked without our code. 

\section{Where the good codes live: details and targeted sampling}
\label{app:targeting}

This appendix provides supplementary information to Sec.~\ref{sec:where}. It primarily serves to elaborate on mathematical details related to targeting $k$. For further details on the other qualitative observations, we refer to the supplementary note in the database repository~\dburl.

\subsection{Components of the group algebra}
\label{app:components}

To understand our first statement, we must take a brief detour to understand the components of a group algebra. Conceptually this underlies all the statements we will make about $k$. 

These components are basically a block-diagonalization of the full algebra. Over $\mathbb{C}$, this is already familiar to the physicist: one can block-diagonalize algebras by projectors that commute with everything \cite{Serre_1977_linear_representations}. For an abelian group, this is just the Fourier transform, with the components being the characters, and for a non-abelian group the components are the irreps, each with its multiplicity. Here we instead work over $\mathbb{F}_2$. There are three things which must then be accounted for: First, the eigenvalues of group elements are roots of unity, which then live in finite extension fields $\mathbb{F}_{2^f}$, just as a real rotation can only be diagonalized after extending to $\mathbb{C}$. Second, the usual way of building projectors by averaging over a subgroup requires dividing by its order, which is impossible when that order is even since $2 = 0$ over $\mathbb{F}_2$. Third, when $|G|$ is even, even extending the field is not enough: some group elements act as Jordan blocks (in $\mathbb{F}_2[C_2]$, $(x+1)^2=0$), which no projector splits.  This never happens for finite groups over $\mathbb{C}$, where every representation can be made unitary.

A couple of preliminary facts. First, given a central idempotent $\eta$, an algebra $R$ splits into $R = R\eta \oplus R(e +\eta)$, where $e + \eta$ is also central idempotent, with both spaces two-sided ideals. In the matrix picture, right multiplication by $\eta$ is simply a projector commuting with both left and right representations, so this block diagonalizes the check matrices of a two-block code.  The \textit{components} are the finest possible splitting of an algebra, given by splitting along all the primitive central idempotents of an algebra~\cite[Prop.~3.6.1]{Webb_2016_finite_group_representations}. Note these are usually called blocks, but we reserve blocks for the $a, b$ of a two-block code. Note that finding the primitive central idempotents over $\mathbb{F}_2$ is polynomial in $|G|$, and only has to be done once per group: squaring is $\mathbb{F}_2$-linear on the commutative center $Z$ of $R$, since $(x+y)^2=x^2+y^2$ there, so the central idempotents are the kernel of the linear map $x\mapsto x^2+x$ on $Z$, and the primitive ones are obtained from a basis of that kernel by splitting each idempotent $\eta$ found so far into $\eta b$ and $\eta(e+b)$ for each basis element $b$.

Next, $\mathbb{F}_{2^f}$ is the field with $2^f$ elements, which is unique up to isomorphism~\cite[Thm.~9.14]{Knapp_2016_basic_algebra}. As mentioned this is needed for the eigenvalues of elements of $\mathbb{F}_2[G]$. Of course, idempotents give us components, which are full blocks rather than single eigenspaces. The analogous object here is the algebra of all $d \times d$ matrices over $\mathbb{F}_{2^f}$, written $M_d(\mathbb{F}_{2^f})$. 

A \textit{simple component} is a component that is isomorphic to $M_d(\mathbb{F}_{2^f})$. For $G$ with $|G|$ odd, all components are simple ($\mathbb{F}_2[G]$ is fully block diagonalizable). This is given by Maschke~\cite[Sec.~18.1]{DummitFoote_2004_abstract_algebra}, together with Wedderburn--Artin and Wedderburn's little theorem~\cite[Thms.~2.2, 2.48]{Knapp_2016_advanced_algebra}. As over $\mathbb{C}$, a simple component $M_d(\mathbb{F}_{2^f})$ is one irrep, $\mathbb{F}_{2^f}^d$, appearing $d$ times. All of this may not hold for $|G|$ even, where Jordan blocks may remain. Below, we will only make a statement about simple components, which are more controllable. 

One final odd end remains. We now know a simple component looks like $M_d(\mathbb{F}_{2^f})$. What determines $f$, and how do we compute it? The answer comes from centers. The center of $M_d(\mathbb{F}_{2^f})$ consists of matrices $\lambda I$, $\lambda\in\mathbb{F}_{2^f}$. This is one dimensional over $\mathbb{F}_{2^f}$, but our original algebra lives over $\mathbb{F}_2$, so the isomorphism must pull back to an $f$-dimensional space, which is the center of the component. Since the component is obtained by projection, we can also obtain its center by projection of the center of the ambient algebra\footnote{For $z\in Z$, $z\eta$ commutes with everything.  Conversely, if $c\in R\eta$ commutes with all of $R\eta$, then $c=c\eta=\eta c$, and for every $r\in R$, $cr=c(\eta r)=c(r\eta)=(r\eta)c=r(\eta c)=rc$, so $c\in Z$ and $c=c\eta\in Z\eta$.}, i.e. the center of a component $R \eta$ is in turn the projection of the center $Z$ of $R$, $Z(R\eta)=Z\eta$. Note that given $R$, $Z$ itself can also be obtained in time polynomial in $|G|$ via linear algebra. The factor $d$ can be obtained by a similar counting argument -- $M_d(\mathbb{F}_{2^f})$ has $d^2$ entries, each $f$-dimensional over $\mathbb{F}_2$. Hence, 
\begin{equation*}
  f=\dim_{\mathbb{F}_2}(Z\eta),\qquad d^2=\dim_{\mathbb{F}_2}(R\eta)/f,
\end{equation*}
which allows us to compute $f, d$ once per group without constructing the isomorphism $R\eta\cong M_d(\mathbb{F}_{2^f})$, which only needs to exist. 

For abelian $G$ ($d=1$) every component is a field, and for a cyclic group of odd order $m$ the
components are familiar: $x^m+1$ has no repeated factors (its derivative $x^{m-1}$ shares no factor with it), so by the Chinese remainder theorem~\cite[Sec.~9.5, Prop.~16]{DummitFoote_2004_abstract_algebra}
$\mathbb{F}_2[x]/(x^m+1)\cong\prod_i\mathbb{F}_2[x]/(p_i)$ over its irreducible factors $p_i$, and $f_i=\deg p_i$. One upshot is that for these groups finding the components is literally factoring $x^m+1$.  For example
$x^7+1=(x+1)(x^3+x+1)(x^3+x^2+1)$ gives $\mathbb{F}_2[C_7]\cong\mathbb{F}_2\times\mathbb{F}_8\times\mathbb{F}_8$.
For non-abelian $G$ part of the extension can go into the matrix size instead: in
$\mathbb{F}_2[C_7\rtimes C_3]\cong\mathbb{F}_2\times\mathbb{F}_4\times M_3(\mathbb{F}_2)^2$ an element of order $7$
has eigenvalues in $\mathbb{F}_8$, yet the two three-dimensional components are over $\mathbb{F}_2$, just as a
rotation matrix lies in $M_2(\mathbb{R})$ although its eigenvalues are complex.

The next proposition concerns dimensions of and from simple components. For a simple component $\mathbb{F}_2[G]\eta\cong M_d(\mathbb{F}_{2^f})$ and $c\in\mathbb{F}_2[G]$, write
$\nu(c\eta)=d-\rank(c\eta)$ for the nullity of $c\eta$: the image of $c$ in the component is a $d\times d$ matrix
over $\mathbb{F}_{2^f}$, and $\nu(c\eta)$ is the dimension of its kernel.  This needs no explicit
isomorphism either: left multiplication by $c\eta$ acts on each of the $d$ columns, so
$\nu(c\eta)=\dim_{\mathbb{F}_2}\{y\in R\eta:\ cy=0\}/(fd)$, computed inside $R\eta$.  We say $c$ \emph{vanishes} on a
component $\eta$ (simple or not) when $c\eta=0$; for a simple component this is $\nu(c\eta)=d$, the strongest form
of being singular there.

\begin{proposition}[Dimension from simple components]
\label{prop:k-components}
Let $\mathbb{F}_2[G]\eta\cong M_d(\mathbb{F}_{2^f})$ be a simple component.  It contributes exactly
$2f\,\nu(a\eta)\,\nu(b\eta)$ to $k$, and $k$ is the sum of the contributions of all components.  In
particular, if $|G|$ is odd, $k=2\sum f\,\nu(a\eta)\,\nu(b\eta)$ over all components.
\end{proposition}

\begin{proof}
  \emph{Splitting over components.}  Let $\eta_1,\dots,\eta_m$ be the primitive central idempotents, so
$\eta_i\eta_j=0$ for $i\ne j$ and $\sum_i\eta_i=e$~\cite[Prop.~3.6.1]{Webb_2016_finite_group_representations}, and write $R=\mathbb{F}_2[G]$ and $R_i=R\eta_i$.  Then for any $y \in \mathbb{F}_2[G]$, we have
\begin{equation*}
  R=\bigoplus_{i=1}^m R_i,\qquad y=\sum_i y_i,\quad y_i=y\eta_i=\eta_i y ,
\end{equation*}
and in particular $a_i=a\eta_i$, $b_i=b\eta_i$. Since $\eta_i$ is central and idempotent,
$ay_i=a\eta_iy_i=a_iy_i\in R_i$ and $y_ib=y_ib_i\in R_i$, so
\begin{equation*}
  ay=\sum_i a_iy_i,\qquad yb=\sum_i y_ib_i ,
\end{equation*}
with the $i$-th term in $R_i$.  The sum is direct, so $ay=0$ if and only if $a_iy_i=0$ for every $i$, and $yb=0$ if and only if $y_ib_i=0$ for every $i$.  

The key point is that this allows us to decompose the annihilator expressions appearing in Prop.~\ref{prop:k}.
\begin{equation*}
  \begin{aligned}
    \ann_r(a)\cap\ann_l(b)&=\bigoplus_{i=1}^m\{y_i\in R_i:\ a_iy_i=0,\ y_ib_i=0\}, \\
    \ann_l(a)\cap\ann_r(b)&=\bigoplus_{i=1}^m\{y_i\in R_i:\ y_ia_i=0,\ b_iy_i=0\},
  \end{aligned}
\end{equation*}
so $t_1=\sum_i t_1^{(i)}$ and $t_2=\sum_i t_2^{(i)}$, with $t_1^{(i)},t_2^{(i)}$ the dimensions of the
$i$-th summands, and by Prop.~\ref{prop:k}
\[
  k=\sum_{i=1}^m k_i,\qquad k_i=t_1^{(i)}+t_2^{(i)} .
\]
This step uses only that the $\eta_i$ are central, orthogonal idempotents summing to $e$, so it holds for non-simple components too.

\emph{One simple component.}  We now examine the contribution of a single simple component. Let us denote the aforementioned isomorphism as  $\varphi:R_i\to M_d(\mathbb{F}_{2^f})$. Let $A = \varphi(a_i), B= \varphi(b_i), Y = \varphi(y_i)$. By isomorphism, we rewrite the former expression of $t_1^{(i)}$ as the $\mathbb{F}_2$-dimension of  
\begin{equation*}
  V_1=\{Y\in M_d(\mathbb{F}_{2^f}):\ AY=0,\ YB=0\}
\end{equation*}
Since $\varphi$ is an $\mathbb{F}_2$-linear bijection, $t_1^{(i)}=\dim_{\mathbb{F}_2}V_1$.  Note that $V_1$ itself is now a space over $\mathbb{F}_{2^f}$, so $\dim_{\mathbb{F}_2}V_1=f\,\dim_{\mathbb{F}_{2^f}}V_1$.

Let the columns of $U$ ($d\times\alpha$) be a basis of $\ker A=\{u:Au=0\}$, so $\alpha=\nu(A)$, and let the rows
of $W$ ($\beta\times d$) be a basis of the left kernel $\{w:wB=0\}$, so $\beta=\nu(B)$; for a square matrix the
left and right kernels have the same dimension $d-\rank$.  Every column of $Y$ lies in $\ker A$, so $Y=UC$ for a
unique $\alpha\times d$ matrix $C$.  Then $YB=UCB$, and since $U$ has independent columns, $YB=0$ if and only if
$CB=0$, i.e.\ every row of $C$ lies in the left kernel of $B$, so $C=DW$ for a unique $\alpha\times\beta$
matrix $D$.  Hence
\begin{equation*}
  V_1=\{\,UDW:\ D\in\mathbb{F}_{2^f}^{\alpha\times\beta}\,\},
\end{equation*}
every such matrix lies in $V_1$, since $AU=0$ and $WB=0$, and $D\mapsto UDW$ is injective, so $\dim_{\mathbb{F}_{2^f}}V_1=\alpha\beta$ and $t_1^{(i)}=f\,\nu(A)\,\nu(B)$, with $\nu(A)=\nu(a\eta_i)$ and $\nu(B)=\nu(b\eta_i)$ by the definition of $\nu$.  An analogous argument holds for $t_2^{(i)}$. Summing gives the first claim. When $|G|$ is odd, every component is simple, which gives the second claim.
\end{proof}

Note that a version of this statement for semisimple bivariate bicycle codes has previously appeared in the literature \cite[Thm.~81]{Sabo_arXiv_2026_semisimple_bb}, \cite[Cors.~2.12, 4.4]{Eberhardt_arXiv_2024_bb_logical_operators}.

For even $|G|$ the non-simple components also contribute to $k$, with no closed form; the census computes $k$ exactly from annihilators (Prop.~\ref{prop:k}), and Prop.~\ref{prop:ideal-target} still gives a lower bound when both blocks vanish on a non-simple component.

\subsection{From algebra to statistics}

Prop.~\ref{prop:k-components} says exactly which pairs have $k>0$, in terms of the images of $a$ and $b$ in the
components. However, a sampler chooses a block by its support, not by its images.  A uniformly random element of $\mathbb{F}_2[G]$
would have exactly uniform, independent images, since $\mathbb{F}_2[G]\cong\prod_i R_i$; our blocks are uniform only among elements of weight $w$, but empirically we find that their images in each nontrivial component behave much like uniform random matrices, for which calculations are much easier. A uniformly random $d\times d$ matrix over $\mathbb{F}_{2^f}$ is invertible when its columns are independent: the $i$-th column must avoid the span of the previous $i-1$, which has $2^{f(i-1)}$ of the $2^{fd}$
vectors.  So it is singular with probability
\[
  p(d,f)=1-\prod_{i=1}^{d}\bigl(1-2^{-fi}\bigr),
\]
which is $2^{-f}$ for a field component (a component with $d=1$) and e.g.\ $0.67$
for $M_3(\mathbb{F}_2)$.  Treating the images as uniform and independent across components, we can heuristically forecast $k > 0$ from the group alone,
$P(k>0)\approx1-\prod\bigl(1-p(d,f)^2\bigr)$ over the nontrivial components; the trivial component is left out because it needs no forecast.\footnote{If
one block has odd weight (always the case at odd check weight), that block is invertible there and it contributes
nothing; if both have even weight, both vanish there and every pair has $k\ge2$.  (For even $|G|$ there is no
trivial component: $\eta_G=\sum_g g$ squares to $0$, and the parity of a block is read off a non-simple component.)} The forecast matches the measured rate at correlation $0.94$ over $306$ non-abelian groups of order $120$--$150$ (Fig.~\ref{fig:theory-k}).  It ignores
non-simple components, which is why it underestimates $P(k>0)$ as $|O_2(G)|$ grows. While this is promising, it ultimately only provides a heuristic -- this is true for each of the strategies outlined below. While in theory they each sound like plausible routes to targeted sampling, in practice we timed them against a uniform sampling strategy before committing.  

\begin{figure}[tbp]
\centering
\includegraphics[width=0.5\linewidth]{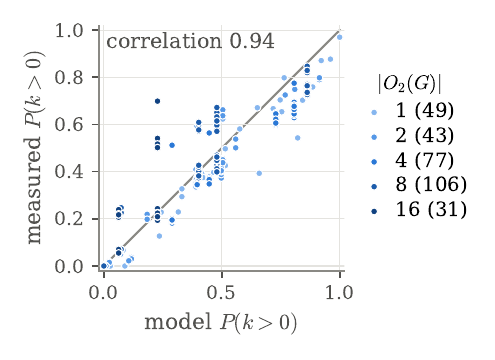}
\caption{What the group algebra predicts about $k$.  For each of $306$ non-abelian groups of order
$120$--$150$, the measured fraction of uniformly random block pairs with $k>0$ ($1000$ pairs per group, block
weights $w_a=4$, $w_b=3$), against the forecast from the group's components alone, which treats the images of the blocks in the components as independent
uniform random matrices.  The correlation, $0.94$, is surprisingly high for so crude a model, but the forecast remains a heuristic: it is not a substitute for computing $k$.  Colors give $|O_2(G)|$, the order of the largest normal $2$-subgroup.
The points above the diagonal, at large $|O_2(G)|$, are groups where the non-simple components, which the forecast
leaves out, contribute much of $k$: a normal $2$-subgroup acts trivially on every simple module in characteristic two, so at most $|G|/|O_2(G)|$ dimensions of $\mathbb F_2[G]$ lie in simple components.}
\label{fig:theory-k}
\label{fig:theory}
\end{figure}

\subsection{Targeted sampling via components}
\label{app:targeted-sampling}

\emph{Cost.}  Targeted sampling can only work if the filter is cheaper than the witness-first ladder. In this case it is: the components, their $d$ and $f$, and their central idempotents depend on the group only and are computed once per group, by linear algebra over $\mathbb{F}_2$. In addition, everything that depends on a block is computed only once per block (as compared to once per pair/code), e.g.\ its annihilators.

\paragraph{Three ways to guarantee $k$.}

All three statements below read off the contribution $2f\,\nu(a\eta)\,\nu(b\eta)$ of
Prop.~\ref{prop:k-components}, and differ only in how much nullity they force:
\begin{center}
\small
\begin{tabular}{@{}llccc@{}}
\toprule
statement & condition & nullity forced & guaranteed contribution & linear \\
\midrule
targeting a component & $a\eta$, $b\eta$ singular & $\nu\ge1$ & $2f$ & no \\
targeting a rank & $au=0$, $u'b=0$ & $\nu(a\eta)\ge r$, $\nu(b\eta)\ge r'$ & $2f\,rr'$ & yes \\
targeting an ideal & $a\eta=b\eta=0$ & $\nu=d$ & $2fd^2$ & yes \\
\bottomrule
\end{tabular}
\end{center}
In principle any of them can drive a sampler, for any group.  In a field component the three conditions
coincide, since a singular element of a field is zero; for abelian groups, where every component is a field, they collapse into the same 
sampler, the \textit{shared-factor sampler}, which is the one we used.  

For non-abelian groups they differ,
and empirically they worked out quite differently: in a weight-seven control at equal time, the rank condition found
$44$ new front codes against $11$ for uniform draws and $18$ for the ideal condition, whose hits were mostly codes
that uniform draws also find.  The later non-abelian runs therefore used rank targeting alongside uniform draws. 

One conjecture for this difference between abelian and non-abelian groups is that uniform sampling for each kind of group has different failure points. In particular, rejected abelian codes in uniform sampling tended to have $k = 0$, whereas most non-abelian codes tended to have non-zero $k$ already.

\begin{proposition}[Targeting a component]
\label{prop:targeting}
Let $\mathbb{F}_2[G]\eta\cong M_d(\mathbb{F}_{2^f})$ be a simple component, and choose $a,b$ with $a\eta$ and
$b\eta$ both singular.  Then $k\ge2f$.
\end{proposition}

\begin{proof}
This is Prop.~\ref{prop:k-components} with $\nu(a\eta),\nu(b\eta)\ge1$.
\end{proof}

As mentioned earlier, for a field component vanishing and singularity coincide; in a matrix component they differ.  Vanishing is a linear
condition on each block, so it can be imposed exactly, and it gives more than one component's worth of $k$:

\begin{proposition}[Targeting an ideal]
\label{prop:ideal-target}
Let $\eta$ be a central idempotent of $\mathbb{F}_2[G]$ and suppose $a\eta=b\eta=0$.  Then
\[
  k\;\ge\;2\dim_{\mathbb{F}_2}\mathbb{F}_2[G]\eta .
\]
For a simple component $\mathbb{F}_2[G]\eta\cong M_d(\mathbb{F}_{2^f})$ this is $2fd^2$; for abelian $G$ it is
the shared-factor bound $k\ge2f$.
\end{proposition}

\begin{proof}
For $y\in\mathbb{F}_2[G]\eta$ we have $y=y\eta=\eta y$, so $ay=(a\eta)y=0$ and $yb=y(\eta b)=y(b\eta)=0$.
Hence $\mathbb{F}_2[G]\eta\subseteq\ann_r(a)\cap\ann_l(b)$, and likewise
$\mathbb{F}_2[G]\eta\subseteq\ann_l(a)\cap\ann_r(b)$, so $t_1,t_2\ge\dim\mathbb{F}_2[G]\eta$ in
Prop.~\ref{prop:k}.  This only uses that $\eta$ is central and idempotent, so it also holds for groups of even order.
\end{proof}

Between singular and zero lie the intermediate ranks, which can also be imposed by linear conditions:

\begin{corollary}[Targeting a rank]
\label{cor:rank-target}
Let $\mathbb{F}_2[G]\eta\cong M_d(\mathbb{F}_{2^f})$ be a simple component, and let
$u,u'\in\mathbb{F}_2[G]\eta$ have ranks $r,r'$ as matrices.  If $au=0$ and $u'b=0$, the component
contributes at least $2f\,r\,r'$ to $k$.
\end{corollary}

\begin{proof}
$au=0$ puts the column space of $u$ in the kernel of $a\eta$, so $\nu(a\eta)\ge r$; $u'b=0$ puts the row
space of $u'$ in the left kernel of $b\eta$, so $\nu(b\eta)\ge r'$.  Prop.~\ref{prop:k-components}
gives the contribution $2f\,\nu(a\eta)\nu(b\eta)$.
\end{proof}

Purely as an empirical observation, a pattern where one block vanishing and the other singular in
the same component appears to be typical of the best codes at weights seven and nine, rather than both blocks vanishing.

\paragraph{The samplers.}

We can turn each statement into a targeted sampling strategy. First, we choose a target. This is a set of linear conditions on the blocks, together with the lower bound on $k$ that every pair satisfying them has.\footnote{In principle because we fixed block weights, this is a small combinatorial search rather than linear algebra.} This is chosen so that this bound lies in a desired range. Then we draw blocks satisfying the conditions.

\begin{itemize}[leftmargin=14pt,itemsep=1pt,topsep=2pt]
\item \emph{Abelian groups.}  Every component is a field.  The target is a set $S$ of components, and both blocks
are drawn in the ideal of elements vanishing on all of them, so $k\ge2\sum_{i\in S}f_i$
(Prop.~\ref{prop:ideal-target}).  For cyclic groups this coincides with the divisor-first rule of earlier cyclic
searches~\cite{Panteleev_Quantum_2021_bposd,Lu_arXiv_2026_divisor_driven}. We refer to this as the \textit{shared-factor sampling} strategy.
\item \emph{Non-abelian groups.}  The target is a pattern of ranks $(r,r')$ in one or a few simple components.
Elements $u,u'$ of those ranks are drawn in the components, then $a$ with $au=0$ and $b$ with $u'b=0$
(Cor.~\ref{cor:rank-target}).  A rank equal to $d$ is imposed with $u=\eta$ and makes the block vanish there;
ranks below $d$ are practical only over small fields, since singular elements of large fields are rare to draw. We refer to this as the \textit{rank-targeting} strategy.
\end{itemize}
In both cases targets are drawn uniformly among the admissible ones, and uniform draws were kept alongside, since
targeting selects which codes are seen.

\section{The literature set and the censuses}
\label{app:literature}

In this paper, we primarily compare our codes against a fixed set of published codes. Furthermore, our results come from censuses whose scope is stated, and then left to run, often overnight. In this section, we describe the literature set and census specifications.

\subsection{The literature set}
\label{app:literature-set}

The literature set is a set of the `best' published codes, as best as we could find. Each record gives a code's
parameters $(n,k,d)$, check weight, qubit degree, construction, and the place
in the source where the code is given (table or page).  The comparisons in
this paper use v.~8 (see Ref.~\dburl{} for versioning details), which has $2{,}611$ records from $47$ papers
(Table~\ref{tab:literature-sources}).

\emph{Selection rules.}
\begin{enumerate}
\item \emph{Scope.}  Explicit qubit stabilizer codes, CSS or not, with
  $n\le1000$ and maximum check weight at most $10$, with larger weights (up to $12$) where we found them, published by the cut-off
  of 29 September 2026. 
\item \emph{Front test.}  A sweep of the papers in scope read $4{,}063$ code
  rows from $41$ papers.  A code entered the set only if no code already
  in it dominated the new code, in at least one cumulative class of check
  weight at most $c$, for $c=6,8,9,10,12$, with published distances taken at
  face value.  Dominated codes were left out; a code dominated only by
  another new code whose distance was an upper bound or an estimate was
  kept.
\item \emph{One archive in full.}  The data archive of Lin and
  Pryadko~\cite{Lin_PRA_2024_2bga} is a census rather than a selection, and
  all $1{,}895$ of its records are in the set, dominated or not.
\item \emph{Later additions.}  V.~7 adds $47$ records published before
  the cut-off that the sweep missed: hyperbolic surface and color codes, and
  one pair-partition code of Okada and Kasai.  V.~8 adds the $29$ codes of Yang,
  Duckering and Dua~\cite{Yang_arXiv_2026_gala} with $n\le1000$, all at check weight twelve, with
  distances exact as published.  The distance floors were set in
  v.~6 and not moved (Sec.~\ref{app:anchors}).
  
\item \emph{Exclusions.}  $111$ records are kept in the set but excluded from
  the fronts: duplicates, codes outside the scope of the comparisons, and a
  few whose published parameters did not check out.
\end{enumerate}
The fronts in this paper use the part of the set with check weight five to
nine and $n\le320$.  The set was frozen on the cut-off date and is
read-only: corrections are made in new dated versions, never in a frozen
file, and the checksum of every version is recorded.

\emph{How the set was assembled.}
\begin{enumerate}
\item The sweep above was carried out by several independent readers, all of them LLM agents (Claude), each
  recording every code with its location in the source.
\item An adversarial audit was carried out, again by LLM agents, of every record on a front against the source
  PDF.
\item The anchor codes (Sec.~\ref{app:anchors}), as well as a selection of standout codes from the literature, were checked by hand against their sources. The sources were also hand-checked by the author against papers they were familiar with. Most of the anchor codes were already in the literature set long before it was frozen.
\item Data archives downloaded and every code rebuilt from its listed
  blocks, rather than transcribed.  The largest is the archive of Lin and
  Pryadko~\cite{Lin_PRA_2024_2bga}, $1{,}895$ records ($73\%$ of the set), each
  rebuilt in GAP.
\item New arXiv postings were monitored up to the cut-off.
\end{enumerate}
Public leaderboards and code databases that are not published papers were
consulted for awareness only and are not in the set.

\emph{Distances.}  We take $d=\min(d_X,d_Z)$.  Where we could, we recomputed
published distances exactly (Appendix~\ref{sec:distance}).  Every archive
record whose published distance is a randomized upper bound of at most $15$
was settled: $1{,}539$ are exact at the published value, and $121$ are lower.
The lowered ones are non-abelian codes for which the archive lists the
distance of one sector, not the minimum over both.  Six codes of Aydin, Tamo
and Barg~\cite{Aydin_arXiv_2026_coset_codes} were made exact at their published values,
with witnesses; four codes remain open with proven lower bounds: $[[252,16]]$ ($d\ge19$) and
$[[306,22]]$ ($d\ge23$) of Aydin et al., and $[[288,20]]$, $[[288,14]]$ ($d\ge20$) of Symons
et al.~\cite{Symons_arXiv_2025_bb_covering_graphs}.  Every lowered distance comes with a witness, a logical operator of
the stated weight.  Of the records on the fronts, $1{,}689$ distances are exact
by our computation and $393$ exact as published; $414$ are published upper
bounds or estimates that we did not settle, and $4$ are lower bounds.

\emph{Two conventions.}  On a \emph{strict} front, a code whose distance is
only an upper bound may sit on the front but never removes another code; our
codes enter with exact distances only.  On a \emph{published-$d$} front,
every published distance is taken at face value.  The figures draw the strict
front and the published-$d$ front together.

\emph{Check weight.}  Records are filed by their maximum check weight.  Where a source
reports a combined weight instead (Qian and Li's ``overall weight'' counts check weight and
qubit degree together~\cite{Qian_arXiv_2026_multiagent_qldpc}), two-block codes were rebuilt and
agree, since there the two coincide; five records of other shapes are filed at the weight as
reported, which is an upper bound on their check weight.  None of them sets a front point or an
anchor.

\begin{table}[tbp]
\centering
\footnotesize
\begin{tabular}{@{}llr@{}}
\toprule
source & codes & records \\
\midrule
\multicolumn{3}{@{}l}{\emph{Two-block and bicycle codes}} \\
Lin--Pryadko, paper and data archive \cite{Lin_PRA_2024_2bga} & two-block, abelian and non-abelian & $1{,}909$ \\
Bravyi et al. \cite{Bravyi_Nature_2024_bb_codes} & bivariate bicycle & $7$ \\
Liang et al. \cite{Liang_PRXQ_2025_twisted_tori} & twisted tori & $79$ \\
Liang et al. \cite{Liang_arXiv_2025_self_dual_bb} & self-dual bivariate bicycle & $6$ \\
Aydin--Tamo--Barg \cite{Aydin_arXiv_2026_coset_codes} & coset and cover two-block & $66$ \\
Symons et al. \cite{Symons_arXiv_2025_bb_covering_graphs} & bivariate-bicycle covers & $33$ \\
Cruz-Benito et al. \cite{CruzBenito_arXiv_2026_llm_bb_search} & bivariate bicycle (search) & $29$ \\
Qian--Li \cite{Qian_arXiv_2026_multiagent_qldpc} & bivariate bicycle (search) & $22$ \\
Tripier et al. \cite{Tripier_arXiv_2026_walking_cat} & generalized and bivariate bicycle & $23$ \\
Wang et al. \cite{Wang_arXiv_2024_coprime_bb} & coprime bivariate bicycle & $3$ \\
Lu et al. \cite{Lu_arXiv_2026_divisor_driven} & generalized bicycle & $2$ \\
Lu et al. \cite{Lu_arXiv_2026_weight8_bb} & bivariate bicycle & $1$ \\
Panteleev--Kalachev \cite{Panteleev_Quantum_2021_bposd} & generalized bicycle & $4$ \\
Mostad et al. \cite{Mostad_arXiv_2025_gb_finite_length} & generalized bicycle & $2$ \\
Rabeti--Mahdavifar \cite{Rabeti_arXiv_2026_univariate_bicycle,Rabeti_arXiv_2025_list_decoding_bicycle} & univariate bicycle & $5$ \\
Voss et al. \cite{Voss_arXiv_2024_multivariate_bicycle} & multivariate bicycle & $1$ \\
Galimova et al. \cite{Galimova_arXiv_2026_trivariate_bicycle} & trivariate bicycle & $1$ \\
Xu et al. \cite{Xu_arXiv_2026_space_group_codes} & space-group (non-abelian) & $6$ \\
Ye et al. \cite{Ye_arXiv_2025_trapped_ion_chains} & trapped-ion chain codes & $1$ \\
Mian et al. \cite{Mian_arXiv_2026_multivariate_multicycle} & multivariate multicycle & $44$ \\
Kovalev \cite{Kovalev_arXiv_2026_floquet_multicycle} & Floquet multicycle & $14$ \\
Khesin--Lu \cite{Khesin_arXiv_2026_mirror_codes} & mirror (non-CSS) & $40$ \\
\midrule
\multicolumn{3}{@{}l}{\emph{Other translation-invariant codes}} \\
Zhang et al. \cite{Zhang_arXiv_2026_routing_codes} & routing codes & $22$ \\
Leroux--Iverson \cite{Leroux_arXiv_2025_romanesco} & Romanesco (non-CSS) & $34$ \\
Steffan et al. \cite{Steffan_arXiv_2025_tile_codes} & tile codes & $3$ \\
Gu et al. \cite{Gu_arXiv_2026_planar_grid} & directional tile codes & $4$ \\
\midrule
\multicolumn{3}{@{}l}{\emph{Lifted-product, pair-partition and high-rate codes}} \\
Lee et al. \cite{Lee_arXiv_2026_ultra_high_rate} & pair-partition and halved & $74$ \\
Okada--Kasai \cite{Okada_arXiv_2026_pair_partition,Okada_arXiv_2026_canonical_logical_bases} & pair-partition & $46$ \\
Okada--Kasai \cite{Okada_arXiv_2026_high_girth_cpm_lifts} & circulant lifts & $8$ \\
Okada--Kasai \cite{Okada_arXiv_2026_two_branch} & finite-field bases & $6$ \\
Hong et al. \cite{Hong_arXiv_2026_rate_fifth} & five-block lifted product & $17$ \\
Bhardwaj et al. \cite{Bhardwaj_arXiv_2026_high_rate_processors} & mitten & $15$ \\
Yang et al. \cite{Yang_arXiv_2026_gala} & GALA (self-dual and rate one-half) & $29$ \\
Lu et al. \cite{LuZ_arXiv_2026_ultra_low_overhead} & quasi-cyclic & $3$ \\
Hirasaki--Lee \cite{Hirasaki_arXiv_2026_lifting_lp} & lifted products (two-block cases) & $1$ \\
Malcolm et al. \cite{Malcolm_arXiv_2025_shyps} & subsystem hypergraph product & $2$ \\
Scruby et al. \cite{Scruby_arXiv_2024_radial_codes} & radial codes & $1$ \\
Xu et al. \cite{Xu_arXiv_2024_homological_product} & homological product & $1$ \\
\midrule
\multicolumn{3}{@{}l}{\emph{Hyperbolic surface and color codes}} \\
Breuckmann et al. \cite{Breuckmann_IEEETIT_2016_hyperbolic_surface_codes,Breuckmann_arXiv_2017_semi_hyperbolic,Breuckmann_thesis_2018_homological_codes} & hyperbolic and semi-hyperbolic surface & $23$ \\
Conrad et al. \cite{Conrad_PhilTransA_2018_stellated_dodecahedron} & star-polyhedron surface & $6$ \\
Vittal et al. \cite{Vittal_MICRO_2024_flag_proxy_networks} & hyperbolic surface and color & $6$ \\
Galimova et al. \cite{Galimova_arXiv_2026_erasure_hyperbolic} & hyperbolic surface & $4$ \\
da Silva--Soares \cite{daSilva_arXiv_2018_hyperbolic_color} & hyperbolic color & $7$ \\
\midrule
total, $47$ papers & & $2{,}611$ \\
\bottomrule
\end{tabular}

\caption{The literature set (v.~8, cut-off 29 September
2026) by source, grouped by the codes they give.  Records: codes of the
paper that passed the selection rules, so a paper's count is not the number
of codes it published; codes found in several papers are counted once.}
\label{tab:literature-sources}
\end{table}

\subsection{Anchors and distance floors}
\label{app:anchors}

Each search was set against an \emph{anchor}: for each check weight and band
of $n$, the literature code of exactly that check weight, with an exact
distance, and with the largest $kd^2/n$ (Table~\ref{tab:anchors}).  The
bands are $n\le130$, $130<n\le170$, $170<n\le220$ and $220<n\le320$.  The
anchor's distance minus two sets the band's \emph{distance floor}: a code of
ours counts as a result in the band only if its distance reaches the floor.
Floors never decrease with $n$, and a band without an anchor inherits the
floor of the band below.  The floor keeps codes of high rate and trivial
distance from counting as progress on a front.

The anchors and floors were fixed on version~6 of the literature set and
were not moved when versions~7 and~8 added records, so that no result changed its
standing after the fact.

\subsection{Census specifications}
\label{app:censuses}

Our censuses fixed a check weight, a split of the
weight between the two blocks, a set of groups from the SmallGroups
library~\cite{GAP4,Besche_IJAC_2002_millennium}, and a target: either a bar on
$kd^2/n$ (\emph{BPT mode}) or a place on the current $(n,k,d)$ front above the
distance floor (\emph{frontier mode}).  No census went beyond $n=320$.  Every census records, for each group
and split, its wall time, its number of worker processes, and every code
that survived its screens with a witness for its distance.

\emph{Exhausted and sampled.}  A census is \emph{exhausted} when it drew every
pair of blocks up to the equivalences of Appendix~\ref{sec:equivalences}, and
its classes come with a proof of completeness (Appendix~\ref{sec:verification}).  Only
exhausted censuses support statements of the form ``no code beats'', and
those statements are made only at the census's bar or floor.  Every other
census is \emph{sampled}: uniform (random pairs of supports containing the
identity), targeted (pairs drawn inside chosen shared-factor ideals, or with prescribed ranks in chosen components, App.~\ref{app:targeted-sampling}), structured (covers of chosen base codes, drawn on purpose), or a top-up (more time on the groups a previous
census found most productive).  For sampled censuses we give coverage as
draws over raw pairs where it can be computed, and claim existence only.

\emph{Compute.}  All censuses ran on the laptops described in the
introduction.  Every census, exhausted or sampled, is listed with its groups, splits, coverage, target and best codes in the supplementary note in the database repository~\dburl.

\subsection{The code database}
\label{app:database}

Besides our codes, the literature set and the coverage records (Sec.~\ref{sec:database}), the database~\dburl{} holds:
\begin{enumerate}
\item for each code of ours, its canonical key and its verification record: which engines and
  independent checks confirm its distance, with dates (Appendix~\ref{sec:verification}).  Verification
  continues after posting, and this record is updated in the database, not in the paper;
\item the fronts and the anchors of this paper, rebuilt from the above by a script;
\item a supplementary note with the statements and proofs of the screens we considered but do not run
  (Appendix~\ref{sec:screens-unused}).
\end{enumerate}

\emph{Selection rules for our codes.}  The database is not a list of every
code our censuses drew.  A code was kept only if it passed the target of its
census and then received an exact distance:
\begin{enumerate}
\item \emph{BPT mode}: a code is kept when
  its $kd^2/n$ exceeds the census's bar.  Bars were set at or somewhat below
  the literature anchor of the band (Sec.~\ref{app:anchors}), for example
  $9$ at weight six, where the anchor is $12$, and $16.2$ at weight eight for
  $n\le100$, where it is $19.2$; so the database also holds codes below the
  anchors, not only codes that beat them.
\item \emph{Frontier mode}: a code is kept when no point of
  the current front of its check weight dominates it, keeping at most $20$
  inequivalent classes with the same $(n,k)$.  The front grew during a run,
  so a code kept early may be dominated by one found later; the fronts of
  this paper are recomputed from the whole database.
\item \emph{Exact distances only}: every stored code has $d_X$ and $d_Z$ exact,
  so its witness has weight exactly $d$.
\item \emph{Rejected codes are not stored as codes.}  Most classes are rejected
  for $k=0$, a disconnected Tanner graph or a block that is a unit, and the
  rest for their distance, by a screen that found a logical below the target.
  For exhausted censuses the coverage records keep one row per class with its
  status, $k$ and bounds on $d$; the distances of rejected classes there are
  upper bounds, often loose, and are labeled as such.  For sampled censuses
  only the counts of draws are kept.
\end{enumerate}

\FloatBarrier
\bibliographystyle{alpha}
\bibliography{references}

\clearpage
\section{Code tables}
\label{app:tables}

In this appendix, we provide the tables of anchor codes (Table~\ref{tab:anchors}), a selection of our codes (Table~\ref{tab:showcase}), and the ten best weight-six codes with $n\le200$ (Table~\ref{tab:w6-top10}). The full specification of every code listed here, including its group with generators and its blocks $a$ and $b$, is available in the database~\dburl.

\begin{table}[H]
\centering
\small
\begin{tabular}{@{}rlrrlcl@{}}
\toprule
& code & $kd^2/n$ & classes & groups & window & in the literature \\
\midrule
1 & $[[192,12,14]]$ & 12.2 & 1 & $C_{3}\rtimes (Q_{8}\rtimes C_{4})$ & -- & \cite{Lin_PRA_2024_2bga} \\
2 & $[[144,12,12]]$ & 12.0 & 1 & $C_{12}\times C_{6}$ & 3 & \cite{Bravyi_Nature_2024_bb_codes} \\
3 & $[[192,16,12]]$ & 12.0 & 1 & $C_{3}\rtimes ((C_{4}\times C_{4})\rtimes C_{2})$ & -- & \cite{Lin_PRA_2024_2bga} \\
4 & $[[200,16,12]]$ & 11.5 & 1 & $C_{10}\times D_{10}$, +1 more & -- & \cite{Lin_PRA_2024_2bga} \\
5 & $[[180,8,16]]$ & 11.4 & 3 & $C_{30}\times C_{3}$, +2 more & 4 & \cite{Liang_PRXQ_2025_twisted_tori,Aydin_arXiv_2026_coset_codes,Lin_PRA_2024_2bga} \\
6 & $[[182,6,18]]$ & 10.7 & 1 & $C_{91}$ & 4 & \cite{Liang_PRXQ_2025_twisted_tori} \\
7 & $[[192,8,16]]$ & 10.7 & 38 & $C_{48}\times C_{2}$, +31 more & 3 & \cite{Liang_PRXQ_2025_twisted_tori} (1 class) \\
8 & $[[186,10,14]]$ & 10.5 & 7 & $C_{93}$ & 3 & \cite{Liang_PRXQ_2025_twisted_tori,Symons_arXiv_2025_bb_covering_graphs,Aydin_arXiv_2026_coset_codes} \\
9 & $[[198,8,16]]$ & 10.3 & 2 & $C_{33}\times C_{3}$ & 4 & \cite{Liang_PRXQ_2025_twisted_tori} (1 class) \\
10 & $[[168,12,12]]$ & 10.3 & 1 & $C_{7}\times A_{4}$ & -- & new \\
\bottomrule
\end{tabular}

\caption{The ten best check-weight-six two-block codes with $n\le200$ by
$kd^2/n$, distinct parameters, all of split $3{+}3$.  The censuses cover every
group with $n=2|G|\le200$ (all abelian groups, all non-abelian groups of order
at most $100$; over $2$-groups an odd block is a unit, so $k=0$) at bar
$kd^2/n>9$, below every entry here.  The $4{+}2$ censuses found nothing above
the bar, and a $5{+}1$ split has a unit block, so $k=0$.  Classes: inequivalent codes with these parameters.  Window as in
Table~\ref{tab:showcase}, for the abelian classes.  In the literature: every paper that gives a code with these parameters; (1 class): only one of the classes is published.}
\label{tab:w6-top10}
\end{table}

\begin{table}[tbp]
\centering
\small
\begin{tabular}{@{}clrrlcr@{}}
\toprule
weight & $n$ band & anchor & $kd^2/n$ & family & 2D & floor \\
\midrule
5 & $\le130$ & $[[60,4,8]]$ & 4.3 & Lin--Pryadko archive \cite{Lin_PRA_2024_2bga} & yes & 6 \\
 & 130--170 & $[[168,4,14]]$ & 4.7 & Lin--Pryadko archive \cite{Lin_PRA_2024_2bga} & yes & 12 \\
 & 170--220 & $[[192,4,16]]$ & 5.3 & Lin--Pryadko archive \cite{Lin_PRA_2024_2bga} & no & 14 \\
 & 220--320 & -- & & none in the literature & & 14 \\
\midrule
6 & $\le130$ & $[[120,8,12]]$ & 9.6 & twisted torus \cite{Liang_PRXQ_2025_twisted_tori} & yes & 10 \\
 & 130--170 & $[[144,12,12]]$ & 12.0 & BB \cite{Bravyi_Nature_2024_bb_codes} & yes & 10 \\
 & 170--220 & $[[192,12,14]]$ & 12.2 & Lin--Pryadko archive \cite{Lin_PRA_2024_2bga} & no & 12 \\
 & 220--320 & $[[254,14,16]]$ & 14.1 & twisted torus \cite{Liang_PRXQ_2025_twisted_tori} & yes & 14 \\
\midrule
7 & $\le130$ & $[[126,10,13]]$ & 13.4 & two-block cover \cite{Aydin_arXiv_2026_coset_codes} & yes & 11 \\
 & 130--170 & $[[144,12,12]]$ & 12.0 & Lin--Pryadko archive \cite{Lin_PRA_2024_2bga} & yes & 11 \\
 & 170--220 & $[[196,12,17]]$ & 17.7 & Lin--Pryadko archive \cite{Lin_PRA_2024_2bga} & yes & 15 \\
 & 220--320 & $[[288,16,18]]$ & 18.0 & Qian-Li multi-agent BB \cite{Qian_arXiv_2026_multiagent_qldpc} & yes & 16 \\
\midrule
8 & $\le130$ & $[[126,20,11]]$ & 19.2 & two-block (base of a cover) \cite{Aydin_arXiv_2026_coset_codes} & yes & 9 \\
 & 130--170 & $[[168,20,14]]$ & 23.3 & two-block cover \cite{Aydin_arXiv_2026_coset_codes} & yes & 12 \\
 & 170--220 & $[[204,22,17]]$ & 31.2 & coset codes \cite{Aydin_arXiv_2026_coset_codes} & no & 15 \\
 & 220--320 & $[[232,62,12]]$ & 38.5 & CPM pair-partition \cite{Okada_arXiv_2026_pair_partition} & no & 15 \\
\midrule
9 & $\le130$ & $[[100,20,8]]$ & 12.8 & ZSZ-LP five-block \cite{Hong_arXiv_2026_rate_fifth} & no & 6 \\
 & 130--170 & $[[150,30,10]]$ & 20.0 & mitten \cite{Bhardwaj_arXiv_2026_high_rate_processors} & no & 8 \\
 & 170--220 & $[[200,40,12]]$ & 28.8 & mitten \cite{Bhardwaj_arXiv_2026_high_rate_processors} & no & 10 \\
 & 220--320 & $[[300,60,14]]$ & 39.2 & mitten \cite{Bhardwaj_arXiv_2026_high_rate_processors} & no & 12 \\
\bottomrule
\end{tabular}

\caption{The literature anchors: for each check weight and $n$ band,
the literature CSS code of that exact weight with a certain distance and the
largest $kd^2/n$, against which our searches were set.  Floor: the distance a
code needs to count as a result in that band, the anchor's distance minus two,
never lower than in a smaller band (a band without an anchor inherits it).
2D: some code with these parameters is a two-block code over an abelian group
with at most two generators.  Fourteen of the nineteen anchors are two-block
codes; the exceptions are a pair-partition code at weight eight and the
lifted-product codes that held every weight-nine band.}
\label{tab:anchors}
\end{table}

\begin{table*}[tbp]
\centering
\small
\begin{tabular}{@{}llrllccr@{}}
\toprule
$n$ band & code & $kd^2/n$ & group & split & 2D & window & anchor code \\
\midrule
\multicolumn{8}{@{}l}{\emph{check weight 5}} \\
$\le130$ & $[[114,4,11]]$ & 4.2 & $C_{57}$ & 3+2 & yes & 3 & $[[60,4,8]]$ (4.3) \\
130--170 & $[[150,4,13]]$ & 4.5 & $C_{75}$ & 3+2 & yes & 3 & $[[168,4,14]]$ (4.7) \\
170--220 & $[[216,4,16]]$ & 4.7 & $C_{108}$ & 2+3 & yes & 4 & $[[192,4,16]]$ (5.3) \\
220--320 & $[[288,4,20]]$ & 5.6 & $C_{24}\times S_{3}$ & 3+2 & no & -- & -- \\
\midrule
\multicolumn{8}{@{}l}{\emph{check weight 7}} \\
$\le130$ & $[[126,12,12]]$ & 13.7 & $C_{63}$ & 4+3 & yes & 4 & $[[126,10,13]]$ (13.4) \\
130--170 & $[[168,12,15]]$ & 16.1 & $C_{84}$ & 4+3 & yes & 4 & $[[144,12,12]]$ (12.0) \\
170--220 & $[[180,12,16]]$ & 17.1 & $C_{30}\times C_{3}$ & 4+3 & yes & 5 & $[[196,12,17]]$ (17.7) \\
220--320 & $[[252,16,20]]$ & 25.4 & $C_{7}\times D_{18}$ & 4+3 & no & -- & $[[288,16,18]]$ (18.0) \\
 & $[[294,18,20]]$ & 24.5 & $(C_{7}\times C_{7})\rtimes C_{3}$ & 4+3 & no & -- &  \\
\midrule
\multicolumn{8}{@{}l}{\emph{check weight 8}} \\
$\le130$ & $[[120,18,12]]$ & 21.6 & $C_{30}\times C_{2}$ & 4+4 & yes & 5 & $[[126,20,11]]$ (19.2) \\
130--170 & $[[170,18,16]]$ & 27.1 & $C_{85}$ & 4+4 & yes & 5 & $[[168,20,14]]$ (23.3) \\
170--220 & $[[180,20,15]]$ & 25.0 & $C_{15}\times S_{3}$ & 4+4 & no & -- & $[[204,22,17]]$ (31.2) \\
220--320 & $[[292,22,23]]$ & 39.9 & $C_{146}$ & 4+4 & yes & 6 & $[[232,62,12]]$ (38.5) \\
\midrule
\multicolumn{8}{@{}l}{\emph{check weight 9}} \\
$\le130$ & $[[102,20,10]]$ & 19.6 & $C_{51}$ & 5+4 & yes & 5 & $[[100,20,8]]$ (12.8) \\
 & $[[126,16,14]]$ & 24.9 & $C_{63}$ & 5+4 & yes & 5 &  \\
130--170 & $[[146,18,15]]$ & 27.7 & $C_{73}$ & 5+4 & yes & 5 & $[[150,30,10]]$ (20.0) \\
 & $[[150,24,13]]$ & 27.0 & $C_{15}\times C_{5}$ & 5+4 & yes & 5 &  \\
 & $[[160,24,14]]$ & 29.4 & $C_{20}\times C_{4}$ & 5+4 & yes & 5 &  \\
170--220 & $[[186,20,18]]$ & 34.8 & $C_{93}$ & 5+4 & yes & 5 & $[[200,40,12]]$ (28.8) \\
 & $[[210,28,16]]$ & 34.1 & $C_{105}$ & 5+4 & yes & 6 &  \\
220--320 & $[[294,24,21]]$ & 36.0 & $(C_{7}\times C_{7})\rtimes C_{3}$ & 6+3 & no & -- & $[[300,60,14]]$ (39.2) \\
\bottomrule
\end{tabular}

\caption{A selection of our codes.  For each check weight and $n$ band: our
code on the CSS front (or the strict-2D front) with the largest $kd^2/n$ among
those whose distance clears the band's floor, and the codes named in the text.
Group: a group with that code (for 2D codes, the abelian group of the smallest
window); split: the weights of the two blocks.  2D: some code of ours with
these parameters is a two-block code over an abelian group with at most two
generators.  Window: the smallest $w$ such that every check acts within a
$w\times w$ block of sites of a twisted torus, minimized over all our codes
with these parameters (--: no 2D layout; n/c: not computed).  Every distance
of ours is exact, with a witness.  Anchor code: the band's anchor (Table~\ref{tab:anchors}), the literature
code of the same check weight with a proven distance and the largest $kd^2/n$, with $kd^2/n$ in
brackets; --: no anchor in the band.}
\label{tab:showcase}
\end{table*}

\clearpage
\section{Supplementary figures}
\label{app:figures}

This appendix collects figures that support Sec.~\ref{sec:results} but are not needed to follow it.  All of them compare CSS codes of exact check weight with $n\le320$, ours at exact distances, with the literature set of App.~\ref{app:literature}.

\noindent Contents of this appendix:
\begin{itemize}[leftmargin=14pt,itemsep=1pt,topsep=2pt]
\item \ref{app:pareto}, Pareto fronts by check weight: weight nine (Fig.~\ref{fig:w9-frontier}), five and six (Fig.~\ref{fig:w5}), seven and eight (Fig.~\ref{fig:w7-w8}).
\item \ref{app:dvn56}, Distance against length at check weights five, six and eight (Fig.~\ref{fig:dvn-app}).
\item \ref{app:fairness}, Other figures of merit: $kd/n$ (Fig.~\ref{fig:headline-kdn}) and the overhead $n/k$ against $d$ (Fig.~\ref{fig:overhead}).
\item \ref{app:kd2n}, The figure of merit $kd^2/n$ by check weight: weights five to eight (Fig.~\ref{fig:kd2n-w56}) and nine (Fig.~\ref{fig:kd2n-w9}).
\end{itemize}

\clearpage
\subsection{Pareto fronts by check weight}
\label{app:pareto}

The $(n,k,d)$ Pareto fronts of each check weight, in four length bands and drawn as distance against rate.  Each panel shows the codes, ours and published, at the corners of the front of its band, and how the front moves when published upper bounds on the distance are taken at face value (dashed); codes that no other code dominates but that sit below a corner are left out.

\begin{figure}[H]
\centering
\includegraphics[width=\linewidth]{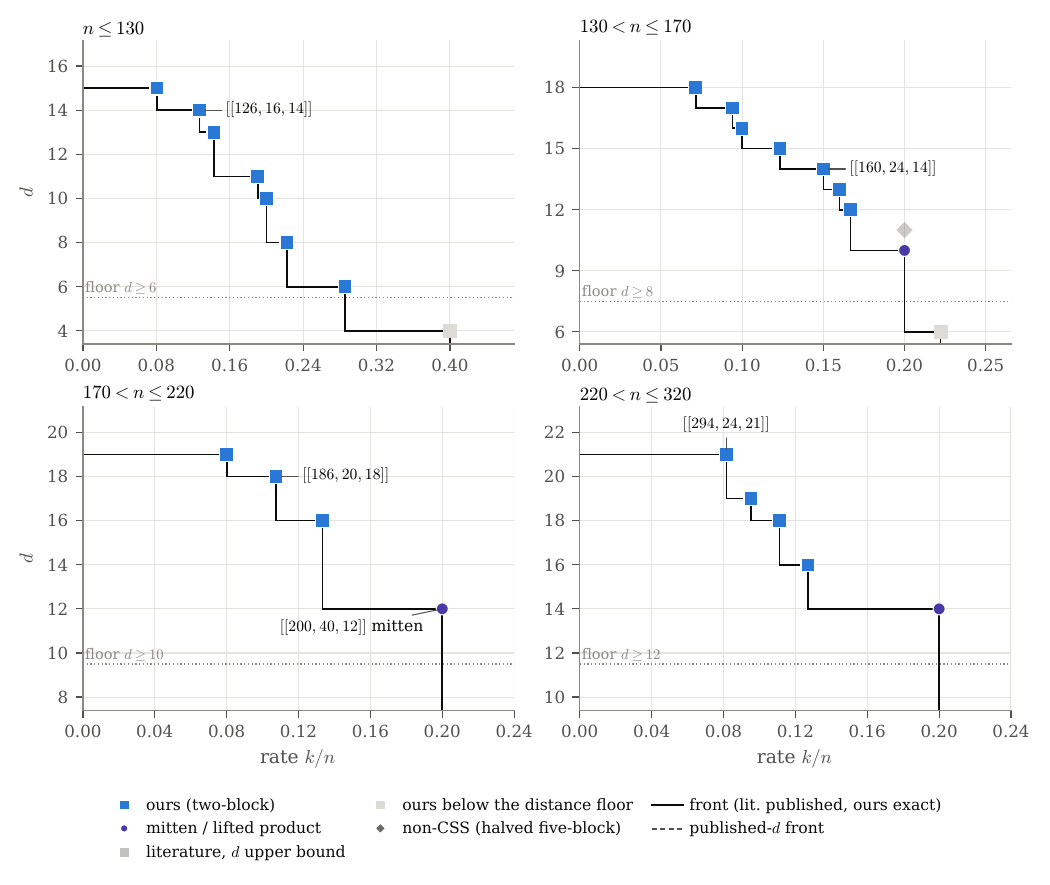}
\caption{Weight-nine CSS codes: $(n,k,d)$ Pareto fronts in four length bands,
distance against rate.  Only the codes at the corners of each front are drawn.  Our codes are blue squares, and only codes above the bar of their census are included (App.~\ref{app:database}).  Literature codes are colored by family: two-block codes are squares, other CSS codes circles, and non-CSS codes gray diamonds, lighter where the distance is only an upper bound.
Solid staircase: the front with literature codes compared at their
published distances and ours at their exact ones (only an exact
distance can remove one of our codes); dashed: every published distance taken
at face value.}
\label{fig:w9-frontier}
\end{figure}

\begin{figure}[H]
\centering
\includegraphics[width=\linewidth]{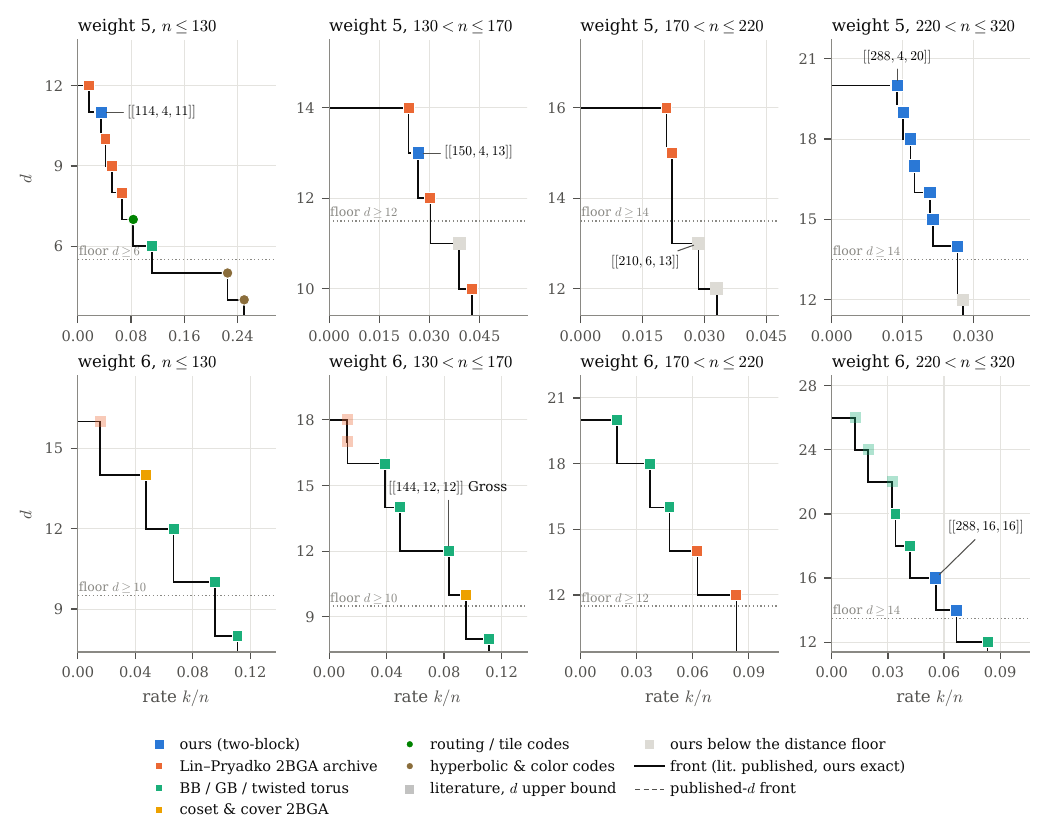}
\caption{Weight-five (top) and weight-six (bottom) CSS codes: Pareto fronts by length band, as in
Figure~\ref{fig:w9-frontier}.}
\label{fig:w5}\label{fig:w6-pareto}
\end{figure}

\begin{figure}[H]
\centering
\includegraphics[width=\linewidth]{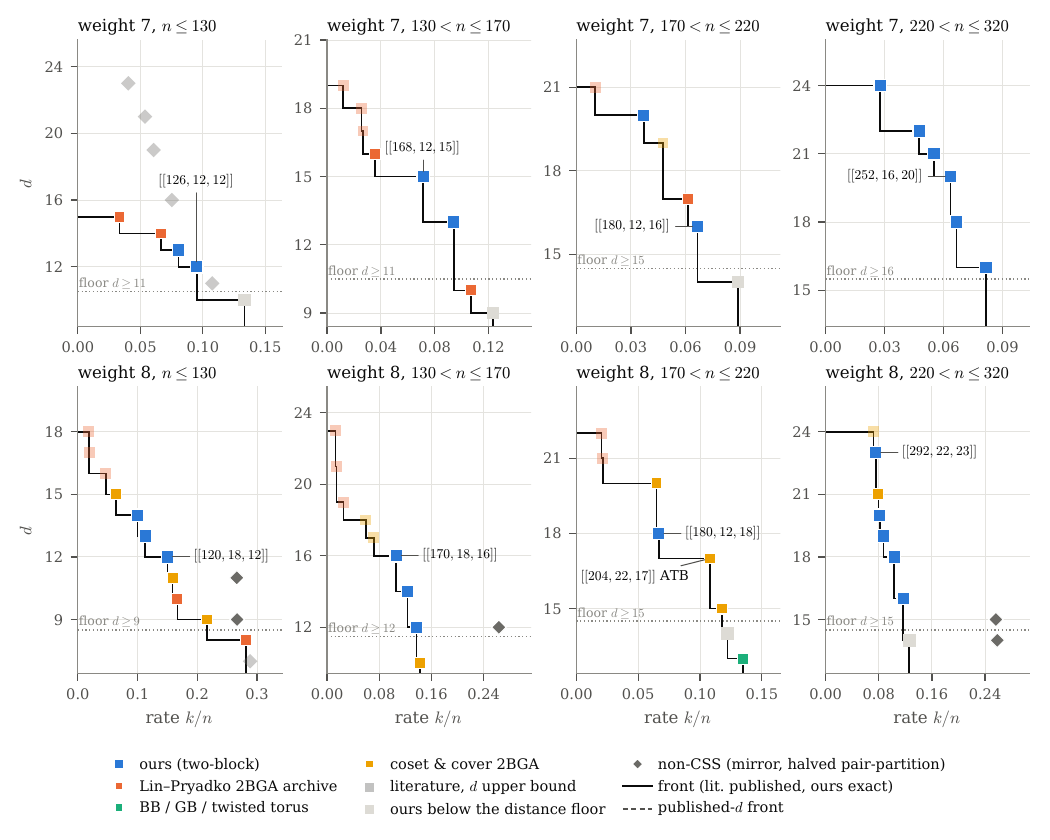}
\caption{Weight-seven (top) and weight-eight (bottom) CSS codes: Pareto fronts by length band, as in
Figure~\ref{fig:w9-frontier}.}
\label{fig:w7-w8}\label{fig:w8-pareto}
\end{figure}

\clearpage
\subsection{Distance against length at check weights five to eight}
\label{app:dvn56}

The view of Fig.~\ref{fig:dvn} for the remaining check weights: at each $n$, the best distance available with at least $K$ logical qubits, with our codes included.

\begin{figure}[H]
\centering
\includegraphics[width=0.55\linewidth]{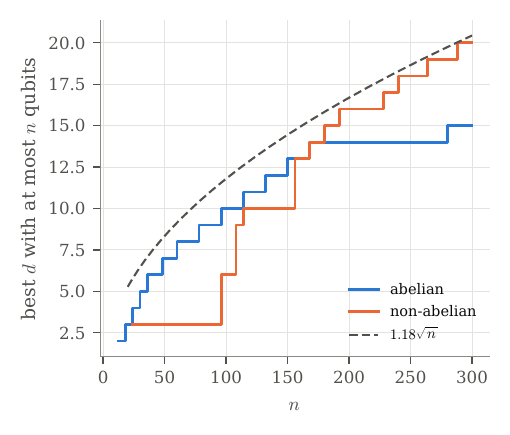}
\caption{Check weight five: the best distance with at most $n$ physical qubits and at least four logical qubits, over abelian and over non-abelian groups (codes of this work; the weight-five census is exhaustive for $n\le300$), against $d=1.18\sqrt{n}$, which no code exceeds (Sec.~\ref{sec:where}).}
\label{fig:w5-sqrt}
\end{figure}

\begin{figure}[H]
\centering
\includegraphics[width=0.6\linewidth]{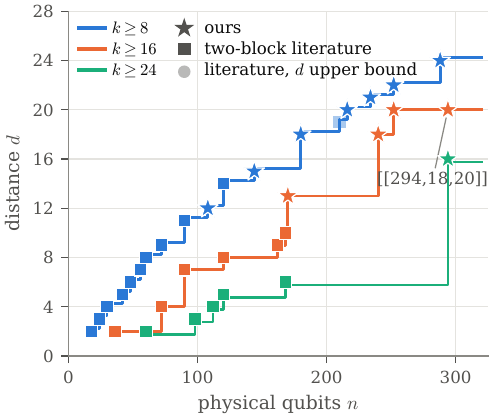}
\caption{Check weight seven: the best distance with at most $n$ physical qubits and at least $K$ logical qubits (CSS, $n\le320$); conventions as in Fig.~\ref{fig:dvn}.}
\label{fig:dvn-w7}
\end{figure}

\begin{figure}[H]
\centering
\includegraphics[width=\linewidth]{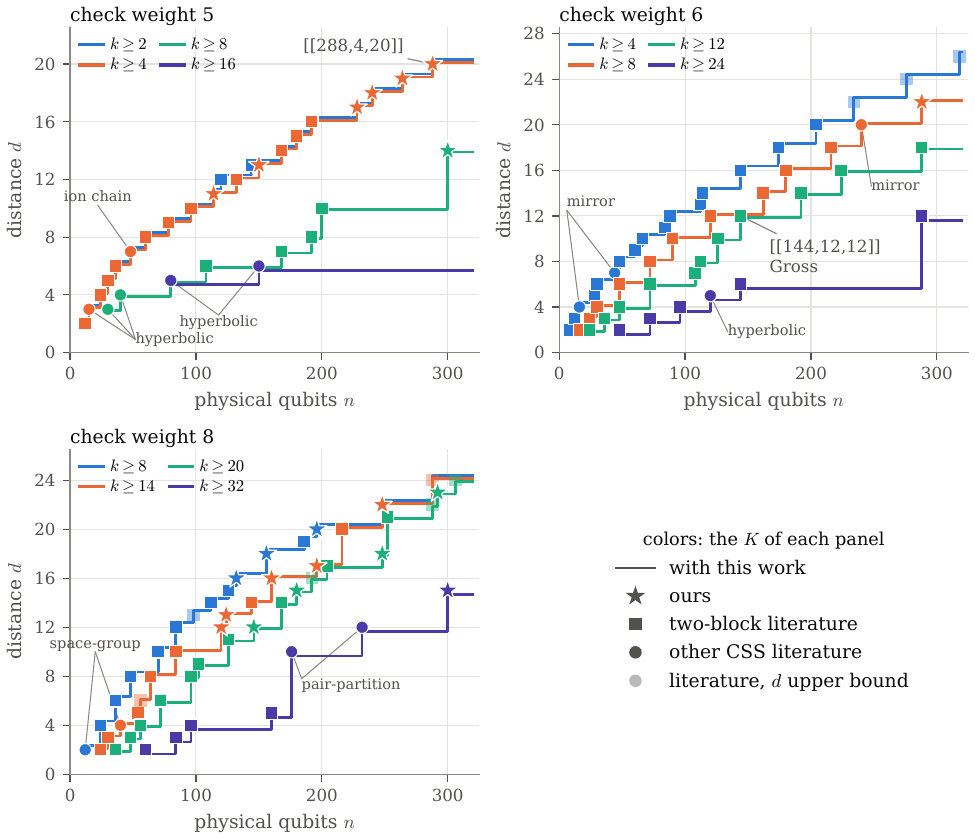}
\caption{Check weights five, six and eight: the best distance with at most $n$
physical qubits and at least $K$ logical qubits (CSS, $n\le320$), with this work;
conventions as in Fig.~\ref{fig:dvn}.  The values
of $K$ differ by panel and are given in each.}
\label{fig:dvn-app}\label{fig:dvn-w5}\label{fig:dvn-w6}\label{fig:dvn-w8}
\end{figure}

\subsection{Exhaustive landscapes}
\label{app:landscapes}

\begin{figure}[H]
\centering
\includegraphics[width=\linewidth]{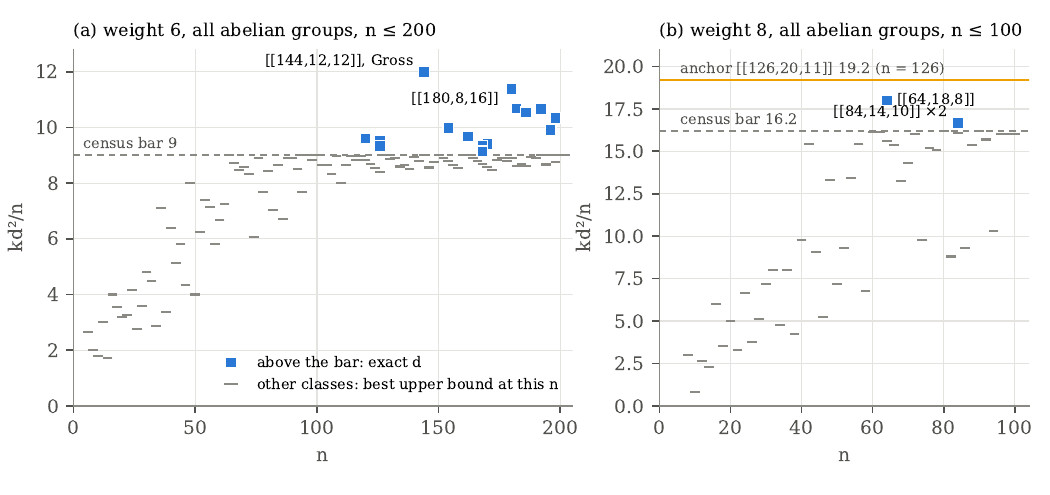}
\caption{Exhaustive landscapes over all abelian groups.  (a) Weight six,
$n\le200$; (b) weight eight, $n\le100$.  Squares are codes above the census bar
with exact distance; dashes give, at each $n$, the best upper bound over all
other classes, each witnessed by an explicit logical; the orange line in (b) is the anchor of the band $n\le130$, $[[126,20,11]]$ at $19.2$ (Table~\ref{tab:anchors}).  The weight-six census continues to $n\le300$ and over non-abelian groups (Sec.~\ref{sec:exhausted}); the panels show its abelian part at $n\le200$.}
\label{fig:landscapes}
\end{figure}

\clearpage

\clearpage
\subsection{Other figures of merit}
\label{app:fairness}

Our searches optimized $kd^2/n$, which weights distance heavily.  These two figures repeat the comparison of Fig.~\ref{fig:headline} under measures that weight the rate more: $kd/n$, and the overhead $n/k$ needed to reach a given distance.

\begin{figure}[H]
\centering
\includegraphics[width=0.95\linewidth]{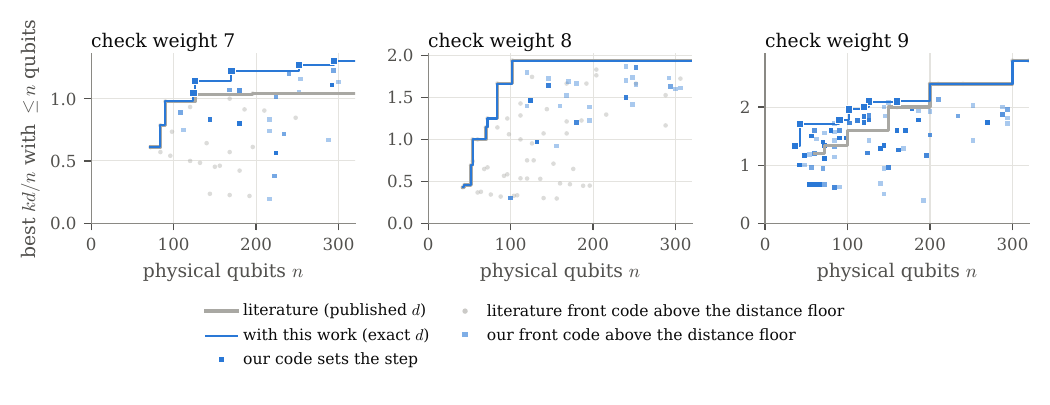}
\caption{As Fig.~\ref{fig:headline}, with $kd/n$ in place of $kd^2/n$.}
\label{fig:headline-kdn}
\end{figure}

\begin{figure}[H]
\centering
\includegraphics[width=0.95\linewidth]{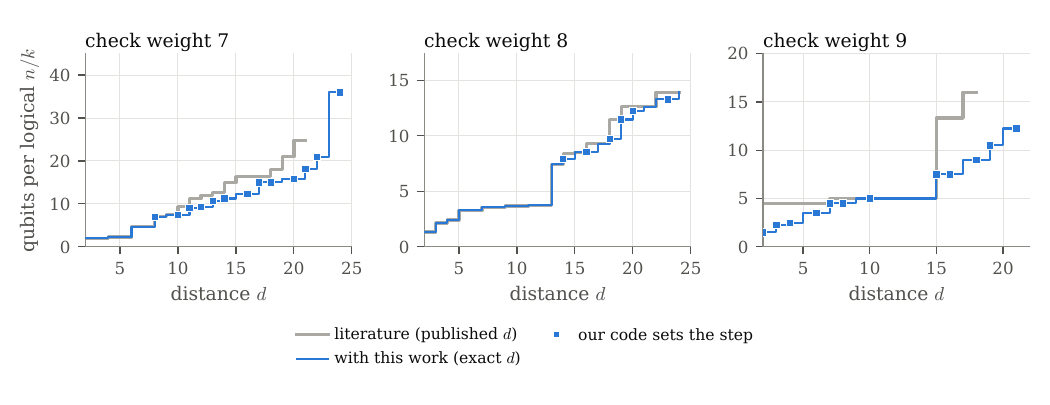}
\caption{Qubits per logical qubit $n/k$ against distance: for each $d$, the
least $n/k$ over front codes with distance at least $d$ ($n\le320$), for the
literature alone (gray) and with this work (blue).  This view discards $n$ itself.}
\label{fig:overhead}
\end{figure}

\clearpage
\subsection{The figure of merit $kd^2/n$ by check weight}
\label{app:kd2n}

$kd^2/n$ against $n$ at each check weight: the best codes at each $n$, and the best code of each family in bands of $40$ qubits.  These are companion views of the fronts in Appendix~\ref{app:pareto}.

\begin{figure}[H]
\centering
\includegraphics[width=\linewidth]{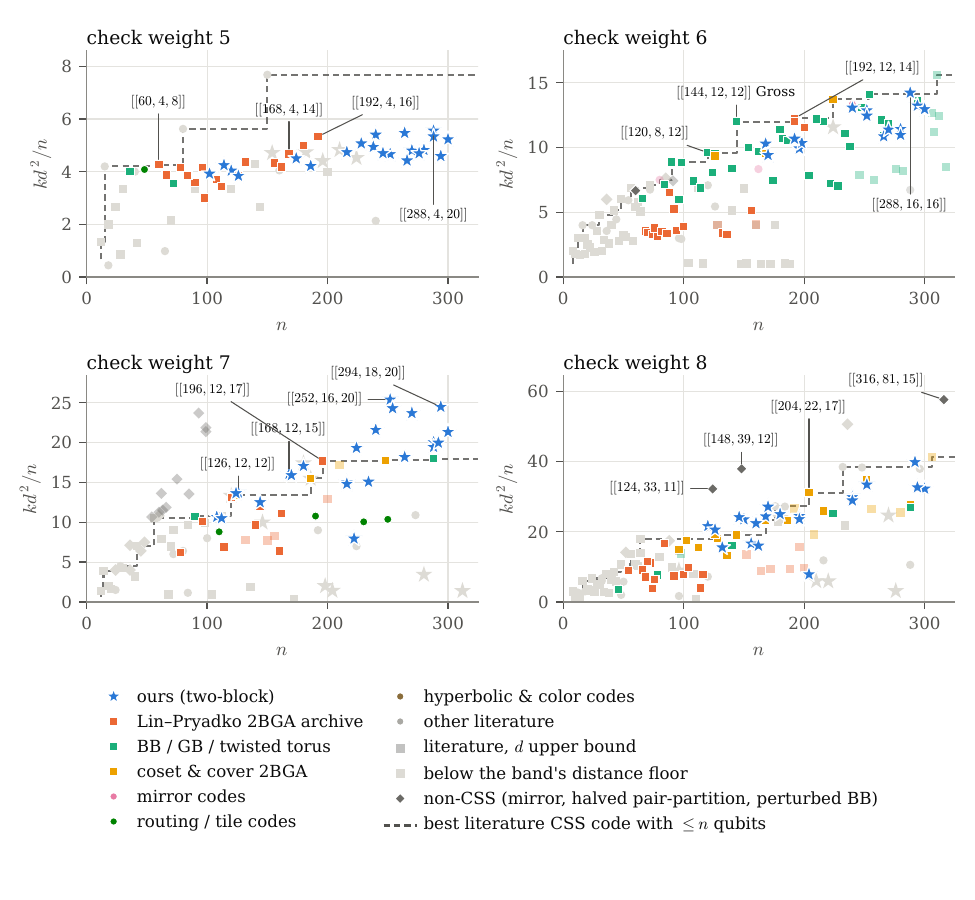}
\caption{$kd^2/n$ against $n$ at check weights five to eight,
codes of exact check weight with $n\le320$.  At each $n$ the figure shows the
code with the highest $kd^2/n$ and up to two others within $5\%$ of it, and
in each band of $40$ qubits the best code of each class (ours, and each
literature family).  Our codes are stars, and only codes above the bar of their census are included (App.~\ref{app:database}); literature codes with an
upper-bound distance are translucent; non-CSS codes are gray diamonds; the best codes are labeled with their $[[n,k,d]]$; dashed: the best
literature CSS code with at most $n$ qubits.}
\label{fig:kd2n-w56}\label{fig:kd2n-w78}
\end{figure}

\begin{figure}[H]
\centering
\includegraphics[width=0.85\linewidth]{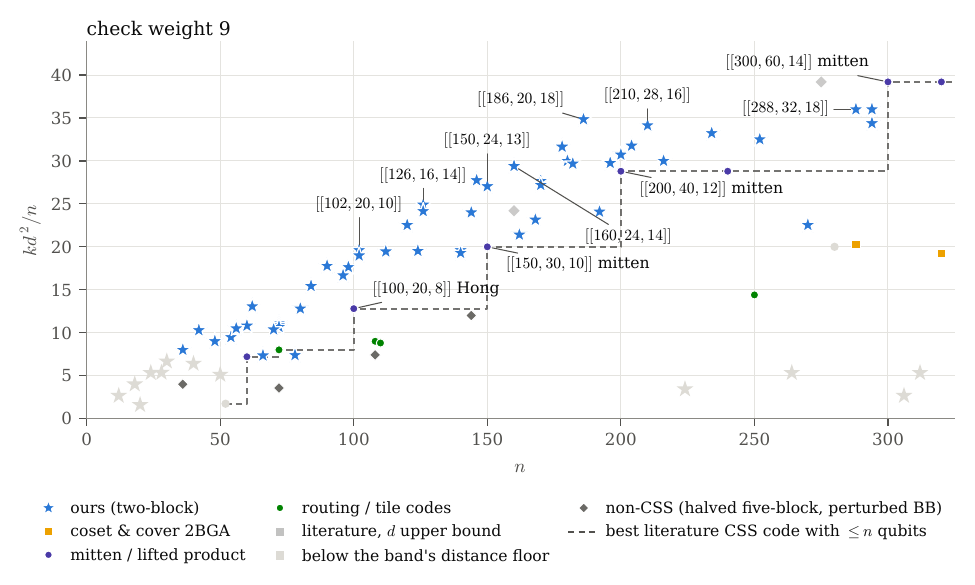}
\caption{$kd^2/n$ against $n$ at check weight nine, as in
Figure~\ref{fig:kd2n-w56}.}
\label{fig:kd2n-w9}
\end{figure}

\end{document}